\documentclass[12pt]{article}
\usepackage{amsthm}
\usepackage{amssymb,amsfonts, bm}
\usepackage{multirow,array}
\usepackage{graphicx,color,rotating}
\usepackage{psfrag}
\usepackage{epstopdf}
\usepackage{subfigure}

\expandafter\let\csname equation*\endcsname\relax
\expandafter\let\csname endequation*\endcsname\relax
\usepackage{amsmath}

\newtheorem{thm}{Theorem}[section]
\newtheorem{lem}{Lemma}[section]

\DeclareMathOperator*{\argmax}{arg\,max}
\def\Bmp#1{ \begin{minipage}{#1} }
\def\Bmpc#1{ \begin{minipage}[c]{#1} }
\def\Bmpt#1{ \begin{minipage}[t]{#1} }
\def\Bmpb#1{ \begin{minipage}[b]{#1} }
\def\Emp{ \end{minipage} }

\def\I{{\mathcal{I}}}

\def\E{{\mathcal{E}}}

\def\O{\mbox{\textit{O}}}

\def\R{{\mathcal{R}}}

\def\K{{\mathcal{K}}}

\def\tf0{\tilde{\varphi}_{0}}

\def\RR{{\mathbb{R}}}

\def\NN{{\mathbb{N}}}
\def\SS{{\mathbb{S}}}

\def\x{{\bf x}}

\def\c{{\bf c}}

\def\u{{\bf u}}
\def\0{{\bf 0}}

\def\bnabla{\boldsymbol{\nabla}}

\newcommand{\tsigma}{\widetilde{\sigma}}
\newcommand{\tgamma}{\widetilde{\gamma}}

\newcommand{\tuEbar}{\widetilde{u}_{\bar{\E}}}

\newcommand{\tuEabar}{\widetilde{u}_{\bar{\E}_{\alpha}}}

\newcommand{\tuvecEbar}{\widetilde{\mathbf{u}}_{\bar{\E}}}

\newcommand{\SE}{\mathcal {S}_{\E}}
\newcommand{\Ea}{\E_{\alpha}}

\newcommand{\PPS}{\mathbb{P}\mathcal {S}}

\newcommand{\overlim}[1]{{\buildrel{#1}\over\longrightarrow\;}}

\makeatletter

\newcommand{\Rmnum}[1]{\expandafter\@slowromancap\romannumeral #1@}
\makeatother
\usepackage{authblk}

\title{Maximum Rate of Growth of Enstrophy in Solutions of the Fractional Burgers Equation}
\author{Dongfang Yun and Bartosz Protas\footnote{Corresponding author. Email address: bprotas@mcmaster.ca}}
\affil{Department of Mathematics and Statistics, McMaster University \\ Hamilton, Ontario, L8S 4K1, Canada}
\date{}

\begin{document}
\maketitle

\begin{abstract}
  This investigation is a part of a research program aiming to
  characterize the extreme behavior possible in hydrodynamic models by
  analyzing the maximum growth of certain fundamental quantities. We
  consider here the rate of growth of the classical and fractional
  enstrophy in the fractional Burgers equation in the subcritical and
  supercritical regimes. Since solutions to this equation exhibit,
  respectively, globally well-posed behavior and finite-time blow-up
  in these two regimes, this makes it a useful model to study the
  maximum instantaneous growth of enstrophy possible in these two
  distinct situations.  First, we obtain estimates on the rates of
  growth and then show that these estimates are sharp up to numerical
  prefactors. This is done by numerically solving suitably defined
  constrained maximization problems and then demonstrating that for
  different values of the fractional dissipation exponent the obtained
  maximizers saturate the upper bounds in the estimates as the
  enstrophy increases. We conclude that the power-law dependence of
  the enstrophy rate of growth on the fractional dissipation exponent
  has the same global form in the subcritical, critical and parts of
  the supercritical regime. This indicates that the maximum enstrophy
  rate of growth changes smoothly as global well-posedness is lost
  when the fractional dissipation exponent attains supercritical
  values.  In addition, nontrivial behavior is revealed for the
  maximum rate of growth of the fractional enstrophy obtained for
  small values of the fractional dissipation exponents.  We also
  characterize the structure of the maximizers in different cases.
\end{abstract}

{\bf{Keywords:}} 
Fractional Burgers equation; extreme behavior; enstrophy growth; numerical optimization; gradient methods

{\bf{AMS subject classifications:}}
35B45, 35Q35, 65K10

\pagestyle{myheadings}
\thispagestyle{plain}


\section{Introduction}
\label{sec:intro}

One of the key questions studied in the mathematical analysis of
evolutionary partial differential equations (PDEs) is the existence of
solutions, both locally and globally in time. The motivation is that,
in order to justify the application of different PDEs as models of
natural phenomena, these equations must be guaranteed to possess
meaningful solutions for physically relevant data. In addition,
characterization of the extreme behavior which can be exhibited by the
solutions of different PDEs is also relevant for our understanding of
the worst-case scenarios which can be realized in the actual physical
systems these PDEs describe. These two types of questions can be
investigated by studying the time evolution of suitable Sobolev norms
of the solutions. In particular, should a given Sobolev norm of the
solution become unbounded at a certain time due to a spontaneous
formation of a singularity, this will signal that the solution is no
longer defined in that Sobolev space; this loss of regularity is
referred to as ``blow-up''.
\medskip

An example of an evolutionary PDE model with widespread applications
whose global-in-time existence remains an open problem is the
three-dimensional (3D) Navier-Stokes system describing the motion of
viscous incompressible fluids.  Questions of existence of solutions to
this system are usually studied for problems defined on unbounded or
{periodic domains} $\Omega$, i.e., $\Omega = \RR^d$ or $\Omega = \SS^d$,
where $d=2,3$. Unlike the two-dimensional (2D) problem where smooth
solutions are known to exist globally in time \cite{kl04}, in 3D
existence of such solutions has been established for short times only
\cite{d09}. Establishing global existence of smooth solutions in 3D is
one of the key open questions in mathematical fluid mechanics and, in
fact, its importance for mathematics in general has been recognized by
the Clay Mathematics Institute as one of its ``millennium problems''
\cite{f00}. Suitable weak solutions were shown to exist in 3D for
arbitrarily long times \cite{l34}, however, such solutions may not be
regular in addition to being nonunique. Similar questions also remain
open for the 3D Euler equation \cite{gbk08}.  While many angles of
attack on this problem have been pursued in mathematical analysis, one
research direction which has received a lot of attention focuses on
the evolution of the {\em enstrophy} $\E(\u)$ which for an
incompressible velocity field $\u(t,\cdot) \; : \; \Omega \rightarrow
\RR^d$ at a given time $t$ is defined as $\E(\u(t)) := (1/2)
\int_{\Omega} | \bnabla\times\u(t,\x)|^2 \, d\Omega = (1/2) \|
\bnabla\u(t,\cdot)\|_{L^2(\Omega)}^2$, where ``$:=$'' means ``equals
to by definition'', i.e., it is proportional to the square of the
$L^2$ norm of the vorticity $\bnabla\times\u$. The reason why this
quantity is interesting in the context of the 3D Navier-Stokes
equation is due to a conditional regularity result proved by Foias and
Temam \cite{ft89} who showed that the solution remains smooth (i.e.,
stays in a suitable Gevrey regularity class) as long as the enstrophy
remains bounded, i.e., for all $t$ such that $\E(\u(t)) < \infty$. In
other words, a loss of regularity must be manifested by the enstrophy
becoming infinite.  While there exist many different conditional
regularity results, this one is particularly useful from the
computational point of view as it involves an easy to evaluate
  quadratic quantity. Analogous conditional regularity results,
although involving other norms of vorticity, were also derived for the
3D Euler equation (e.g., the celebrated Beale-Kato-Majda (BKM)
criterion \cite{bkm84}).
\medskip

In order to assess whether or not the enstrophy can blow up in finite
time one needs to study its instantaneous rate of growth $d\E/dt$
which can be estimated as \cite{d09}
\begin{equation}
\frac{d\E}{dt} < C \, \E^3,
\label{eq:dEdt3D}
\end{equation}
for some $C>0$ (hereafter $C$ will denote a generic positive constant
whose actual value may vary between different estimates). It was shown
in \cite{l06,ld08}, see also \cite{ap16}, that this estimate is in
fact sharp, in the sense that, for each {given enstrophy
  $\bar{\E}>0$} there exists an incompressible velocity field
{$\tuvecEbar$} with {$\E(\tuvecEbar) =\bar{\E}$}, such that
{$d\E(\tuvecEbar) / dt \sim \bar{\E}^3$ as $\bar{\E}\rightarrow
  +\infty$}. The fields {$\tuvecEbar$} were found by numerically
solving a family of variational maximization problems for different
values of {$\bar{\E}$} (details of this approach will be
discussed further below). However, the corresponding {\em finite-time}
estimate obtained by integrating \eqref{eq:dEdt3D} in time takes the
form
\begin{equation}
{\E(\u(t)) \leq \frac{\bar{\E}}{\sqrt{1 - \frac{C}{2}\,\bar{\E}^2\, t}}}, \quad t \ge 0,
\label{eq:Et3D}
\end{equation}
and it is clear that based on this estimate alone it is not possible
to ensure a prior boundedness of enstrophy in time. Thus, the question
of finite-time blow-up may be recast in terms of whether or not it is
possible to find initial data $\u_0$ such that the corresponding flow
evolution saturates the right-hand side of \eqref{eq:Et3D}. We emphasize
that for this to happen the rate of growth of enstrophy given in
\eqref{eq:dEdt3D} would need to be sustained over a {\em finite} window
of time, rather than just instantaneously, a behavior which has not
been observed so far \cite{ap16} (in fact, a singularity may arise in
finite time even with enstrophy growth occurring at a slower sustained
rate of $d\E / dt \sim \E^{\gamma}$ where $\gamma > 2$ \cite{ap16}).
\medskip

The question of the maximum enstrophy growth has been tackled in
computational studies, usually using initial data chosen in an ad-hoc
manner, producing no evidence of a finite-time blow-up in the 3D
Navier-Stokes system \cite{bp94,p01,oc08,opc12,dggkpv13}. However, for
the 3D Euler system the situation is different and the latest
computations reported in \cite{lh14a,lh14b} indicate the possibility
of a finite-time blow-up. A new direction in the computational studies
of extreme behavior in fluid flow models relies on the use of
variational optimization approaches to systematically search for the
most singular initial data. This research program, initiated in
\cite{l06,ld08}, aims to probe the sharpness, or realizability, of
certain fundamental estimates analogous to \eqref{eq:dEdt3D} and
\eqref{eq:Et3D} and defined for various hydrodynamic models such as
the one-dimensional (1D) viscous Burgers equation and the 2D/3D
Navier-Stokes system. Since the 1D viscous Burgers equation and the 2D
Navier-Stokes system are both known to be globally well-posed in the
classical sense \cite{kl04}, there is no question about the
finite-time blow-up in these problems. However, the relevant estimates
for the growth of certain Sobolev norms, both instantaneously and in
finite time, are obtained using very similar functional-analysis tools
as estimates \eqref{eq:dEdt3D} and \eqref{eq:Et3D}, hence the question
of their sharpness is quite pertinent as it may offer valuable
insights about estimates \eqref{eq:dEdt3D}--\eqref{eq:Et3D}.  An
estimate such as \eqref{eq:dEdt3D} (or \eqref{eq:Et3D}) is declared
``sharp'', if for increasing values of $\E$ (or {$\bar{\E}$})
the quantity on the left-hand side (LHS) exhibits the same power-law
dependence on $\E$ (or {$\bar{\E}$}) as the upper bound on the
right-hand side (RHS). What makes the fractional Burgers equation
interesting in this context is that it is a simple model which
exhibits either a globally well-posed behavior or finite-time blow-up
depending on the value of the fractional dissipation exponent.
Therefore, it offers a convenient testbed for studying properties of
estimates applicable in these distinct regimes.

\medskip

Assuming the domain $\I := (0,1)$ to be periodic, we write the 1D
fractional Burgers equation as
\begin{subequations}
\label{eq:fburgers}
\begin{alignat}{2}
& u_t+ uu_x + \nu\,{(-\Delta)}^\alpha u = 0, && t>0,\ x\in \I,   \label{eq:fburgers_a}\\
& \text{Periodic Boundary Conditions}, \quad && t>0, \label{eq:fburgers_b}\\
& u(0, x) = u_0(x), && x\in \I \label{eq:fburgers_c}
\end{alignat}
\end{subequations}
for some viscosity coefficient $\nu > 0$ and with ${(-\Delta)}^\alpha$
denoting the fractional Laplacian which for sufficiently regular
functions $v$ defined on a periodic domain and $\alpha \ge 0$ is
defined via the relation
\begin{equation}
\mathcal{F}\left[{(-\Delta)}^\alpha v\right](\xi) 
= \left({2\pi|\xi|}\right)^{2\alpha} \mathcal{F}[v](\xi),
\label{eq:F}
\end{equation}
where $\mathcal{F}[\cdot](\xi)$ represents the Fourier transform. We
{remark} that in the special cases of $\alpha = 0 $ and $\alpha =
1$ the fractional Laplacian reduces to, respectively, the identity
operator and the (negative) classical Laplacian. {Is is
  interesting to note that in addition to $\int_0^1 u\, dx$ the
  quantity $\int_0^1 {(-\Delta)}^{(1-\alpha)}u\,dx$ is also conserved
  during evolution governed by system \eqref{eq:fburgers}.
  Furthermore, in the periodic setting, $\int_0^1 u\, dx = 0$ also
  implies that $\int_0^1 {(-\Delta)}^{(1-\alpha)}u\,dx = 0$.}  We now
define the associated
\begin{subequations}
\label{eq:EEa}
\begin{alignat}{2}
&\text{(classical) enstrophy:}& \qquad \E(u) &:= \frac{1}{2}\int_0^1 {|u_x|}^2\,dx, \quad\qquad \text{and} 
\label{eq:E} \\
&\text{fractional enstrophy:}& \qquad \E_\alpha(u) &:= \frac{1}{2}\int_0^1 {\left| {(-\Delta)}^\frac{\alpha}{2}u \right|}^2\,dx.
\label{eq:Ea}
\end{alignat}
\end{subequations}
It should be noted that $\E(u)$ and $\E_{\alpha}(u)$ become equivalent
when $\alpha=1$, which is a consequence of the relation
\begin{equation}
\int_0^1 {(-\Delta)}^\frac{1}{2}u \cdot {(-\Delta)}^\frac{1}{2}u\,dx= \int_0^1 u \cdot {(-\Delta)}^1u\,dx = \int_0^1 u_x \cdot u_x\,dx
\label{eq:E1}
\end{equation}
following from the properties of the fractional Laplacian \eqref{eq:F}.
\medskip

Evidently, when $\alpha = 1$, system \eqref{eq:fburgers} reduces
  to the classical Burgers equation for which a number of relevant
  results have already been obtained in the seminal studies
  \cite{l06,ld08}. It was shown in these investigations that the rate
  of growth of the classical enstrophy $\E(u)$ is subject to the
following bound
\begin{equation}
\frac{d \E}{dt} \leq \frac{3}{2}{\left( \frac{1}{\pi^2\nu} \right)}^\frac{1}{3} \E^\frac{5}{3}.
\label{eq:dEdt1D}
\end{equation}
By considering a family of variational optimization problems
\begin{equation}
\begin{aligned}
& \max\limits_{{ u\in H^2(\I) }} \frac{d\E(u)}{dt} \\
& \text{subject to}\ \E(u)={\bar{\E}}
\end{aligned},
\label{eq:maxdEdt}
\end{equation}
parameterized by {$\bar{\E}>0$}, in which {$H^s(\I)$, $s\in\RR$,} is the
Sobolev space of functions defined on the periodic interval $\I$ {and
  possessing square-integrable derivatives of up to (fractional) order
  $s$} \cite{af05}, it was then demonstrated that estimate
\eqref{eq:dEdt1D} is in fact sharp. Remarkably, the authors in
\cite{l06,ld08} were able to solve problem \eqref{eq:maxdEdt}
analytically in closed form (although the structure of the maximizers
is quite complicated and involves special functions). When using
the instantaneously optimal initial states {$\tuEbar$} obtained for
different values of {$\bar{\E}$} as the initial data $u_0$ for Burgers
system \eqref{eq:fburgers} with $\alpha = 1$, the maximum
enstrophy growth $\max_{t\ge 0} \E(u(t)) -{\bar{\E}}$ achieved during the
resulting flow evolution was proportional to ${\bar{\E}}$. The question
about the {\em maximum} enstrophy growth achievable in finite time was
investigated in \cite{ap11a} where the following estimate was obtained
\begin{equation}
{\max_{{t>0}} \E(u(t))\leq {\left[ \bar{\E}^\frac{1}{3}+ 
\frac{1}{16} {\left( \frac{1}{\pi^2\nu} \right)}^\frac{4}{3}\bar{\E} \right]}^3 
\ \overlim{\bar{\E} \rightarrow \infty} \ \frac{1}{4096} {\left( \frac{1}{\pi^2\nu} \right)}^4 \bar{\E}^3.}
\label{eq:Et1D}
\end{equation}
To probe its sharpness, a family of variational optimization problems
\begin{equation}
\begin{aligned}
& \max\limits_{\phi\in H^1(\I)} \left[{\E(u(T))} - {\bar{\E}}\right] \\
& \text{subject to}\ \E(\phi)={\bar{\E}}
\end{aligned},
\label{eq:maxdEt}
\end{equation}
where $\phi$ is the initial data for the Burgers system, i.e., $u_0 =
\phi$ in \eqref{eq:fburgers_c}, was solved numerically for a broad
range of initial enstrophy values {$\bar{\E}$} and lengths $T$ of
the time window.  It was found that the maximum finite-time enstrophy
growth {$\max_{T>0} \max_{\phi} \E(u(T)) - \bar{\E}$} scales as
{$\bar{\E}^{3/2}$ and these observations were later} rigorously
justified by Pelinovsky in \cite{p12} using the Laplace method
combined with the Cole-Hopf transformation. In a related study
\cite{p12b}, a dynamical-systems approach was used to reveal a
self-similar structure of the maximizing solution in the limit of
large enstrophy. This asymptotic solution was shown to have the form
of a viscous shock wave superimposed on a linear rarefaction wave.  In
that study similar maximizing solutions were also constructed on the
entire real line. {The observed dependence of the maximum
  finite-time growth of enstrophy $\max_{T>0} \max_{\phi} \E(u(T)) -
  \bar{\E}$ on $\bar{\E}$ is thus} significantly weaker than the
maximum growth stipulated by estimate \eqref{eq:Et1D} in the limit
${\bar{\E}} \rightarrow \infty$, demonstrating that this estimate
is {\em not} sharp and may be improved (which remains an open
problem).  The question how the maximum finite-time growth of
enstrophy in the Burgers system may be affected by additive stochastic
noise was addressed in \cite{pp15}.  Using an approach based on
Monte-Carlo sampling, it was shown that the stochastic excitation does
not decrease the extreme enstrophy growth, defined in a suitable
probabilistic setting, as compared to what is observed in the
deterministic case.  \medskip

The question of the extreme behavior in 2D Navier-Stokes flows was
addressed in \cite{ap13a,ap13b}. Since on 2D unbounded and
doubly-periodic domains the enstrophy may not increase, the relevant
quantity in this setting is the {\em palinstrophy} defined as the
$L^2$ norm of the vorticity gradient. In \cite{ap13a} it was shown,
again by numerically solving suitably defined variational maximization
problems, that the available bounds on the rate of growth of
palinstrophy are sharp and that, somewhat surprisingly, the
corresponding maximizing vorticity fields give rise to flow evolutions
which also saturate the estimates for the palinstrophy growth
in finite time. Thus, paradoxically, as far as the sharpness of the
finite-time estimates is concerned, the situation in 2D is more
satisfactory than in 1D.
\medskip

The goal of the present investigation is to advance the
research program outlined above by considering the extreme behavior
possible in the fractional Burgers system \eqref{eq:fburgers}
  when $\alpha \in [0,1]$. The reason why this problem is
interesting from the point of view of this research program is
that, as discussed in \cite{adv07, kns08, ddl09}, the fractional
  Burgers system is globally well-posed when $\alpha \ge 1/2$ and
exhibits a finite-time blow-up in the supercritical regime when
$\alpha < 1/2$ (it was initially demonstrated in \cite{kns08} that the
blow-up occurs in the Sobolev space $H^s(\I)$, $s>3/2-2\alpha$, and
this results was later refined in \cite{ddl09} where it was shown that
under certain conditions on the initial data the blow-up occurs in
$W^{1,\infty}(\I)$ for all $\alpha < 1/2$; eventual
regularization of solutions after blow-up was discussed in
\cite{ccs10}).  Thus, the behavior changes fundamentally when the
fractional dissipation exponent $\alpha$ is reduced below $1/2$ (this
aspect was also illustrated in \cite{a15}).  Furthermore, there is
also a certain similarity with the 3D Navier-Stokes system which is
known to be globally well-posed in the classical sense in the
presence of fractional dissipation with $\alpha \ge 5/4$ \cite{kp12}.
Our specific objectives are therefore twofold:
\begin{itemize}
\item first, we will obtain upper bounds on the rate of growth of
  enstrophy generalizing estimate \eqref{eq:dEdt1D} to the fractional
  dissipation case with $\alpha \in [0,1]$ in \eqref{eq:fburgers}; this
  will be done separately for both the classical and fractional
    enstrophy, cf.~\eqref{eq:E} and \eqref{eq:Ea}, and

  \item secondly, we will probe the sharpness, in the sense
      defined above, of these new estimates by numerically solving
    corresponding variational maximization problems; in this latter
    step we will also provide insights about the structure of the
    optimal states saturating different bounds.
\end{itemize}
Based on this, we will conclude how the maximum instantaneous growth
of enstrophy changes between the regimes with globally well-posed
behavior and with finite-time blow-up.

The structure of the paper is as follows: in the next section we
derive upper bounds on the rate of growth of the classical and
fractional enstrophy as functions of the fractional dissipation
exponent $\alpha$, then in Section \ref{sec:method} we provide details
of the computational approach designed to probe the sharpness of these
bounds, whereas in Section \ref{sec:results} we present numerical
results obtained for the two cases; discussion and final conclusions
are deferred to Section \ref{sec:final}.
\medskip

\section{Upper Bounds on the Rate of Growth of the Classical and Fractional Enstrophy}
\label{sec:bounds}

In this section we first use system \eqref{eq:fburgers} to obtain
expressions for the rate of growth of the classical and fractional
enstrophy \eqref{eq:E} and \eqref{eq:Ea} in terms of the state variable
$u$. Next, we derive estimates on these rates of growth in terms of
the instantaneous enstrophy values $\E$ and $\Ea$. These
results are stated in the form of theorems in two subsections below.
\medskip

In order to obtain an expression for the growth rate $d\E/dt$ of the
classical enstrophy \eqref{eq:E}, we multiply the fractional Burgers
equation \eqref{eq:fburgers_a} by $(-u_{xx})$, integrate the resulting
relation over the periodic interval $\I$ and then perform integration
by parts to obtain
\begin{equation}
\begin{aligned}
\frac{d\mathcal{E}}{dt}& =\frac{1}{2}\frac{d}{dt}\int_0^1 {|u_x|}^2\,dx \\
& = \int_0^1 u_{xx} u u_x\,dx + \nu\int_0^1 u_{xx} {(-\Delta)}^\alpha u \,dx \\
&= -\frac{1}{2}\int_0^1 u_x^3\,dx- \nu\int_0^1 {\left[ {(-\Delta)}^{\frac{\alpha}{2}} u_x\right]}^2 \,dx \\
& =: \mathcal {R}_{\mathcal {E}}(u)\,.
\end{aligned}
\label{eq:R}
\end{equation}

Analogously, in order to obtain an expression for the growth rate
$d\E_{\alpha} / dt$ of the fractional enstrophy \eqref{eq:Ea}, we
multiply the fractional Burgers equation \eqref{eq:fburgers_a} by
${(-\Delta)}^\alpha u$ and after performing similar steps as above we
obtain
\begin{equation}
\begin{aligned}
\frac{d\mathcal{E}_\alpha}{dt}& = \frac{1}{2}\frac{d}{dt}\int_0^1 {\left| {(-\Delta)}^\frac{\alpha}{2}u \right|}^2\,dx \\
& = -\int_0^1 {(-\Delta)}^\alpha u\cdot u u_x \,dx - \nu\int_0^1 {\left[ {(-\Delta)}^\alpha u \right]}^2 \,dx \\
& =: \mathcal {R}_{\mathcal {E}_\alpha}(u)\,.
\end{aligned}
\label{eq:Ra}
\end{equation}

\subsection{Estimate of $d\E/dt$}
\label{sec:estimR}

We begin by estimating the cubic integral in \eqref{eq:R} which is
addressed by
\begin{lem}\label{R3}
  For $\alpha\in (\frac{1}{6}, 1]$ and a sufficiently smooth function
  $u$ defined on the periodic interval $\I$, we have
  \begin{equation}
    {\| u \|}_{L^3(\mathcal{I})}^3 \leq C_1 \, {\| u \|}_{L^2(\mathcal{I})}^{3-\frac{1}{2\alpha}}\ {\| {(-\Delta)}^\frac{\alpha}{2} u \|}_{L^2(\mathcal{I})}^\frac{1}{2\alpha}
    \label{eq:R3}
  \end{equation}
  with some constant $C_1$ depending on $\alpha$.
\end{lem}

\begin{proof}
In \cite{l06,ld08} the following estimate was established
\begin{equation}
{\| u \|}_{L^3(\mathcal{I})}^3 \leq 
\frac{2}{\sqrt{\pi}} {\| u \|}_{L^2(\mathcal{I})}^\frac{5}{2}\,{\| u_x \|}_{L^2(\mathcal{I})}^\frac{1}{2},
\end{equation}
from which it follows, upon noting that ${\| u_x \|}_{L^2(\mathcal{I})}={\|
  {(-\Delta)}^\frac{1}{2} u \|}_{L^2(\mathcal{I})}$, that inequality
\eqref{eq:R3} holds when $\alpha = 1$.
\medskip

Since $u$ is defined on the periodic interval, it has a discrete
Fourier series representation
\begin{equation}
u(x)=\sum_k \widehat{u}_k\,e^{2\pi i kx}\,,
\label{eq:uhat}
\end{equation}
where $k \in \NN$ is the wavenumber and $\widehat{u}_k$ the
corresponding Fourier coefficient. {In the case when
  $\alpha>1/2$, }we split the sum \eqref{eq:uhat} at $k=\kappa$, so
that
\begin{equation}
\begin{aligned}
|u(x)|&=\left|\sum_k \widehat{u}_k e^{2\pi i k x}\right|\\
&\leq \sum_{|k|\leq\kappa} |\widehat{u}_k| + \sum_{|k|>\kappa} |\widehat{u}_k|\\
&=\sum_{|k|\leq\kappa} 1\cdot|\widehat{u}_k|+ \sum_{|k|>\kappa} {(2\pi |k|)}^{-\alpha} {(2\pi |k|)}^{\alpha} |\widehat{u}_k|\\
&\leq \sqrt{\sum_{|k|\leq\kappa}1^2 \sum_{|k|\leq\kappa}{|\widehat{u}_k|}^2} + \sqrt{\sum_{|k|>\kappa} {(2\pi |k|)}^{-2\alpha} \sum_{|k|>\kappa} {(2\pi |k|)}^{2\alpha}{|\widehat{u}_k|}^2}\\
& \leq {(2\kappa)}^\frac{1}{2} {\| u \|}_{L^2} + {\left(\frac{2}{2\alpha-1}\right)}^\frac{1}{2} {(2\pi)}^{-\alpha} {\kappa}^\frac{-2\alpha+1}{2} {\| {(-\Delta)}^\frac{\alpha}{2}u \|}_{L^2}\,,
\end{aligned}
\label{eq:usplit}
\end{equation}
where $\kappa$ is a parameter to be determined, and the
Plancherel theorem, Cauchy-Schwarz inequality as well as the
inequality
\begin{equation*}
\sum_{|k|>\kappa} {(2\pi|k|)}^{-2\alpha}\leq 2\,{(2\pi)}^{-2\alpha}\int_\kappa^\infty {x}^{-2\alpha}\,dx 
= \frac{2\,{(2\pi)}^{-2\alpha}}{2\alpha-1} {\kappa}^{-2\alpha+1},\quad \alpha > \frac{1}{2} 
\end{equation*}
were used to obtain the last inequality in
  \eqref{eq:usplit}. The upper bound in \eqref{eq:usplit} is minimized
by choosing $\kappa={ {(2\alpha-1)}^\frac{1}{2\alpha} {\|
    {(-\Delta)}^\frac{\alpha}{2}u \|}_{L^2}^\frac{1}{\alpha} } /
\left( 2\pi{\| u \|}_{L^2}^\frac{1}{\alpha} \right)$ which produces
\begin{equation*}
{\| u \|}_{L^\infty(\mathcal{I})}\leq 
C\,{\| u \|}_{L^2(\mathcal{I})}^{1-\frac{1}{2\alpha}}\,{\| {(-\Delta)}^\frac{\alpha}{2}u \|}_{L^2(\mathcal{I})}^\frac{1}{2\alpha},\quad \text{where} \quad C=\frac{2\alpha}{{(2\alpha-1)}^{1-\frac{1}{4\alpha}}\sqrt{\pi}}. 
\end{equation*}
Then, we finally obtain
\begin{equation*}
{\| u \|}_{L^3(\mathcal{I})}^3 \leq {\| u \|}_{L^\infty(\mathcal{I})}\,\int_0^1 u^2\,dx \leq 
C \, {\| u \|}_{L^2(\mathcal{I})}^{3-\frac{1}{2\alpha}}\ {\| {(-\Delta)}^\frac{\alpha}{2}u \|}_{L^2(\mathcal{I})}^\frac{1}{2\alpha},
\end{equation*}
which proves that inequality \eqref{eq:R3} holds for
$1/2<\alpha<1$.  
\medskip

When $\alpha=1/2$, we have for an arbitrary $x_0 \in \I$ 
\begin{equation}
\begin{aligned}
u^2(x_0) & \leq 2\int_0^{x_0}u(x)u'(x) dx \\ 
& \leq 2 \int_0^1 \sum_{k,\,j} 2\pi|k|\,\widehat{u}_k\,\widehat{u}_j\,  e^{2\pi i(k-j) x}\, dx \\
& \leq 4\pi\sum_k |k||\widehat{u}_k|^2 \\
& = 2{\| {(-\Delta)}^\frac{1}{4}{u} \|}_{L^2(\mathcal{I})}^2\,,
\end{aligned}
\label{eq:u2}
\end{equation}
where the last equality is obtained by the Plancherel theorem.  It then follows that the upper bound on $u$ is given by
\begin{equation*}
{\| u \|}_{L^\infty(\mathcal{I})} \leq \sqrt{2} \, {\| {(-\Delta)}^\frac{1}{4}u \|}_{L^2(\mathcal{I})}
\end{equation*}
and we have
\begin{equation*}
{\| u \|}_{L^3(\mathcal{I})}^3 \leq {\| u \|}_{L^\infty(\mathcal{I})}\,\int_0^1 u^2\,dx \leq 
\sqrt{2}\,{\| u \|}_{L^2(\mathcal{I})}^2\,{\| {(-\Delta)}^\frac{\alpha}{2}u \|}_{L^2(\mathcal{I})}\, .
\end{equation*}
\medskip

When $1/6 < \alpha < 1/2$, the fractional Gagliardo-Nirenberg
inequality established in \cite{hyz12} yields
\begin{equation*}
{\| u \|}_{L^3(\mathcal{I})}^3 \leq C\,{\| u\|}_{L^2(\mathcal{I})}^{3-\frac{1}{2\alpha}}\,{\| {(-\Delta)}^\frac{\alpha}{2}u \|}_{L^2(\mathcal{I})}^\frac{1}{2\alpha},
\quad \text{where} \quad C=\frac{1}{\sqrt{2\pi}}{\left[ \frac{\Gamma(\frac{1-\alpha}{2})}{\Gamma(\frac{1+\alpha}{2})} \right]}^\frac{1}{2\alpha}\,. 
\end{equation*}
We note that this fractional Gagliardo-Nirenberg inequality was
originally obtained in \cite{hyz12} for an unbounded domain $\RR^d$
with some $0< d \in \NN$ and here we restrict it to the periodic
interval $\I$. As a result, the constant $C$ may not be optimal.

The lemma is thus proved.
\end{proof}

We remark that when $\I=\mathbb{R}$ the generalized
Gagliardo-Nirenberg inequality from \cite{mrr13, hmow11} can be
applied to deduce the same inequality as in Lemma \ref{R3}, but
without an explicit expression for the prefactor. We are now in the
position to state
\begin{thm}\label{dEdt}
  For $\alpha\in (\frac{1}{4}, 1]$ the rate of growth of enstrophy is
  subject to the bound
\begin{equation}
\begin{aligned}
\frac{d\mathcal {E}}{dt} & \leq \sigma_1 \,\mathcal {E}^{\gamma_1}, 
\quad \text{where} \quad \gamma_1 = \frac{6\alpha-1}{4\alpha-1} \quad \text{and}  \\
\sigma_1 &= \left\{
\begin{alignedat}{2}
&\frac{4\alpha-1}{(2\alpha-1)2^\frac{2\alpha+1}{4\alpha-1}\nu^\frac{1}{4\alpha-1}\pi^\frac{2\alpha}{4\alpha-1}}, && \mathrm{for}\ \frac{1}{2}<\alpha\leq 1\,, \\
&\frac{1}{2\nu}, && \mathrm{for}\ \alpha=\frac{1}{2}\,,\\
&\frac{ 4\alpha-1 }{ 2^\frac{8\alpha+1}{4\alpha-1}\pi^\frac{2\alpha}{4\alpha-1}\alpha^\frac{4\alpha}{4\alpha-1}\nu^\frac{1}{4\alpha-1} } {\left[ \frac{\Gamma(\frac{1-\alpha}{2})}{\Gamma(\frac{1+\alpha}{2})} \right]}^\frac{2}{4\alpha-1}, & \qquad & \mathrm{for}\ \frac{1}{4}<\alpha<\frac{1}{2}\,.
\end{alignedat}
\right.
\end{aligned}
\label{eq:dEdt}
\end{equation}
\end{thm}

\begin{proof}
  Applying inequality \eqref{eq:R3} to estimate the cubic integral
  \eqref{eq:R}, we have
\begin{equation}
\frac{d\mathcal{E}}{dt} \leq \frac{C_1}{2} {\| u_x \|}_{L^2(\mathcal{I})}^{3-\frac{1}{2\alpha}}\ {\| {(-\Delta)}^\frac{\alpha}{2} u_x \|}_{L^2(\mathcal{I})}^\frac{1}{2\alpha} - \nu\,{\| {(-\Delta)}^\frac{\alpha}{2} u_x \|}_{L^2(\mathcal{I})}^2\,,
\label{eq:dEdt2}
\end{equation}
where $C_1$ is defined in \eqref{eq:R3}. Then, Young's inequality is
used to estimate the first term on the RHS of \eqref{eq:dEdt2} such
that the second term is eliminated (we refer the reader to
\cite{l06,ld08} for details of this step which is analogous to the
case with $\alpha = 1$). We note that this last step is valid
  only when $\alpha > 1/4$ and we finally obtain
\begin{equation*}
\frac{d\mathcal{E}}{dt} \leq \frac{(4\alpha-1) }{ {(8\alpha)}^\frac{4\alpha}{4\alpha-1} \nu^\frac{1}{4\alpha-1}}\,C_1^\frac{4\alpha}{4\alpha-1} \, {\| u_x \|}_{L^2(\mathcal{I})}^{\frac{2(6\alpha-1)}{4\alpha-1}},\quad \alpha>1/4
\end{equation*}
which is equivalent to \eqref{eq:dEdt}. The theorem is thus proved.
\end{proof}
{\noindent As regards the range of applicability of estimate
  \eqref{eq:dEdt}, we note that $\lim_{\alpha \rightarrow (1/4)^+}
  \gamma_1(\alpha) = \infty$, so $1/4$ represents a natural lower
  bound on $\alpha$, cf.~Figure \ref{fig:estim}(b).}

\subsection{Estimate of $d\E_{\alpha}/dt$}
\label{sec:estimRa}

As above, we begin by estimating the cubic integral in \eqref{eq:Ra}
which is addressed by
\begin{lem}\label{R3a}
  For $\alpha\in (\frac{3}{4}, 1]$ and a sufficiently smooth function
  $u$ defined on the periodic interval $\I$ we have
\begin{equation}
\left| \int_0^1 {(-\Delta)}^\alpha u \cdot u u_x \, dx \right| \leq 
C_{\alpha} \,{\| {(-\Delta)}^\frac{\alpha}{2} u \|}_{L^2(\mathcal{I})}^{\frac{8\alpha-3}{2\alpha}}\ {\| {(-\Delta)}^\alpha u \|}_{L^2(\mathcal{I})}^\frac{3-2\alpha}{2\alpha} 
\label{eq:R3a}
\end{equation}
with some constant $C_{\alpha}$ depending on $\alpha$.
\end{lem}

\begin{proof}
  Based on the discrete Fourier series representation \eqref{eq:F} of $u$,
  we have
\begin{equation}
\begin{aligned}
|u_x(x)|&=\left|\sum_k (2\pi i k)\widehat{u}_k e^{2\pi ikx}\right|\\
&\leq \sum_{|k|\leq\kappa} 2\pi |k| |\widehat{u}_k| + \sum_{|k|>\kappa} 2\pi |k| |\widehat{u}_k|\\
&=\sum_{|k|\leq\kappa} {(2\pi|k|)}^{1-\alpha} {(2\pi |k|)}^\alpha |\widehat{u}_k| + \sum_{|k|>\kappa} {(2\pi |k|)}^{1-2\alpha}  {(2\pi |k|)}^{2\alpha} |\widehat{u}_k|\\
&\leq \sqrt{\sum_{|k|\leq\kappa} {(2\pi|k|)}^{2-2\alpha}  \sum_{|k|\leq\kappa}  {(2\pi|k|)}^{2\alpha}{|\widehat{u}_k|}^2} + \\
&\ \ \ \ \sqrt{\sum_{|k|>\kappa} {(2\pi |k|)}^{2-4\alpha}  \sum_{|k|>\kappa} {(2\pi |k|)}^{4\alpha}{|\widehat{u}_k|}^2}\\
& \leq \sqrt{2}{(2\pi)}^{1-\alpha} \kappa^{\frac{3}{2}-\alpha} {\| {(-\Delta)}^\frac{\alpha}{2} u \|}_{L^2(\mathcal{I})} + \sqrt{\frac{2}{4\alpha-3}} {(2\pi)}^{1-2\alpha} \kappa^{{\frac{3}{2}-2\alpha}} {\| {(-\Delta)}^\alpha u \|}_{L^2(\mathcal{I})}\,,
\end{aligned}
\label{eq:ux}
\end{equation}
where $\kappa$ is a splitting parameter to be determined, and
the Plancherel theorem, Cauchy-Schwarz inequality and as well as
inequalities
\begin{subequations}
\label{eq:ineq1}
\begin{align}
& \sum_{|k|\leq\kappa} {(2\pi|k|)}^{2-2\alpha} \leq 2 \kappa {(2\pi \kappa)}^{2-2\alpha} = 2{(2\pi)}^{2-2\alpha}\kappa^{3-2\alpha}\,,  \label{eq:ineq1a} \\
& \sum_{|k|>\kappa} {(2\pi |k|)}^{2-4\alpha} \leq 2{(2\pi)}^{2-4\alpha}\int_\kappa^\infty x^{2-4\alpha} dx = \frac{2{(2\pi)}^{2-4\alpha}}{4\alpha-3} \kappa^{3-4\alpha}, \quad \alpha>\frac{3}{4} \label{eq:ineq1b}
\end{align}
\end{subequations}
were applied to obtain the last inequality in \eqref{eq:ux}. The
upper bound in \eqref{eq:ux} is minimized by choosing $\kappa={
  {(4\alpha-3)}^\frac{1}{2\alpha} {\| {(-\Delta)}^\alpha u
    \|}_{L^2(\mathcal{I})}^\frac{1}{\alpha} } / \allowbreak \left(
  2\pi{(3-2\alpha)}^\frac{1}{\alpha}{\| {(-\Delta)}^\frac{\alpha}{2} u
    \|}_{L^2(\mathcal{I})}^\frac{1}{\alpha} \right)$ which yields
\begin{equation}
\begin{aligned}
{ {\| u_x \|}_{L^\infty(\mathcal{I})} }& \leq 
C_{\alpha} \, {\| {(-\Delta)}^\frac{\alpha}{2} u \|}_{L^2(\mathcal{I})}^{\frac{4\alpha-3}{2\alpha}}\ {\| {(-\Delta)}^\alpha u \|}_{L^2(\mathcal{I})}^\frac{3-2\alpha}{2\alpha}, \quad \text{where} \\
C_{\alpha} & = \frac{ 2\alpha }{ \sqrt{\pi}{(4\alpha-3)}^\frac{6\alpha-3}{4\alpha}{(3-2\alpha)}^\frac{3-2\alpha}{2\alpha} }.
\end{aligned}
\label{eq:uLinf}
\end{equation}
We then finally obtain
\begin{align*}
\left| \int_0^1 {(-\Delta)}^\alpha u \cdot u u_x \, dx \right| & \leq {\| u_x \|}_{L^\infty(\mathcal{I})}\ \left| \int_0^1 {\left|{(-\Delta)}^\frac{\alpha}{2} u\right|}^2 dx \right| \\
&\leq C_{\alpha} {\| {(-\Delta)}^\frac{\alpha}{2} u \|}_{L^2(\mathcal{I})}^{\frac{8\alpha-3}{2\alpha}}\ {\| {(-\Delta)}^\alpha u \|}_{L^2(\mathcal{I})}^\frac{3-2\alpha}{2\alpha},
\end{align*}
which proves the lemma.
\end{proof}

We are now in the position to state 
\begin{thm}\label{dEadt}
For $\alpha\in (\frac{3}{4}, 1]$ the rate of growth of fractional enstrophy is subject to the bound
\begin{equation}
\begin{aligned}
\frac{d\E_{\alpha}}{dt} & \leq \sigma_{\alpha} \,\E_{\alpha}^{\gamma_{\alpha}}, 
\quad \text{where} \quad \gamma_{\alpha} = \frac{8\alpha-3}{6\alpha-3} \quad \text{and}  \\
\sigma_{\alpha} &= \frac{ 2^\frac{4\alpha-3}{6\alpha-3}(6\alpha-3) }{ \pi^\frac{2\alpha}{6\alpha-3}\nu^\frac{3-2\alpha}{6\alpha-3}{(3-2\alpha)}^\frac{3-2\alpha}{6\alpha-3}(4\alpha-3) }.
\end{aligned}
\label{eq:dEadt}
\end{equation}
\end{thm}

\begin{proof}
  Applying inequality \eqref{eq:R3a} to estimate the cubic integral in
  \eqref{eq:Ra}, we have
\begin{equation}
\frac{d\mathcal{E}_\alpha}{dt} \leq C_{\alpha} \, {\| {(-\Delta)}^\frac{\alpha}{2} u \|}_{L^2(\mathcal{I})}^{\frac{8\alpha-3}{2\alpha}}\ {\| {(-\Delta)}^\alpha u \|}_{L^2(\mathcal{I})}^\frac{3-2\alpha}{2\alpha} - \nu {\| {(-\Delta)}^\alpha u \|}_{L^2(\mathcal{I})}^2\,,
\label{eq:dEadt2}
\end{equation}
where $C_{\alpha}$ is defined in \eqref{eq:uLinf}. Then, Young's
inequality is used to estimate the first term on the RHS of
\eqref{eq:dEadt2} such that the second term is eliminated and we
finally obtain \eqref{eq:dEadt}. The theorem is thus proved.
\end{proof}
{\noindent As regards the somewhat narrow range of applicability
  of estimate \eqref{eq:dEadt}, we remark that it is a consequence of
  the limitations of the ``spectral splitting'' approach used to bound
  $|u_x(x)|$ in \eqref{eq:ux}, cf.~inequality \eqref{eq:ineq1b}.}

\subsection{Relation Between New Estimates \eqref{eq:dEdt} and
  \eqref{eq:dEadt} and Classical Estimate \eqref{eq:dEdt1D}}

In estimates \eqref{eq:dEdt} and \eqref{eq:dEadt}, assuming the
viscosity coefficient $\nu$ is fixed, the exponents and prefactors are
functions of the fractional dissipation order $\alpha$, i.e.,
$\gamma_1 = \gamma_1(\alpha)$, $\sigma_1 = \sigma_1(\alpha)$,
$\gamma_{\alpha} = \gamma_{\alpha}(\alpha)$ and $\sigma_{\alpha} =
\sigma_{\alpha}(\alpha)$. Their dependence on $\alpha \in (1/4,1]$ and
$\alpha \in (3/4,1]$, which are the respective ranges of the
fractional dissipation orders for which the estimates \eqref{eq:dEdt}
and \eqref{eq:dEadt} are valid, is shown in Figure \ref{fig:estim}. It
is clear that $\gamma_1(\alpha), \gamma_{\alpha}(\alpha) \rightarrow
5/3$ as $\alpha \rightarrow 1$, indicating that our upper bounds
\eqref{eq:dEdt} and \eqref{eq:dEadt} are consistent with the original
estimate \eqref{eq:dEdt1D} obtained in \cite{l06,ld08}.
Analogous property holds for the prefactor $\sigma_1(\alpha)$,
  but not for $\sigma_{\alpha}(\alpha)$.  Sharpness of these
estimates will be assessed in Sections \ref{sec:resultsR} and
\ref{sec:resultsRa}.
\begin{figure}
\centering
\subfigure[]{
\label{fig:subfig:a} 
\includegraphics[width=3.35in]{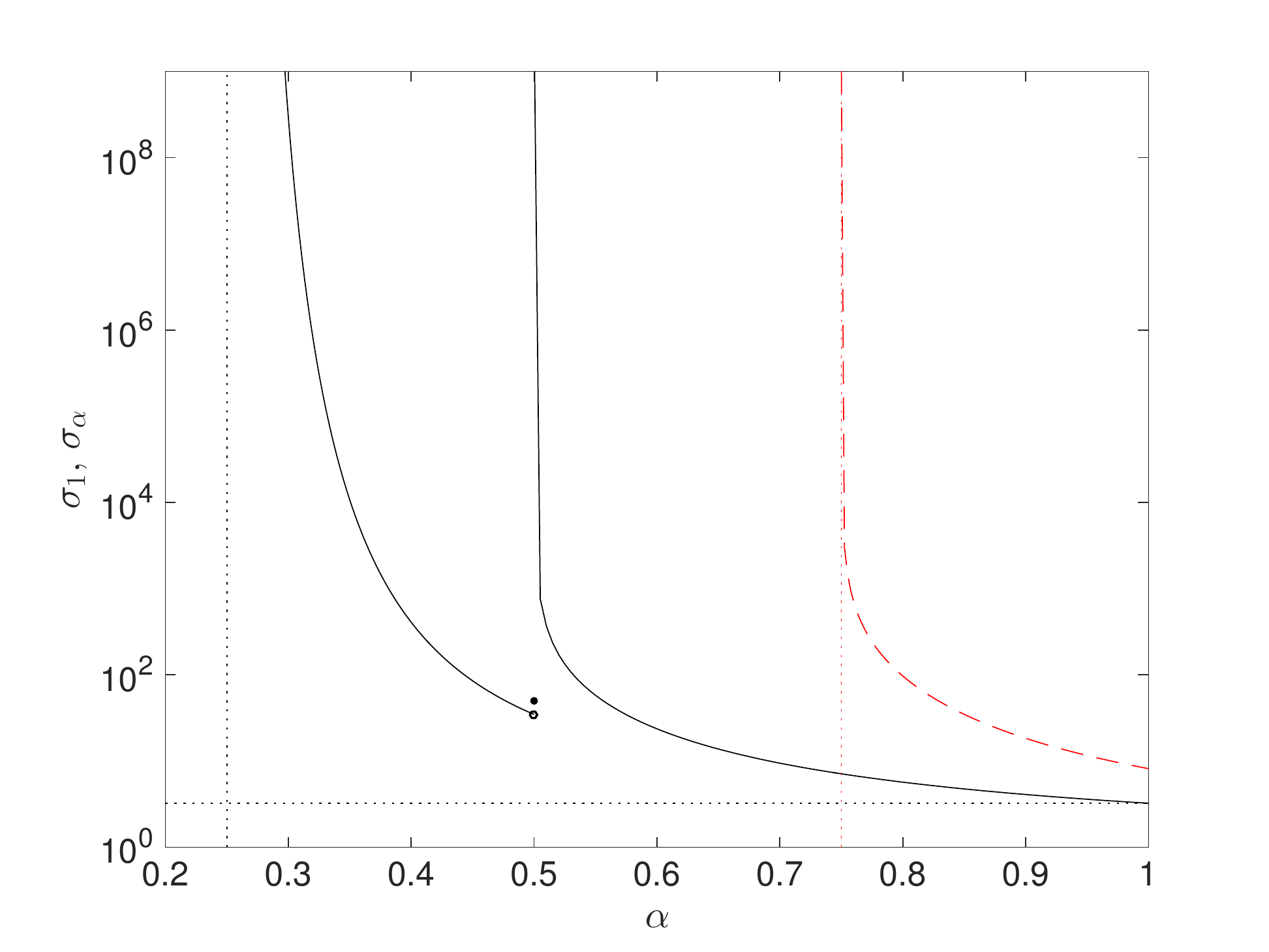}}
\hspace{-0.2in}
\subfigure[]{
\label{fig:subfig:b} 
\includegraphics[width=3.35in]{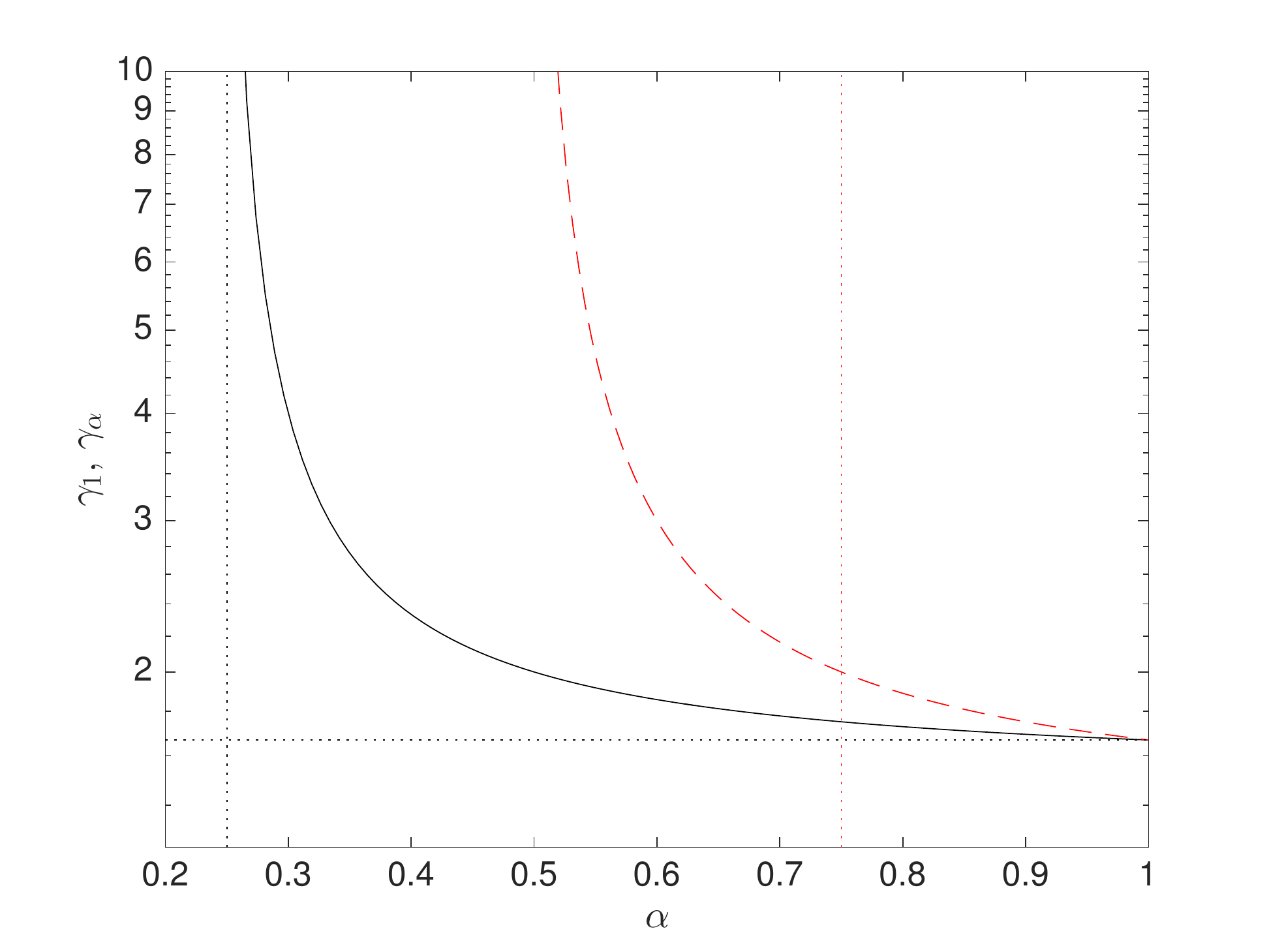}}
\caption{Constant prefactors $\sigma_1$, $\sigma_{\alpha}$ (a) and
  exponents $\gamma_1$, $\gamma_{\alpha}$ (b) of upper bounds
  \eqref{eq:dEdt} (black solid lines) and \eqref{eq:dEadt} (red dashed
  lines) as functions of the fractional dissipation exponent $\alpha$.
  The vertical dashed lines correspond to $\alpha=1/4$ and
    $3/4$, whereas the horizontal dashed lines represent the
  prefactors and exponents from estimate \eqref{eq:dEdt1D}
  corresponding to $\alpha=1$ \cite{l06,ld08}. The viscosity
  coefficient is $\nu=0.01$.}
\label{fig:estim} 
\end{figure}

\section{Methodology for Probing Sharpness of Estimates \eqref{eq:dEdt} and \eqref{eq:dEadt}}
\label{sec:method}

In this section we discuss the approach which will allow us to verify
whether or not the upper bounds \eqref{eq:dEdt} and \eqref{eq:dEadt}
on the rate of growth of the classical and fractional enstrophy are
sharp. An estimate of the type \eqref{eq:dEdt} or \eqref{eq:dEadt} is
considered ``sharp'' when the upper bound on its {right-hand
  side} can be realized by certain fields $u$ with prescribed
(classical or fractional) enstrophy. In other words, in regard to
estimate \eqref{eq:dEdt}, if we can find a family of fields
{$\tilde{u}_{\bar{\E}}$} such that
{$\E(\tilde{u}_{\bar{\E}}) = \bar{\E}$} and {
  $\R_{\E}(\tilde{u}_{\bar{\E}}) \rightarrow \sigma_1 \,
  \bar{\E}^{\gamma_1}$ } when {$\bar{\E} \rightarrow \infty$},
then this estimate is declared sharp (analogously for estimate
\eqref{eq:dEadt}). We note that, given the power-law structure of the
upper bounds in \eqref{eq:dEdt} and \eqref{eq:dEadt}, the question of
sharpness may apply independently to both the exponents $\gamma_1$ and
$\gamma_{\alpha}$ as well as the prefactors $\sigma_1$ and
$\sigma_{\alpha}$ (with the caveat that if the exponent is not sharp,
then the question about sharpness of the corresponding prefactor
becomes moot). A natural way to systematically search for fields
saturating an estimate is by solving suitable variational maximization
problems \cite{l06,ld08,ap11a,ap13a} and such problems corresponding
to estimates \eqref{eq:dEdt} and \eqref{eq:dEadt} are introduced
below.

\subsection{Variational Formulation}
\label{sec:variational}

For given {$\bar{\E} > 0$} and {$\bar{\E}_\alpha > 0$ on
  the RHS} in estimates \eqref{eq:dEdt} and \eqref{eq:dEadt}, we
define the respective maximizers as
\begin{equation}
\begin{aligned}
& {\widetilde{u}_{\bar{\E}}} = \argmax\limits_{u\in {\mathcal {S}_{\bar{\E}}}}  \mathcal {R}_{\mathcal {E}}(u)\,, \\
& {\mathcal {S}_{\bar{\E}}}=\left\{u\in H^{1+\alpha}(\mathcal{I}): \frac{1}{2}\int_0^1 {\left|u_x\right|}^2\,dx={\bar{\E}}\right\}\,,
\end{aligned}
\label{eq:maxR}
\end{equation}
and
\begin{equation}
\begin{aligned}
& { \widetilde{u}_{\bar{\E}_\alpha} } = \argmax\limits_{u\in { \mathcal {S}_{\bar{\E}_\alpha}}}  \mathcal {R}_{\mathcal {E}_\alpha}(u)\,, \\
& {\mathcal {S}_{\bar{\E}_\alpha}} =
\left\{
\begin{alignedat}{2}
& \left\{u\in H^{2\alpha}(\mathcal{I}): \frac{1}{2}\int_0^1 {\left| {(-\Delta)}^\frac{\alpha}{2}u \right|}^2\,dx={\bar{\E}_\alpha} \right\}& \qquad & \text{for} \ \alpha>\frac{1}{2},  \\
& \left\{u\in H^1(\mathcal{I}): \frac{1}{2}\int_0^1 {\left| {(-\Delta)}^\frac{\alpha}{2}u \right|}^2\,dx={\bar{\E}_\alpha} \right\} &&  \text{for} \ \alpha\leq \frac{1}{2},
\end{alignedat}
\right.
\end{aligned}
\label{eq:maxRa}
\end{equation}
where ``$\argmax$'' denotes the state realizing the maximum, whereas
${\mathcal {S}_{\bar{\E}}}$ and ${\mathcal
  {S}_{\bar{\E}_\alpha}}$ represent the constraint manifolds. The
choice of the Sobolev spaces in the definitions of these manifolds is
dictated by the minimum regularity of $u$ required in order to make
the expressions for $\R_{\E}(u)$ and $\R_{\E_{\alpha}}(u)$,
cf.~\eqref{eq:R} and \eqref{eq:Ra}, well defined.  {While
  establishing rigorously the solvability of optimization problems
  \eqref{eq:maxR} and \eqref{eq:maxRa}, especially for large values of
  {$\bar{\E}$} and {$\bar{\E}_\alpha$}, is a difficult task
  (which is outside the scope of the present study), the computational
  results reported in Section \ref{sec:results} indicate that the
  maxima defined by these problems are indeed attained.}  A numerical
approach to solution of maximization problems \eqref{eq:maxR} and
\eqref{eq:maxRa} is presented below.

\subsection{Gradient-Based Solution of Problems  \eqref{eq:maxR} and \eqref{eq:maxRa}}
\label{sec:grad}

To fix attention, we first focus on the solution of problem
\eqref{eq:maxR}. For a fixed {$\bar{\E}$}, its maximizers can be computed as
${\tuEbar} = \lim_{n\rightarrow \infty} u^{(n)}$, where the consecutive
approximations $u^{(n)}$ are defined via the following iterative
procedure (in practice, only a finite number of iterations is
performed)
\begin{subequations}
\label{eq:iter}
\begin{align}
u^{(n+1)} & = \mathbb{P}\mathcal {S}\left(u^{(n)} + \tau_n\nabla^{H^{1+\alpha}}\mathcal {R}_{\mathcal {E}}(u^{(n)}) \right),
\qquad n=1,2,\dots,  \label{eq:itera}\\
u^{(1)} & = u^0\,, \label{eq:iterb}
\end{align}
\end{subequations}
in which $u^0$ is the initial guess chosen such that $u^0 \in
{\mathcal {S}_{\bar{\E}}}$ {and $\int_0^1 u^0 \,dx = 0$},
$\nabla^{H^{1+\alpha}}\R_{\E}(u^{(n)})$ is the Sobolev gradient of
$\R_{\E}(u)$ evaluated at the $n$th iteration with respect to the
suitable Sobolev topology, $\tau_n$ is the corresponding length of the
step and $\PPS \; : \; H^{1+\alpha}(\I) \rightarrow {\mathcal
  {S}_{\bar{\E}}}$ is the projection operator ensuring that every
iterate $u^{(n)}$ is restricted to the constraint manifold
${\mathcal {S}_{\bar{\E}}}$. For simplicity of presentation, in
\eqref{eq:itera} we used the steepest-ascent approach, however, in
practice one can use a more advanced maximization approach
characterized by faster convergence, such as for example the conjugate
gradients method \cite{nw00}. For the maximization problem
\eqref{eq:maxRa} the maximizer
{$\widetilde{u}_{\bar{\E}_\alpha}$} can be computed with an
algorithm similar to \eqref{eq:iter}, but involving the constraint
manifold {$\mathcal {S}_{\bar{\E}_\alpha}$}, the gradient
$\nabla^{{H^{2\alpha}}}\R_{\Ea}$ (or, $\nabla^{{H^{1}}}\R_{\Ea}$) and
the projection operator $\PPS_{\alpha}$, all of which are defined
below.

As regards evaluation of the gradients $\nabla^{H^{1+\alpha}}\R_{\E}$
and $\nabla^{{H^{2\alpha}}}\R_{\Ea}$ (or,
$\nabla^{{H^{1}}}\R_{\Ea}$), the starting point are the G\^ateaux
differentials of the objective functions $\eqref{eq:R}$ and
\eqref{eq:Ra}, defined as
\begin{equation*}
\mathcal {R}'_{\mathcal {E}}(u; v) := \lim_{\epsilon\rightarrow 0} \frac{\mathcal {R}_{\mathcal {E}}(u+\epsilon v)- \mathcal {R}_{\mathcal {E}}(u)}{\epsilon} \quad \text{and} \quad
\mathcal {R}'_{\mathcal {E}_\alpha}(u; v) := \lim_{\epsilon\rightarrow 0} \frac{\mathcal {R}_{\mathcal {E}_\alpha}(u+\epsilon v)-\mathcal {R}_{\mathcal {E}_\alpha}(u)}{\epsilon}
\end{equation*}
which can be evaluated as follows
\begin{subequations}
\label{eq:dR}
\begin{align}
\mathcal {R}'_{\mathcal {E}}(u; v) & = \int_0^1 \left[3u_xu_{xx}+2\nu {(-\Delta)}^\alpha u_{xx}\right]v \, dx\,, \label{eq:dRa} \\
\mathcal {R}'_{\mathcal {E}_\alpha}(u; v) & = \int_0^1 \left[ {(-\Delta)}^\alpha u_x\cdot u - {(-\Delta)}^\alpha(u\cdot u_x)-2\nu {(-\Delta)}^{2\alpha}u \right]v \, dx\,, \label{eq:dRb}
\end{align}
\end{subequations}
where integration by parts has been used to factorize the
``direction'' $v$. Next, recognizing that, for a fixed $u$ and when
regarded as functions of the second argument $v$, the G\^ateaux
differentials \eqref{eq:dRa}--\eqref{eq:dRb} are bounded linear
functionals on the given Sobolev spaces, we can invoke the
Riesz theorem \cite{l69} to obtain
\begin{subequations}
\label{eq:Riesz}
\begin{alignat}{3}
 \mathcal {R}'_{\E}(u; v) 
&=  {\left\langle \nabla^{L^2}\mathcal {R}_{\mathcal {E}}(u), v \right\rangle}_{L^2(\mathcal{I})} &
&= {\left\langle \nabla^{H^{1+\alpha}}\mathcal {R}_{\mathcal {E}}(u), v \right\rangle}_{H^{1+\alpha}(\mathcal{I})}\,,& &  \label{eq:Riesza} \\
 \mathcal {R}'_{\Ea}(u; v) 
&=  {\left\langle \nabla^{L^2}\mathcal {R}_{\Ea}(u), v \right\rangle}_{L^2(\mathcal{I})} &
&= {\left\langle \nabla^{H^{2\alpha}}\mathcal {R}_{\Ea}(u), v \right\rangle}_{H^{2\alpha}(\mathcal{I})}, &
& \quad \text{for} \ \alpha>\frac{1}{2}\,, \label{eq:Rieszb} \\
&&  
&= {\left\langle \nabla^{H^{1}}\mathcal {R}_{\Ea}(u), v \right\rangle}_{H^{1}(\mathcal{I})}, &
& \quad \text{for} \ \alpha \le \frac{1}{2}\,,  \label{eq:Rieszc}
\end{alignat}
\end{subequations}
where the Sobolev inner products are defined as
\begin{subequations}
\label{eq:ip}
\begin{align}
 {\left\langle u, v \right\rangle}_{H^{1+\alpha}(\mathcal{I})} & := \int_0^1 u\,v + \ell_1^2 u_x\,v_x + \ell_2^{2+2\alpha} {(-\Delta)}^\frac{\alpha}{2} u_x\cdot {(-\Delta)}^\frac{\alpha}{2} v_x\,dx\,, \label{eq:ipa} \\
{\left\langle u, v \right\rangle}_{H^{2\alpha}(\mathcal{I})} & := \int_0^1 u\,v + \ell_1^2 u_x\,v_x + \ell_2^{4\alpha} {(-\Delta)}^{\alpha-\frac{1}{2}} u_x\cdot {(-\Delta)}^{\alpha-\frac{1}{2}} v_x\,dx\,, \label{eq:ipb} \\
{\left\langle u, v \right\rangle}_{H^1(\mathcal{I})} & := \int_0^1 u\,v + \ell_1^1 u_x v_x \, dx\, \label{eq:ipc}
\end{align}
\end{subequations}
in which $\ell_1$ and $\ell_2$ are constants with the dimension of a
length-scale (the significance of the choice of these constants will
be discussed below). In order to characterize the Sobolev gradients
defined in the spaces $H^{1+\alpha}(\I)$, $H^{2\alpha}(\I)$ and
$H^{1}(\I)$, cf.~\eqref{eq:maxR}--\eqref{eq:maxRa}, we first derive
expressions for the $L^2$ gradients from
  \eqref{eq:dRa}--\eqref{eq:dRb}
\begin{subequations}
\label{eq:gradL2}
\begin{align}
\nabla^{L^2}\mathcal {R}_{\mathcal {E}}(u) & = 3u_xu_{xx}+2\nu {(-\Delta)}^\alpha u_{xx}\,, \label{eq:gradL2a} \\
\nabla^{L^2}\mathcal {R}_{\mathcal {E}_\alpha}(u) & = {(-\Delta)}^\alpha u_x\cdot u - {(-\Delta)}^\alpha(u\cdot u_x)-2\nu {(-\Delta)}^{2\alpha}u \label{eq:gradL2b}
\end{align}
\end{subequations}
and then invoke Riesz relations \eqref{eq:Riesza}--\eqref{eq:Rieszc}
which, upon performing integration by parts and noting that $v$
  is arbitrary, yields
\begin{subequations}
\label{eq:gradH}
\begin{alignat}{2}
 \left[  {1} - l_1^2 \Delta - l_2^{2+2\alpha} {(-\Delta)}^\alpha \Delta \right] \nabla^{H^{1+\alpha}}\mathcal {R}_{\mathcal {E}}(u) & = \nabla^{L^2}\mathcal {R}_{\mathcal {E}}(u)\,, & \qquad & \text{on} \ \I, \label{eq:gradHa} \\
 \left[  {1} - l_1^2 \Delta - l_2^{4\alpha} {(-\Delta)}^{2\alpha-1} \Delta \right] \nabla^{H^{2\alpha}}\mathcal {R}_{\mathcal {E}_\alpha}(u) & = \nabla^{L^2}\mathcal {R}_{\mathcal {E}_\alpha}(u)\,,  & & \text{on} \ \I, \label{eq:gradHb} \\
 \left[  {1} - l_1^2 \Delta \right] \nabla^{H^{1}}\mathcal {R}_{\mathcal {E}_\alpha}(u) & = \nabla^{L^2}\mathcal {R}_{\mathcal {E}_\alpha}(u)\,,  & & \text{on} \ \I, \label{eq:gradHc} \\
& \hspace*{-3.0cm}  \text{Periodic Boundary Conditions.} && \nonumber
\end{alignat}
\end{subequations}
These boundary-value problems allow us to determine the required
Sobolev gradients in terms of the $L^2$ gradients
\eqref{eq:gradL2a}--\eqref{eq:gradL2b}. We now return to the question of
the choice of the length-scale parameters $\ell_1$ and $\ell_2$. As is
evident from the form of these expressions, for different values of
$\ell_1$ and $\ell_2$ inner products
\eqref{eq:ipa}--\eqref{eq:ipc} define {\em equivalent} norms as
long as $\ell_1, \ell_2 \in (0,\infty)$. On the other hand, as
demonstrated in \cite{pbh04}, computation of Sobolev gradients by
solving elliptic boundary-value problems
\eqref{eq:gradHa}--\eqref{eq:gradHc} can be interpreted as application
of spectral low-pass filters to the $L^2$ gradients with the
parameters $\ell_1$ and $\ell_2$ defining the cut-off length-scales.
Thus, while for sufficiently ``good'' initial guesses $u^0$ iterations
of the type \eqref{eq:iter} with different values of $\ell_1$ and
$\ell_2$ lead to the {\em same} maximizer {$\tuEbar$}, the actual rate of
convergence usually depends very strongly on the choice of these
parameters \cite{ap11a,ap13a}.  \medskip

Since the constraints defining the manifolds in \eqref{eq:maxR} and
\eqref{eq:maxRa} are quadratic, the corresponding projection operators
are naturally defined in terms of the following rescalings
(normalizations)
\begin{subequations}
\begin{align}
 \mathbb{P}\mathcal {S}(u) & := \sqrt{ \frac{{\bar{\E}}}{ \frac{1}{2}\int_0^1 {|u_x|}^2\,dx } }\ u \,, \label{eq:PN} \\
 \mathbb{P}\mathcal {S}_\alpha(u) & := \sqrt{ \frac{{\bar{\E}_\alpha}}{ \frac{1}{2}\int_0^1 {\big| {(-\Delta)}^\frac{\alpha}{2}u \big|}^2\,dx } }\ u\ \label{eq:PNa}
\end{align}
\end{subequations}
{which in the language of optimization on manifolds can be
  interpreted as ``retractions'' from the tangent subspace to the
  constraint manifold \cite{ams08}. The form of expressions
  \eqref{eq:PN} and \eqref{eq:PNa} is particularly simple, because the
  constraint manifolds ${\mathcal {S}_{\bar{\E}}}$ and
  ${\mathcal {S}_{\bar{\E}_\alpha}}$, cf.~\eqref{eq:maxR} and
  \eqref{eq:maxRa}, can be regarded as ``spheres'' in their respective
  functional spaces.} Finally, the optimal step length $\tau_n$ in
\eqref{eq:iter} is found by solving an arc-minimization problem
\begin{equation}
\tau_n :=  \argmax\limits_{\tau>0} \left\{ \mathcal {R}_{\mathcal {E}}
\left[ \mathbb{P}\mathcal {S}\left( u^{(n)}+\tau\nabla^{H^{1+\alpha}} \mathcal {R}_{\mathcal {S}}(u^{(n)}) \right) \right]\right\}\, 
\label{eq:taun}
\end{equation}
which is an adaptation of standard line minimization \cite{nw00} to
the case with quadratic constraints. This step is performed with a
straightforward generalization of Brent's method \cite{numRecipes} and
an analogous approach is also used to compute the step size when
solving problem \eqref{eq:maxRa}.

\subsection{Tracing Solutions Branches via Continuation}
\label{sec:continuation}

Families of maximizers {$\tuEbar$} corresponding to a range of enstrophy
values ${\bar{\E}} = \E^{(m)}$, $m=0,1,\dots$, are obtained using a
continuation approach where the solution $\widetilde{u}_{\E^{(m)}}$
determined by the iteration process \eqref{eq:iter} at some enstrophy
value $\E^{(m)}$ is used as the initial guess $u^0$ for iterations at
the next, slightly larger, value $\E^{(m+1)} = \E^{(m)} + \Delta\E$
for some $\Delta\E > 0$. By choosing sufficiently small steps
$\Delta\E$, one can ensure that the iterations at a given enstrophy
value are rapidly convergent. The same continuation approach is also
used to compute the families of the maximizers {$\widetilde{u}_{\bar{\E}_\alpha}$}.

\subsection{Rates of Growth of Enstrophy for $\alpha = 0$}
\label{sec:a0}

To close this section, we provide some comments about the rates of
growth of the classical and fractional enstrophy in the case when
$\alpha = 0$. As regards the first quantity, from \eqref{eq:R} we see
that
\begin{equation}
\R_{\E}(u) \, \le \, \frac{1}{2} \| u_x \|_{L^3(\I)}^3 - \nu  \| u_x \|_{L^2(\I)}^2 
\label{eq:R0}
\end{equation}
which, in view of the property $\| v \|_{L^p(\I)} \le \| v
\|_{L^q(\I)}$ true for $p \le q$ \cite{af05}, implies that for $u \in
\SE$, $\R_{\E}(u)$ may not be bounded and hence the maximization
problem \eqref{eq:maxR} does not have solutions when $\alpha = 0$.
Concerning the rate of growth of the fractional enstrophy, from
\eqref{eq:Ra} we obtain
\begin{equation}
\frac{d\K(u)}{dt}  = -\K(u),
\label{eq:Ra0}
\end{equation}
where $\K(u) := \E_0(u) = (1/2) \int_0^1 u^2 \, dx$ is the kinetic
energy and we used the property that $\int_0^1 uuu_x\, dx = 0$. Then,
the maximization problem \eqref{eq:maxRa} takes the form $\max_{u, \;
  \K(u) = {\bar{\K}}} \left[ - \K(u)\right]$ for some
${\bar{\K}} > 0$, and it is clear that any function $u \in
H^1(\I)$ such that $\K(u) = {\bar{\K}}$ is a solution. Thus, when
$\alpha = 0$, the maximization problem \eqref{eq:maxRa} has infinitely
many solutions.

\section{Computational Results}
\label{sec:results}

In this section we present and analyze the results obtained by solving
problems \eqref{eq:maxR} and \eqref{eq:maxRa} for a broad range of
{$\bar{\E}$} and {$\bar{\E}_\alpha$} and for different
values of the fractional dissipation exponent $\alpha \in [0,1]$. We
begin by describing the numerical techniques employed to discretize
the approach introduced in Section \ref{sec:grad} and then summarize
the values of the different numerical parameters used.  \medskip

For a given field $u$, the gradient expressions
\eqref{eq:gradL2a}--\eqref{eq:gradL2b} are evaluated using a
spectrally accurate Fourier-Galerkin approach in which the nonlinear
terms are computed in the physical space with dealiasing based on the
$2/3$ rule \cite{b01}. A similar Fourier-Galerkin approach is also
used to solve the boundary-value problems
\eqref{eq:gradHa}--\eqref{eq:gradHc} for the Sobolev gradients. As
will be discussed in more detail below, the maximizers
{$\tuEbar$} and {$\tuEabar$} corresponding to increasing
values of, respectively, {$\bar{\E}$} and
{$\bar{\E}_\alpha$} are characterized by shock-like steep fronts
of decreasing thickness.  Resolving these regions accurately requires
numerical resolutions (given in terms of the numbers $N$ of Fourier
modes) increasing with {$\bar{\E}$} and {$\bar{\E}_\alpha$}.
In our computations we used the resolutions $N = 512,\dots,8388608$
which were refined adaptively based on the criterion that the Fourier
coefficients corresponding to several largest resolved wavenumbers be
at the level of the round-off errors, i.e., $|\hat{u}_k| \sim
\O(10^{-14})$ for $k \lessapprox N$.  {By carefully performing
  such grid refinement it was possible to assert that in all cases the
  computed approximations converge to well-defined maximizers.}  In
all cases considered the viscosity was set to $\nu = 0.01$. The
length-scale parameters $\ell_1$ and $\ell_2$ in the definitions of
the Sobolev inner products \eqref{eq:ip} were chosen to maximize the
rate of convergence of iterations \eqref{eq:iter} for given values of
$\alpha$, {$\bar{\E}$} or {$\bar{\E}_\alpha$}. The best
results were obtained for $\ell_1 \in [1,10]$ and $\ell_2 \in
[10^{-6},1]$ with smaller values corresponding to larger
{$\bar{\E}$} and {$\bar{\E}_\alpha$}. Since the
gradient-based method from Section \ref{sec:grad} may find {\em local}
maximizers only, in addition to the continuation approach described in
Section \ref{sec:continuation}, we have also started iterations
\eqref{eq:iter} with several different initial guesses $u^0$ intended
to nudge the iterations towards other possible local maximizers
(typically, $u^0(x) = A \sin(2\pi k x)$ with $k=1,2,\dots$ and $A$
chosen so that $\E(u^0)$ or $\Ea(u^0)$ was equal to a prescribed
value). However, these attempts did not reveal any additional
maximizers other than the ones already found using the continuation
approach.

\subsection{Maximum Growth Rate of Classical Enstrophy}
\label{sec:resultsR}

We consider solutions to maximization problem \eqref{eq:maxR} for
$\alpha \in (1/4,1]$, which is the range of fractional dissipation
exponents for which estimate \eqref{eq:dEdt} is valid. The maximum
rate of growth $\R_{\E}({\tuEbar})$ is shown as a function of
{$\bar{\E}$} for different values of $\alpha$ in Figure
\ref{fig:R}.  In this figure we also indicate the upper bounds from
estimate \eqref{eq:dEdt}. The corresponding maximizers
{$\tuEbar$} are shown both in the physical and spectral space in
Figure \ref{fig:tuE} for $ {\bar{\E}}=5,50,500$ and $\alpha \in
[0.5,1]$. It is evident from this figure that the sharp fronts in
these maximizers become steeper with increasing {$\bar{\E}$} and
this effect is more pronounced for smaller values of $\alpha$. This
aspect is further illustrated in Figure \ref{fig:tuE2} where the
maximizers are shown for $\alpha=0.3,0.4$ and small enstrophy values.
Needless to say, accurate determination of maximizers {$\tuEbar$}
for such small values of the fractional dissipation exponent $\alpha$
requires a very refined numerical resolution, cf.~Figure
\ref{fig:tuE2}(d), making the optimization problem \eqref{eq:maxR}
harder and more costly to solve.  This also explains why for small
values of $\alpha$ the data for $\R_{\E}({\tuEbar})$ in Figure
\ref{fig:R} is available only for small {$\bar{\E}$}.  The
relation $\R_{\E}({\tuEbar})$ versus {$\bar{\E}$} is
characterized by certain generic properties evident for all values of
$\alpha$ --- while for small {$\bar{\E}$} the quantity
$\R_{\E}({\tuEbar})$ exhibits a steep growth with
{$\bar{\E}$}, for larger values of {$\bar{\E}$} it develops
a power-law dependence on {$\bar{\E}$}. This behavior can be
quantified by fitting the relation $\R_{\E}({\tuEbar})$ versus
{$\bar{\E}$} locally with the formula $\tsigma_1
{\bar{\E}}^{\tgamma_1}$ and determining the parameters
$\tsigma_1$ and $\tgamma_1$ as functions of {$\bar{\E}$} via a
least-squares procedure.  The dependence of thus determined prefactor
$\tsigma_1$ and exponent $\tgamma_1$ on {$\bar{\E}$} is shown for
different $\alpha$ in Figures \ref{fig:fit1}(a,c,e) and
\ref{fig:fit1}(b,d,f), respectively.  In these figures we also
indicate the values of $\sigma_1$ and $\gamma_1$ obtained in estimate
\eqref{eq:dEdt}. We observe, that as {$\bar{\E}$} increases, both
the prefactor and the exponent obtained via the least-squares fit
approach well-defined limiting values. These limiting values are then
compared against the relations $\sigma_1 = \sigma_1(\alpha)$ and
$\gamma_1 = \gamma_1(\alpha)$ from estimate \eqref{eq:dEdt} in Figures
\ref{fig:sigmagamma1}(a,b). It is clear from Figure
\ref{fig:sigmagamma1}(b) that there is a good quantitative agreement
between the exponent in estimate \eqref{eq:dEdt} and the computational
results. On the other hand, in Figure \ref{fig:sigmagamma1}(a) we see
that the numerically determined prefactors are smaller than the
prefactor derived in estimate \eqref{eq:dEdt}, although they do
exhibit similar trends with $\alpha$ (except for the discontinuity of
the latter at $\alpha = 1/2$). This demonstrates that exponent
$\gamma_1$ in estimate \eqref{eq:dEdt} is sharp, while prefactor
$\sigma_1$ might be improved. We add that we also attempted to solve
the maximization problem \eqref{eq:maxR} for $\alpha \in (0,1/4]$,
however, we were unable to obtain converged solutions. In agreement
with the discussion of the case $\alpha = 0$ in Section \ref{sec:a0},
this indicates that $d\E/dt$ may be unbounded for $\alpha \le 1/4$.

\begin{figure}
\begin{center}
\includegraphics[width=0.95\textwidth]{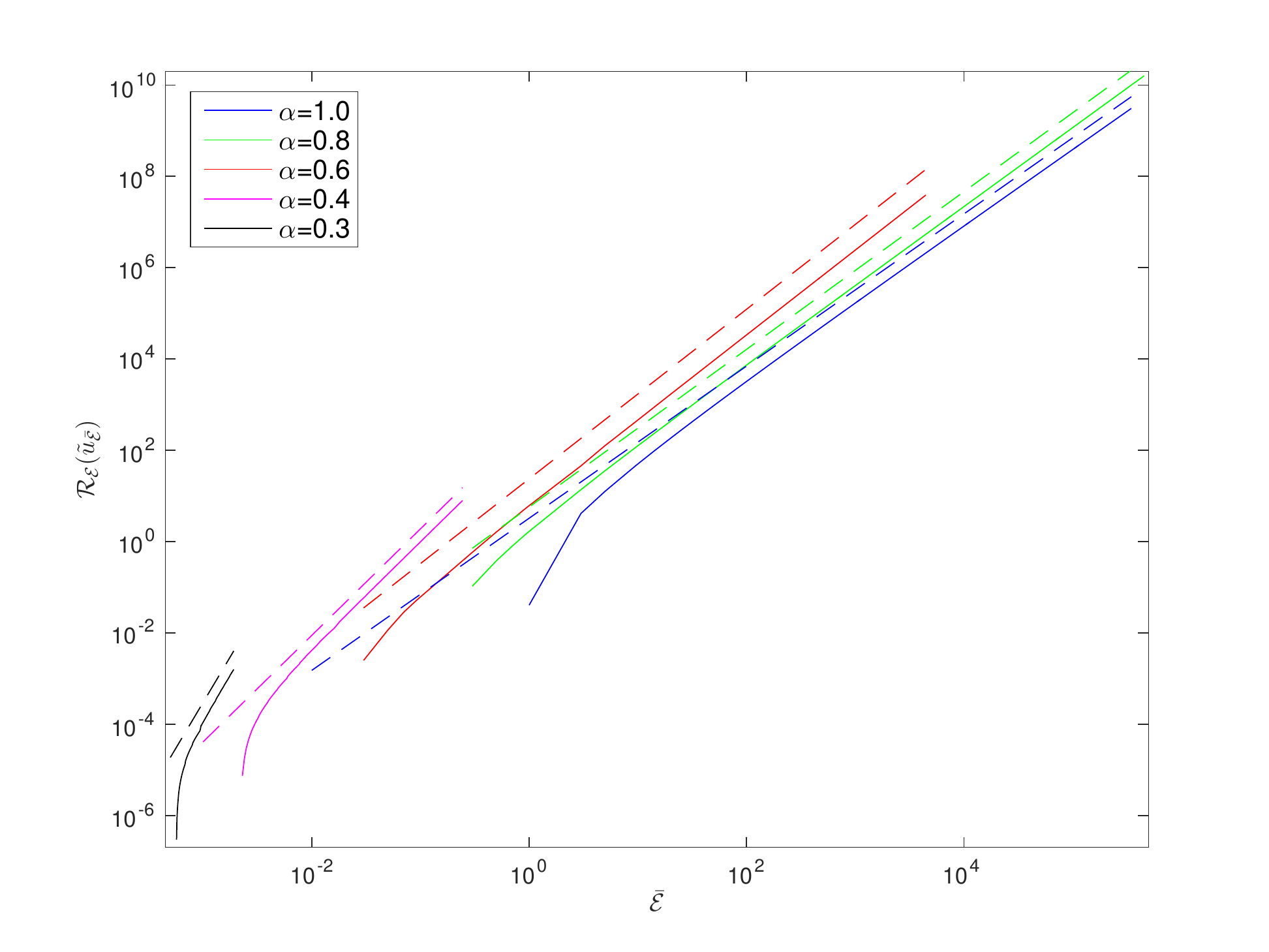}
\caption{Dependence of the maximum enstrophy rate of growth
  $\R_{\E}({\tuEbar})$, obtained by solving optimization problems
  \eqref{eq:maxR}, on {$\bar{\E}$} for different values of $\alpha$ (solid
  lines). The dashed lines represent the corresponding upper bounds
  from estimate \eqref{eq:dEdt}.}
\label{fig:R}
\end{center}
\end{figure}

\begin{figure}
\centering
\subfigure[][]{\includegraphics[width=3.35in]{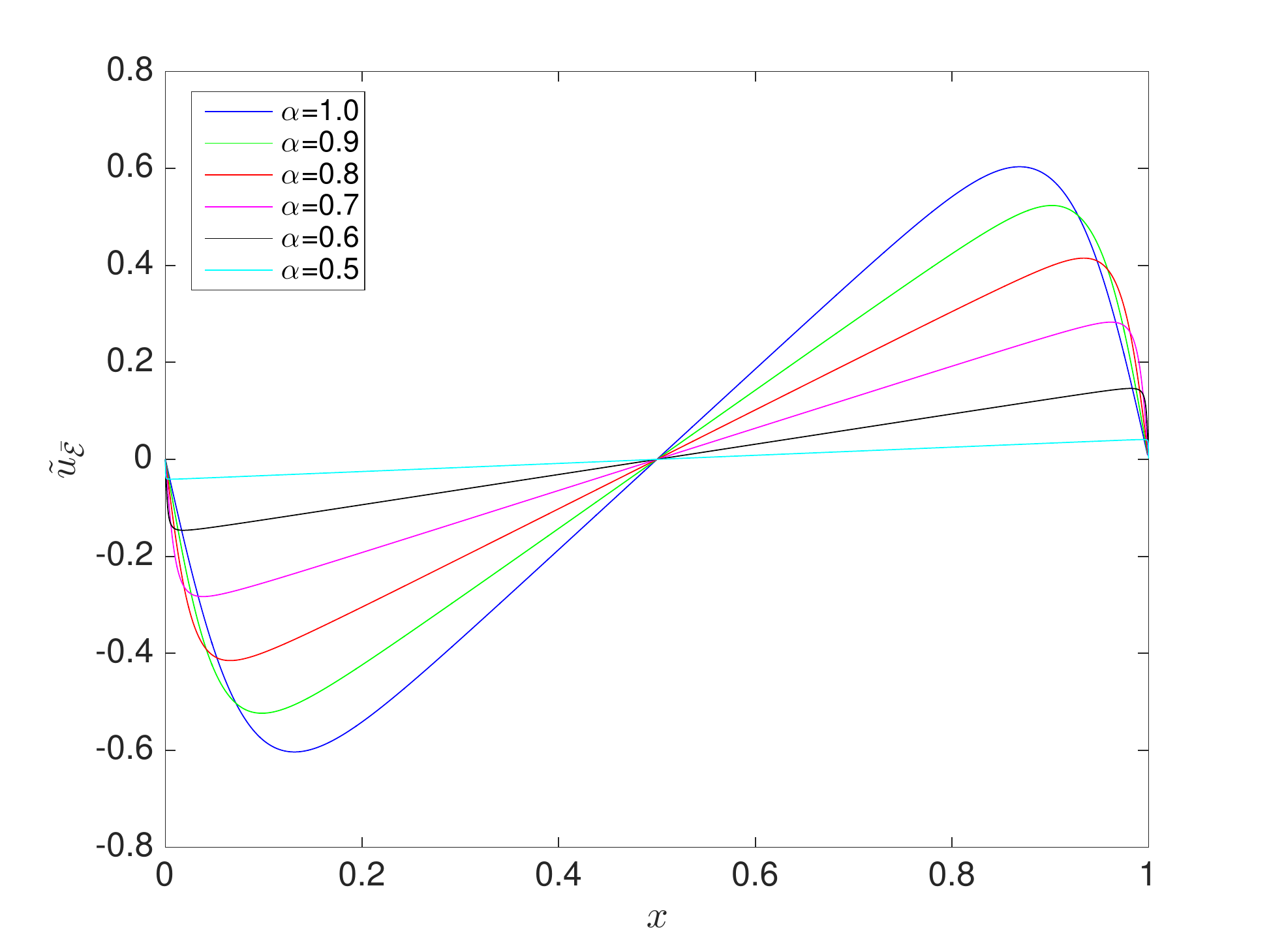}}
\hspace{-0.2in}
\subfigure[][]{\includegraphics[width=3.35in]{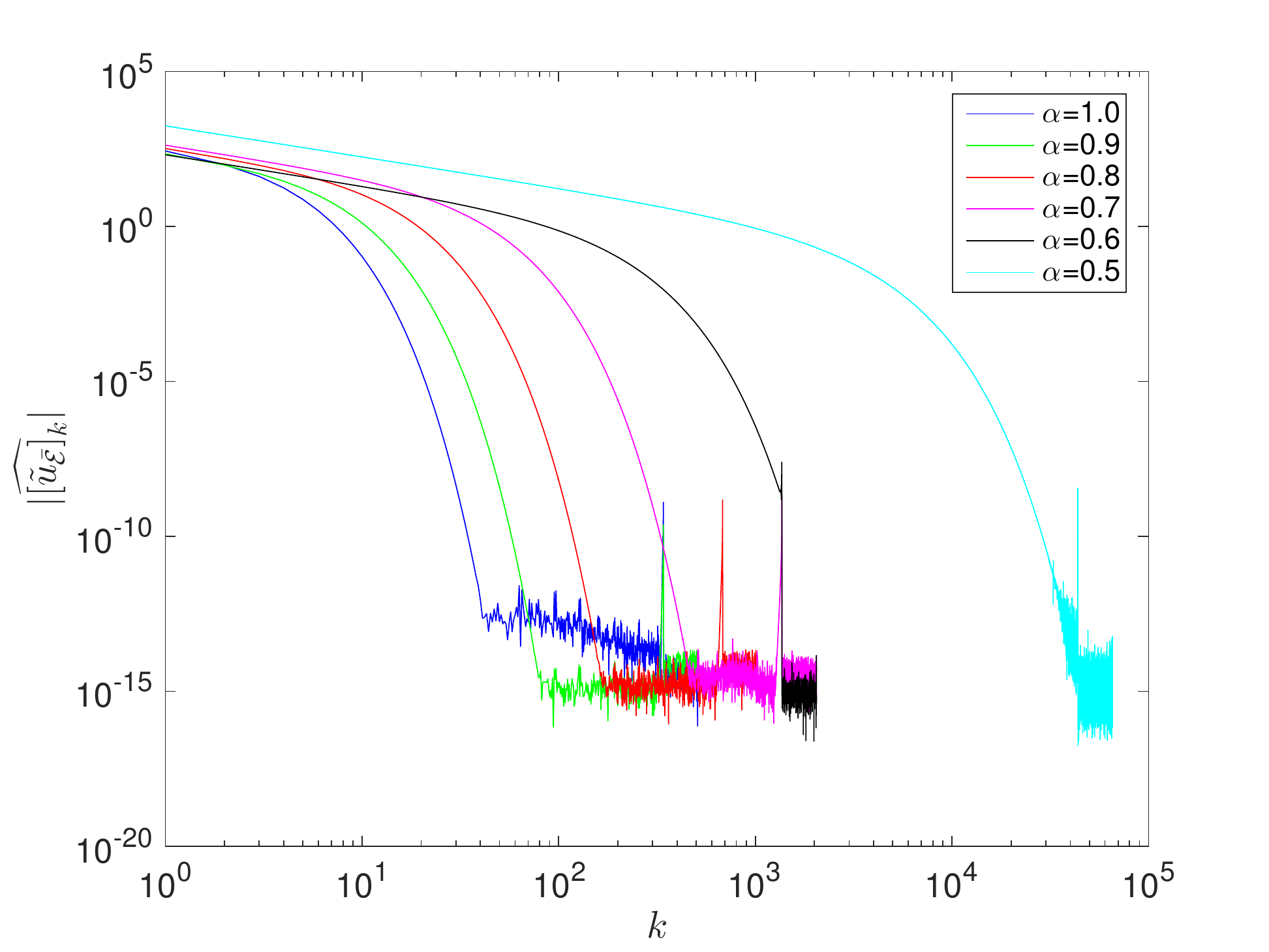}}
\subfigure[][]{\includegraphics[width=3.35in]{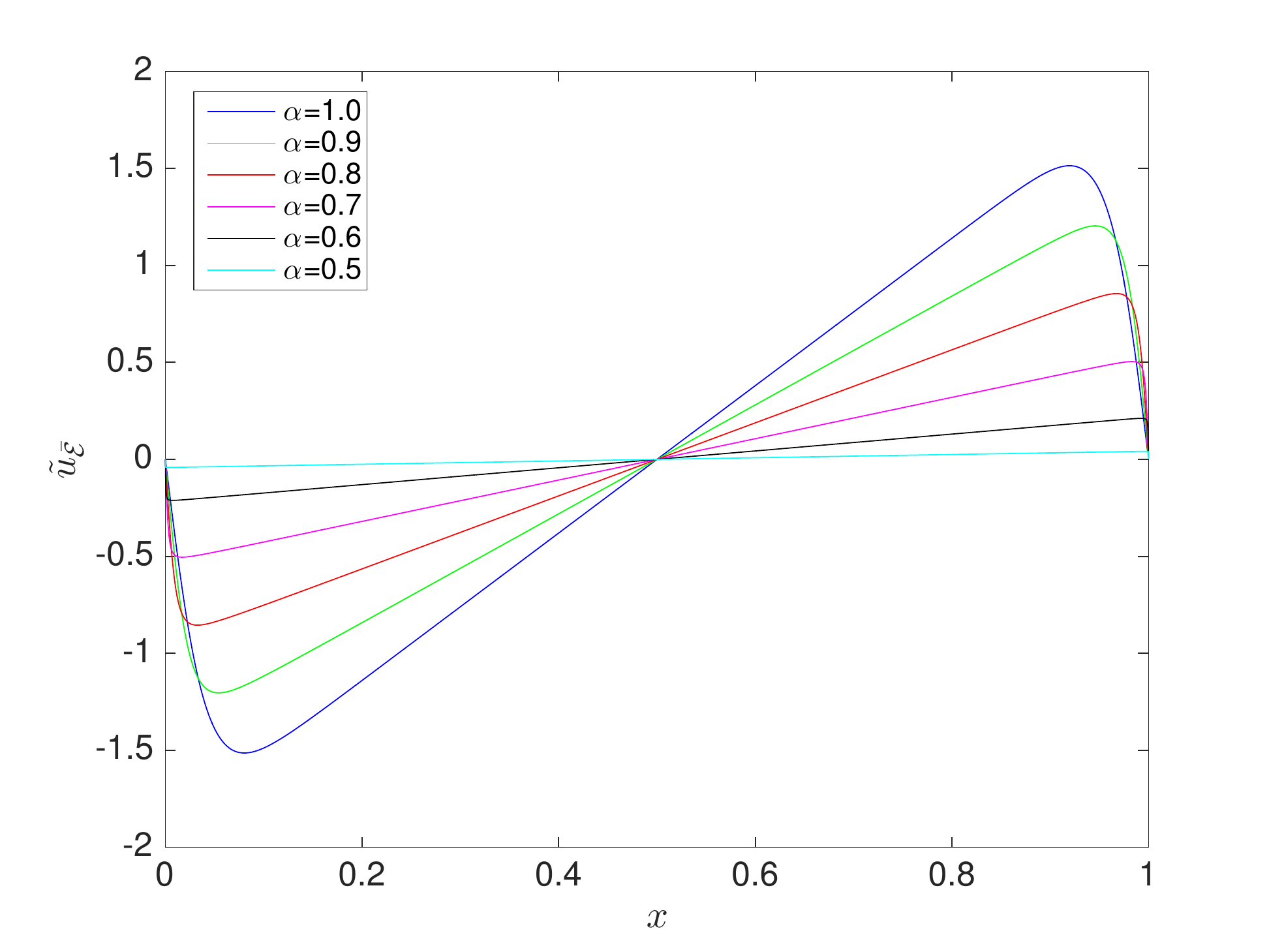}}
\hspace{-0.2in}
\subfigure[][]{\includegraphics[width=3.35in]{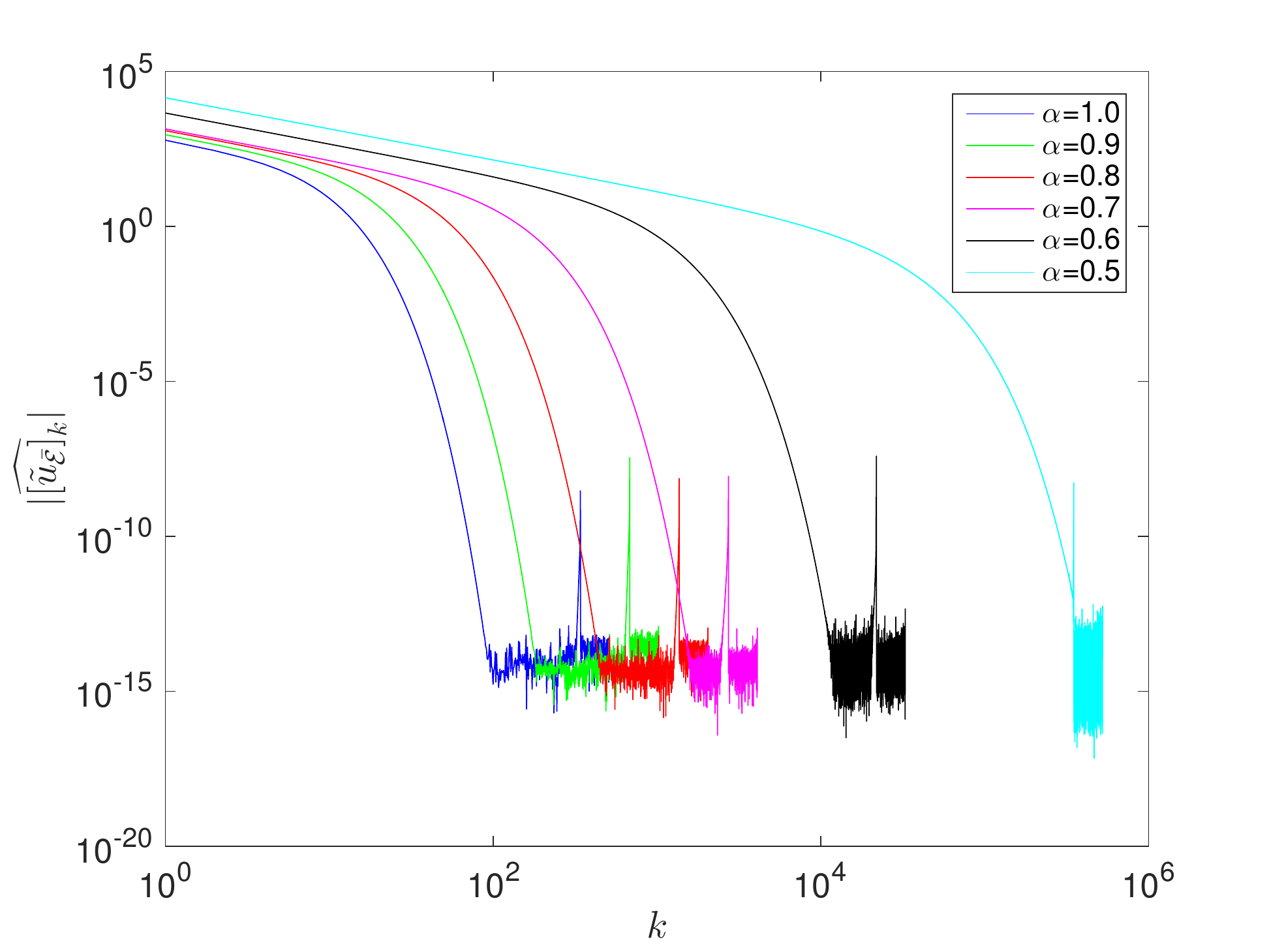}}
\subfigure[][]{\includegraphics[width=3.35in]{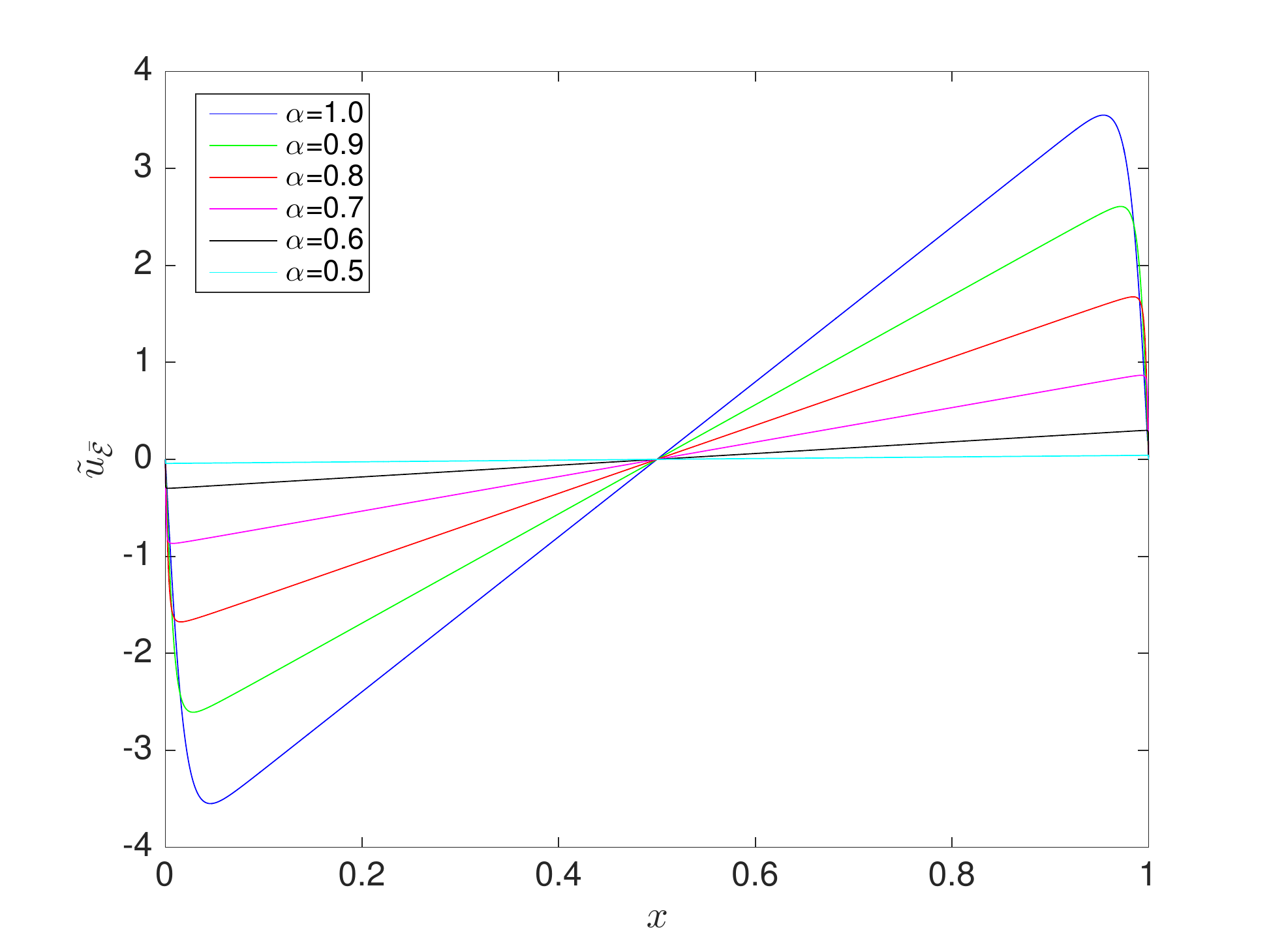}}
\hspace{-0.2in}
\subfigure[][]{\includegraphics[width=3.35in]{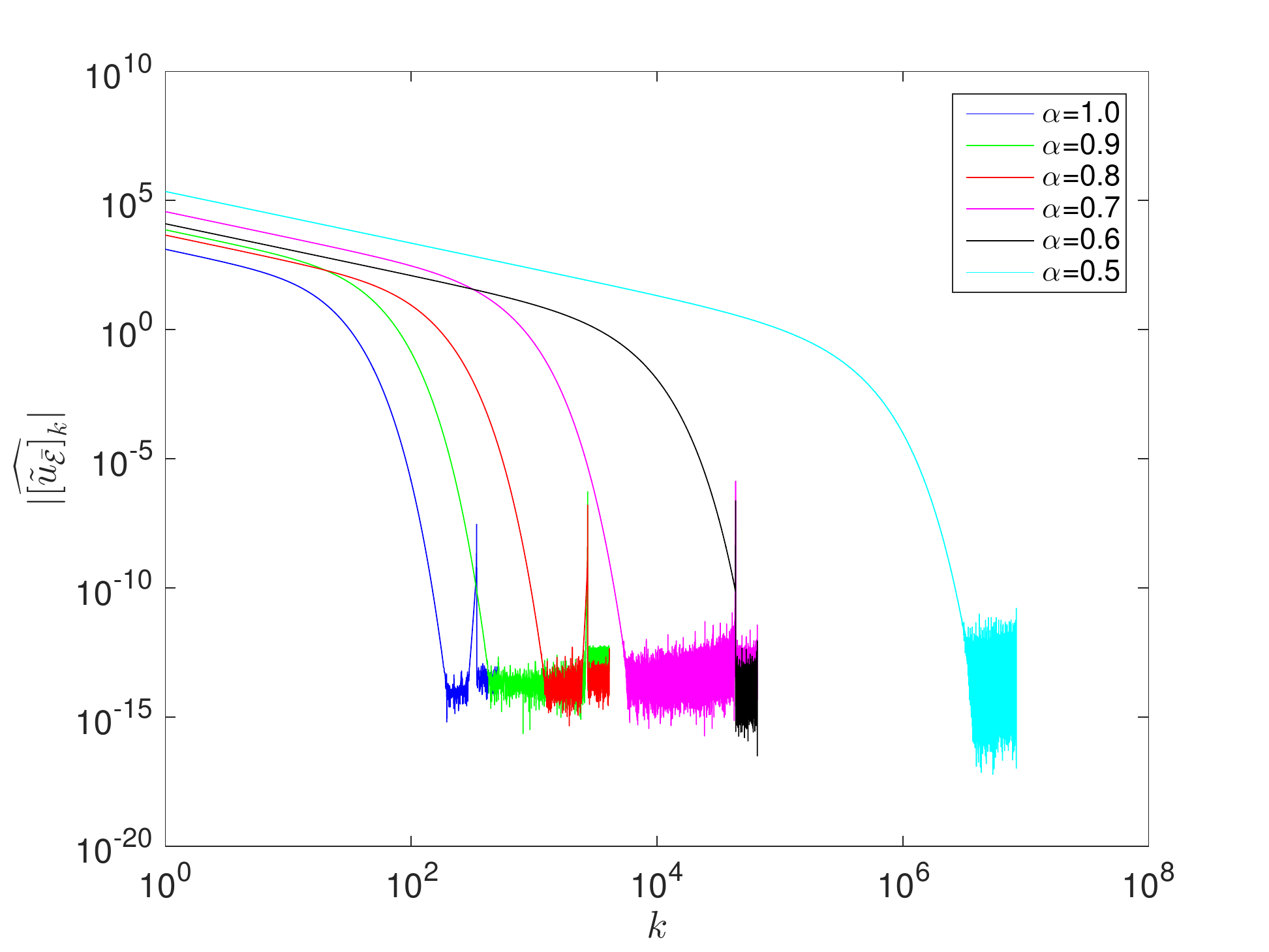}}
\caption{Maximizers {$\tuEbar$} obtained for ${\bar{\E}}=5$ (a,b), ${\bar{\E}}=50$ (c,d)
  and ${\bar{\E}}=500$ (e,f) and different values of $\alpha$. The fields are
  shown in the physical (a,c,e) and spectral (b,d,f) space.}
\label{fig:tuE}
\end{figure}

\begin{figure}
\centering
\subfigure[][]{\includegraphics[width=3.35in]{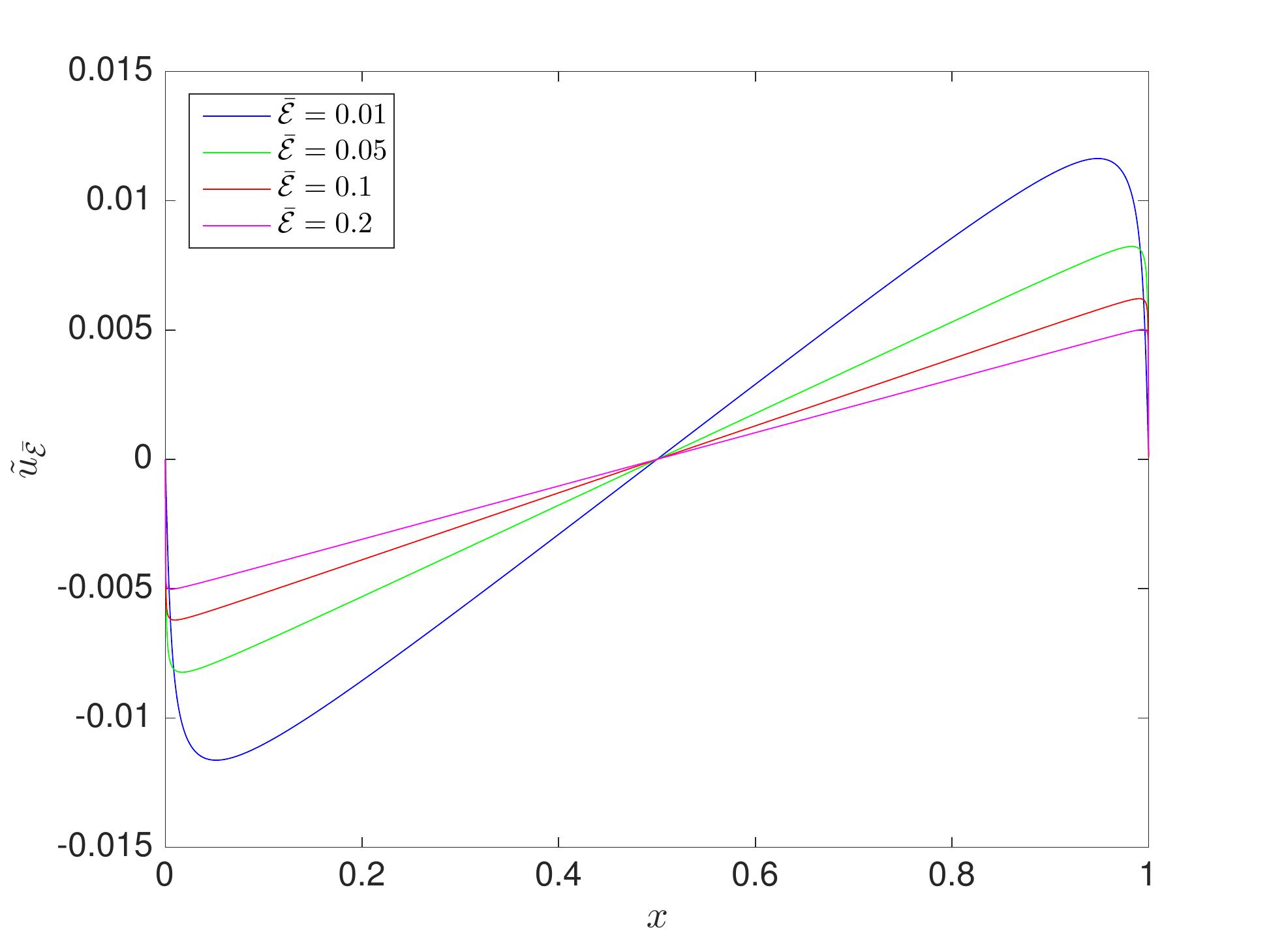}}
\hspace{-0.2in}
\subfigure[][]{\includegraphics[width=3.35in]{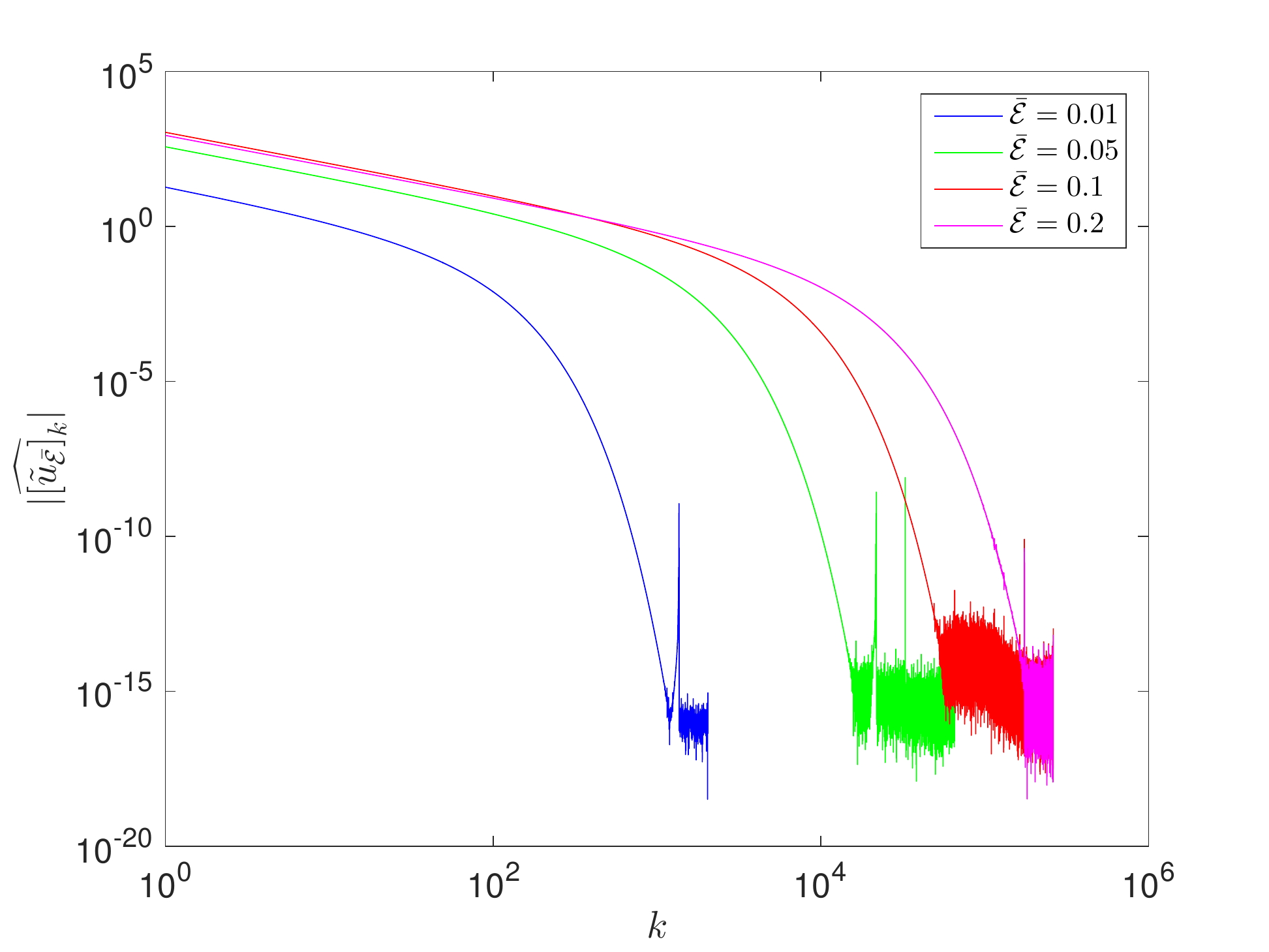}}
\subfigure[][]{\includegraphics[width=3.35in]{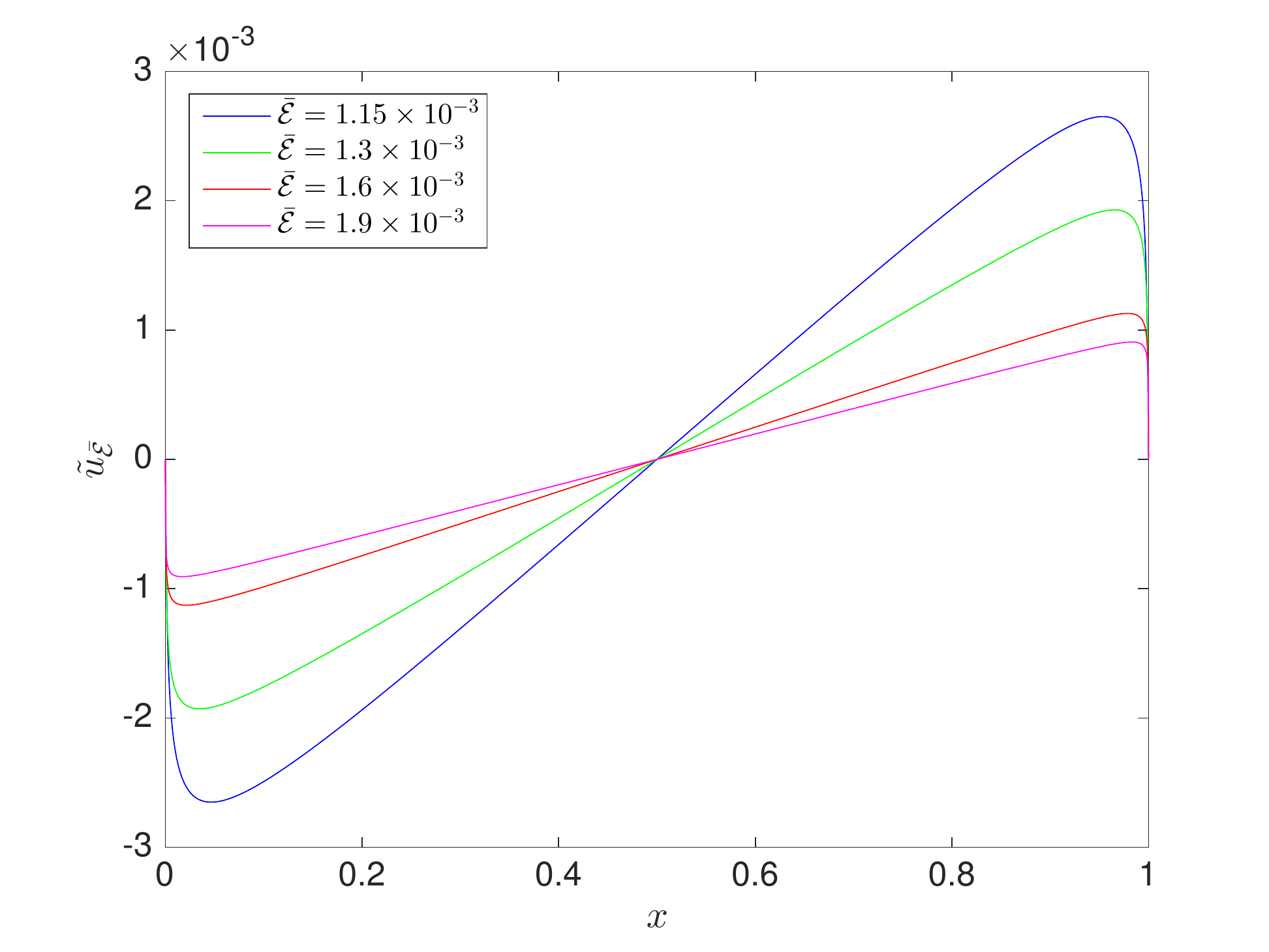}}
\hspace{-0.2in}
\subfigure[][]{\includegraphics[width=3.35in]{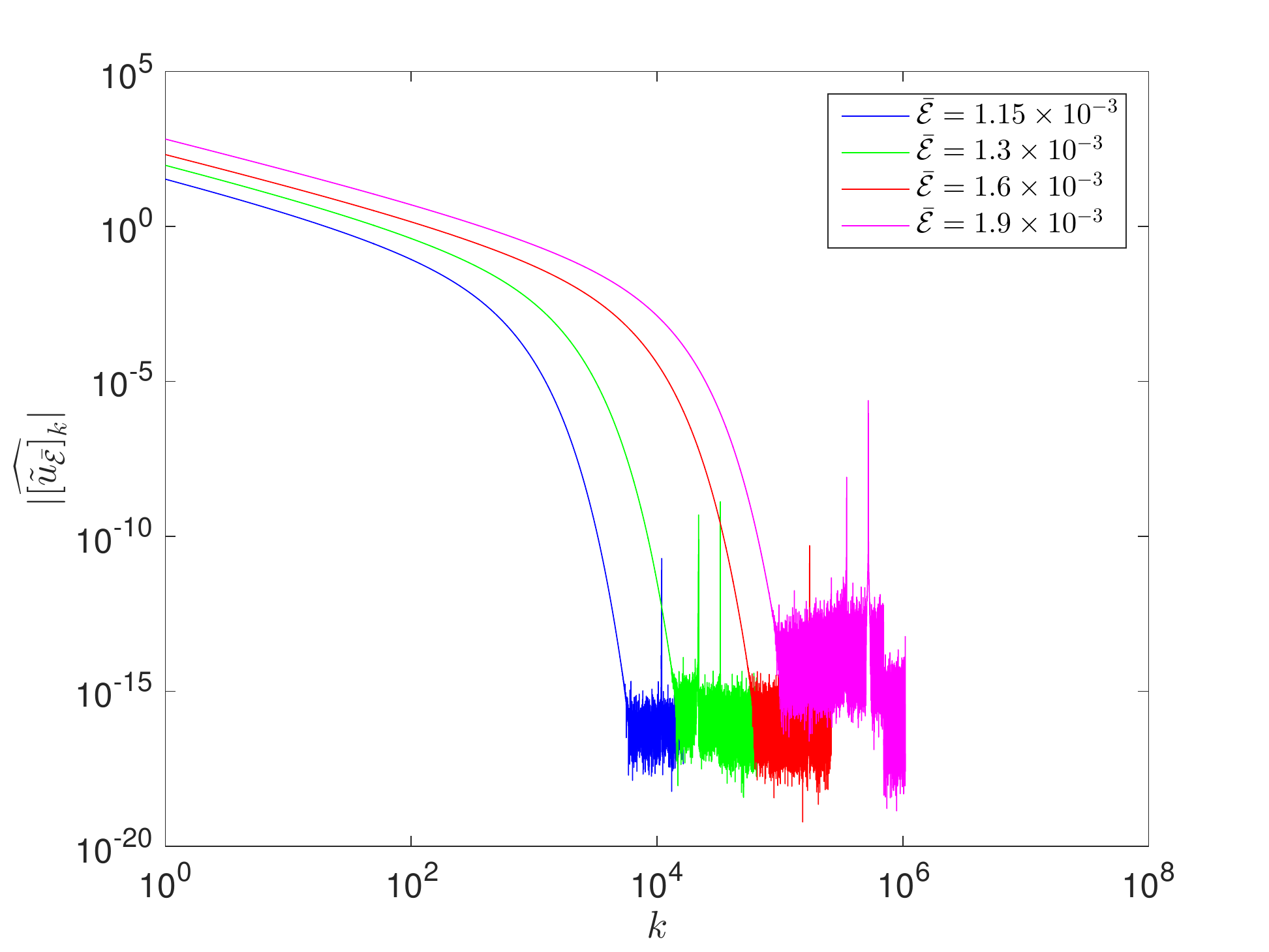}}
\caption{Maximizers {$\tuEbar$} obtained for $\alpha=0.4$ (a,b) and
  $\alpha=0.3$ (c,d) and different values of {$\bar{\E}$}. The fields are
  shown in the physical (a,c) and spectral (b,d) space.}
\label{fig:tuE2}
\end{figure}

\begin{figure}
\centering
\subfigure[][]{\includegraphics[width=3.35in]{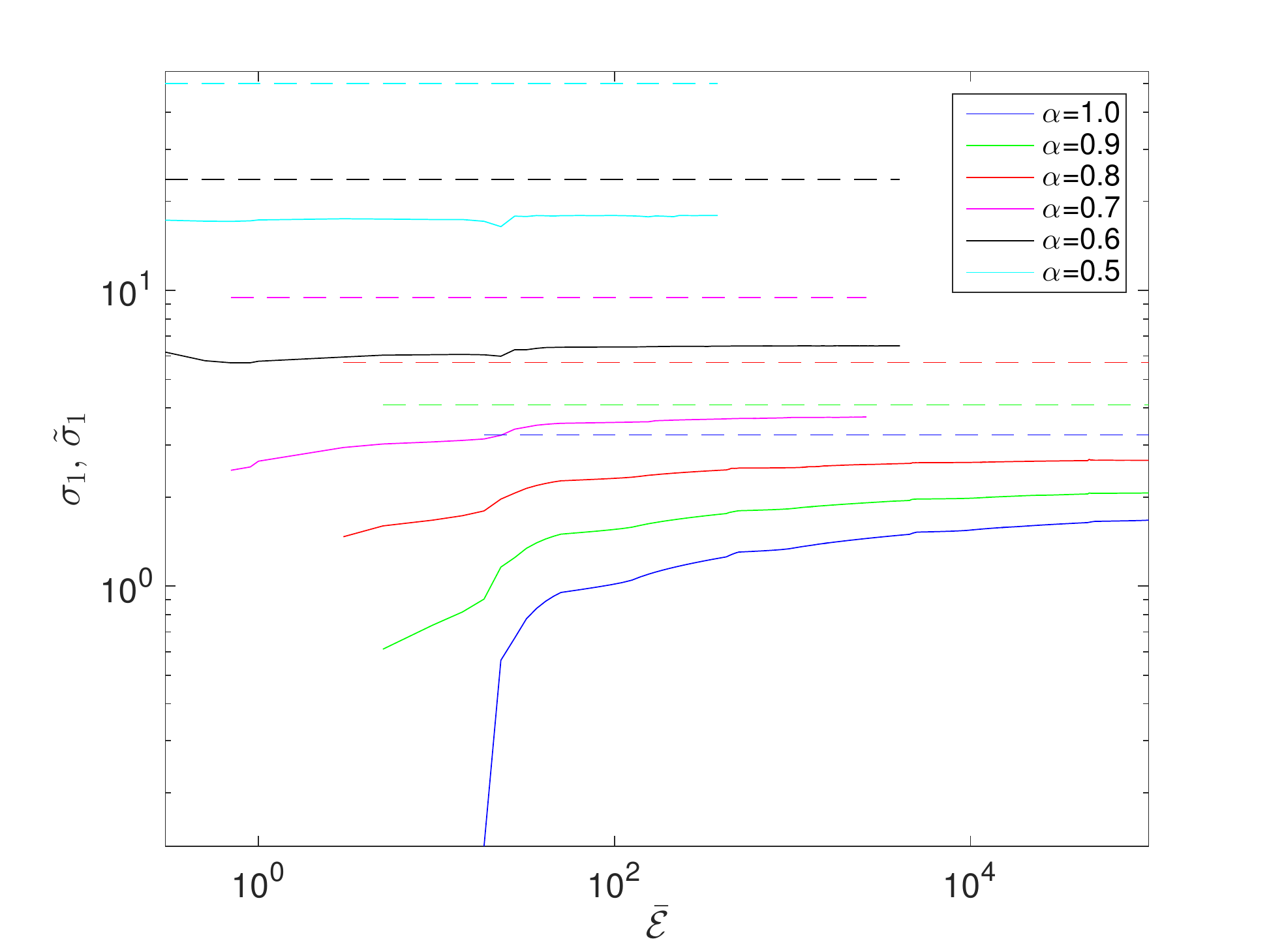}}
\hspace{-0.2in}
\subfigure[][]{\includegraphics[width=3.35in]{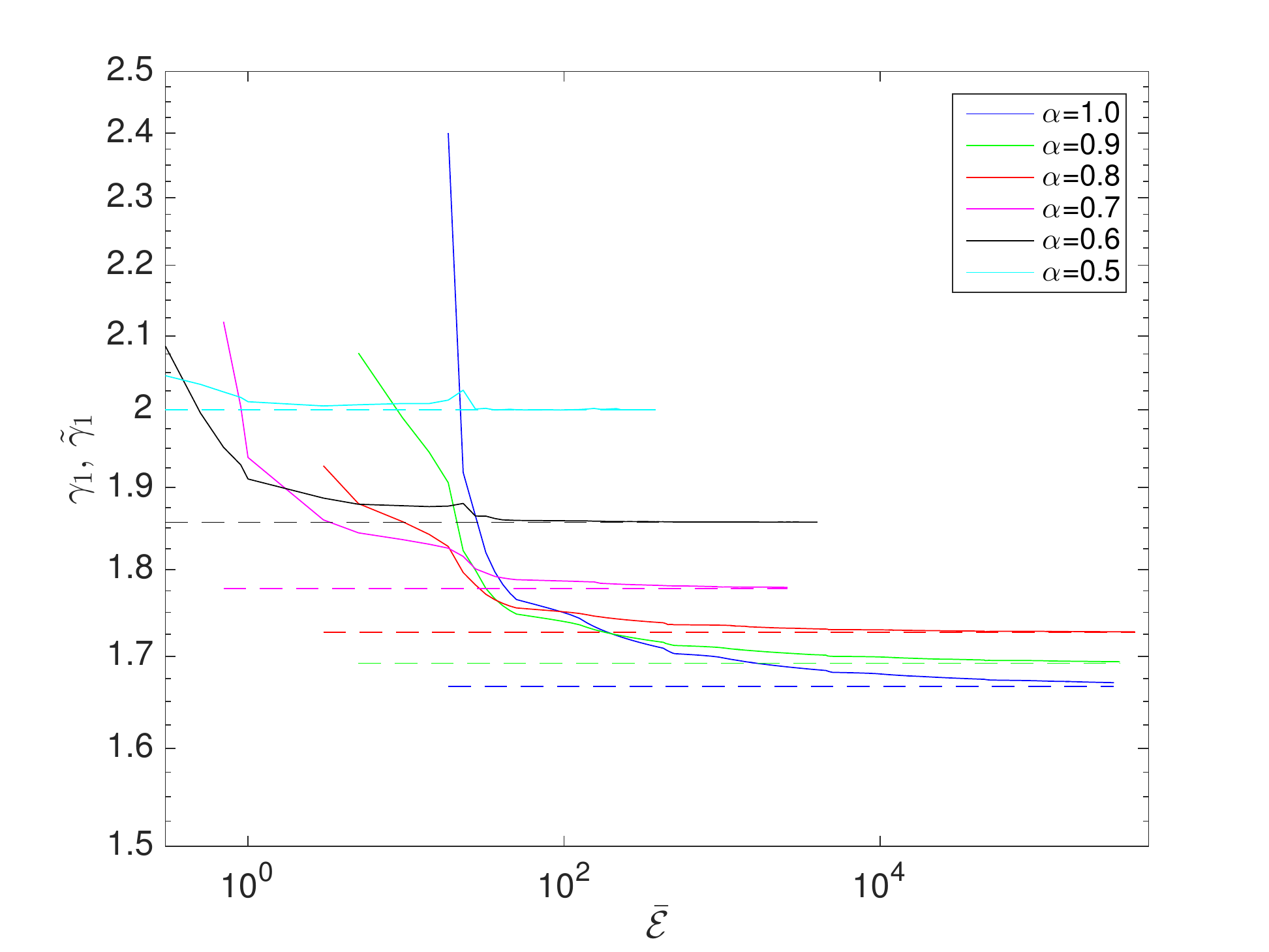}}
\subfigure[][]{\includegraphics[width=3.35in]{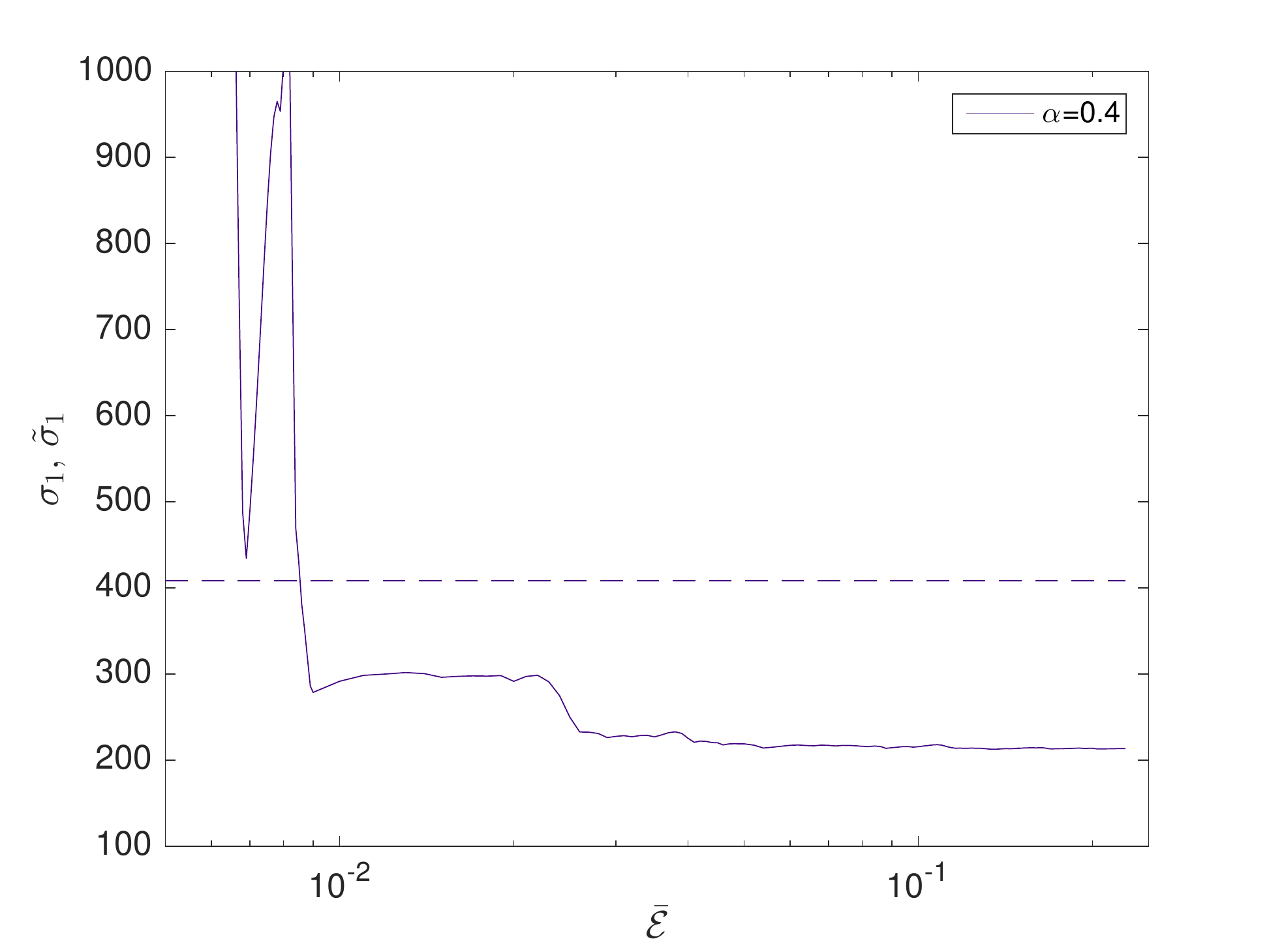}}
\hspace{-0.2in}
\subfigure[][]{\includegraphics[width=3.35in]{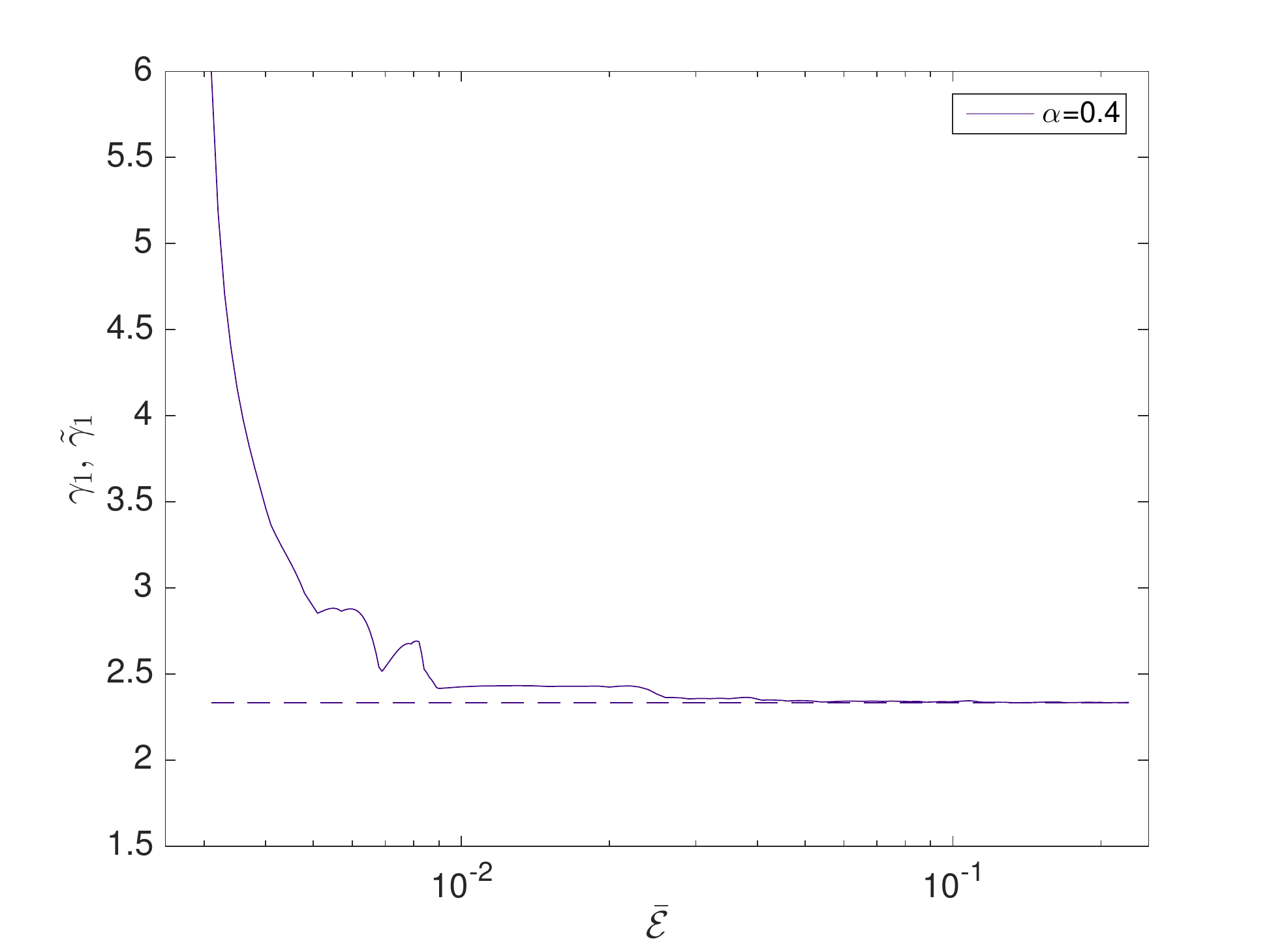}}
\subfigure[][]{\includegraphics[width=3.35in]{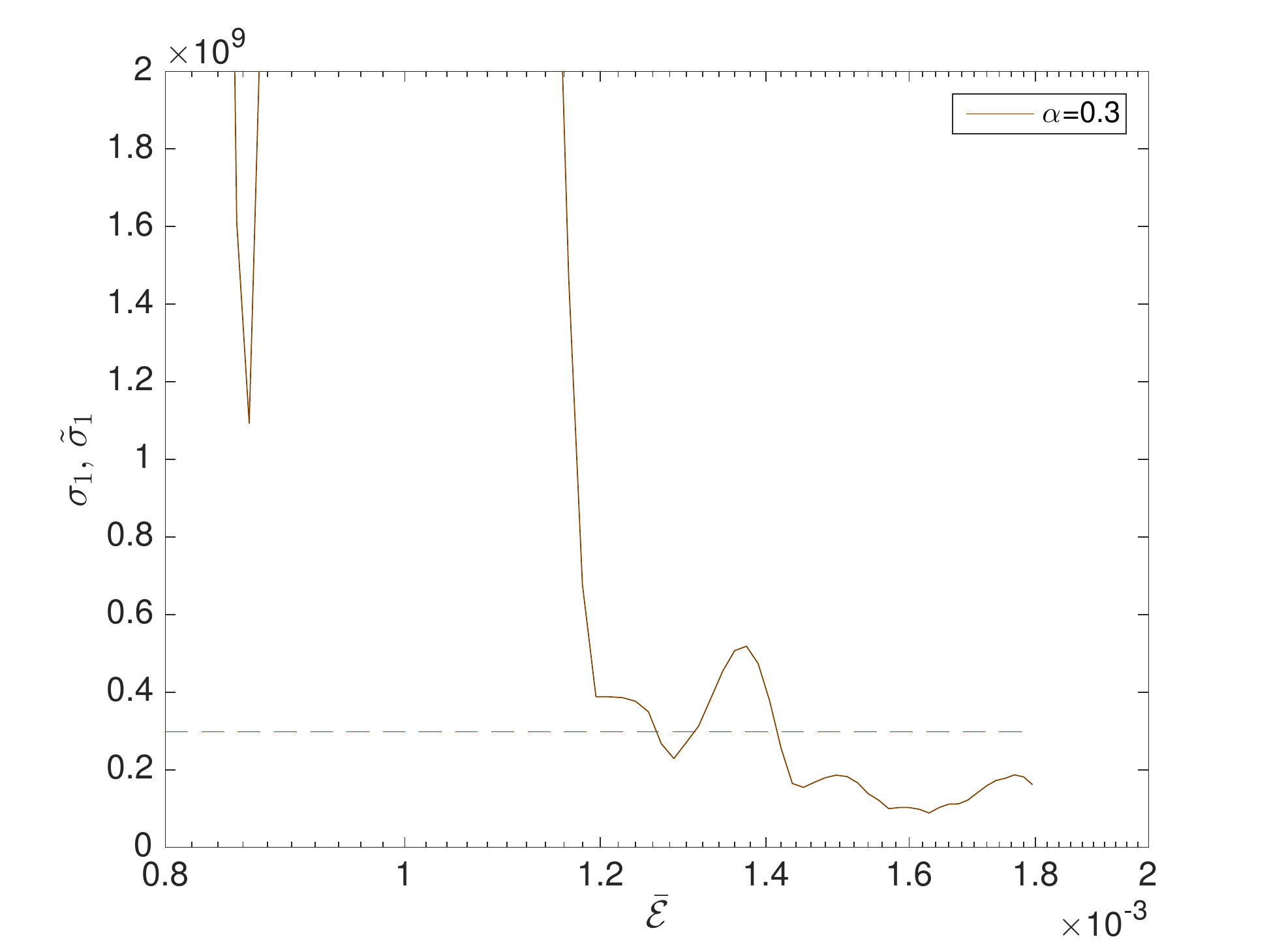}}
\hspace{-0.2in}
\subfigure[][]{\includegraphics[width=3.35in]{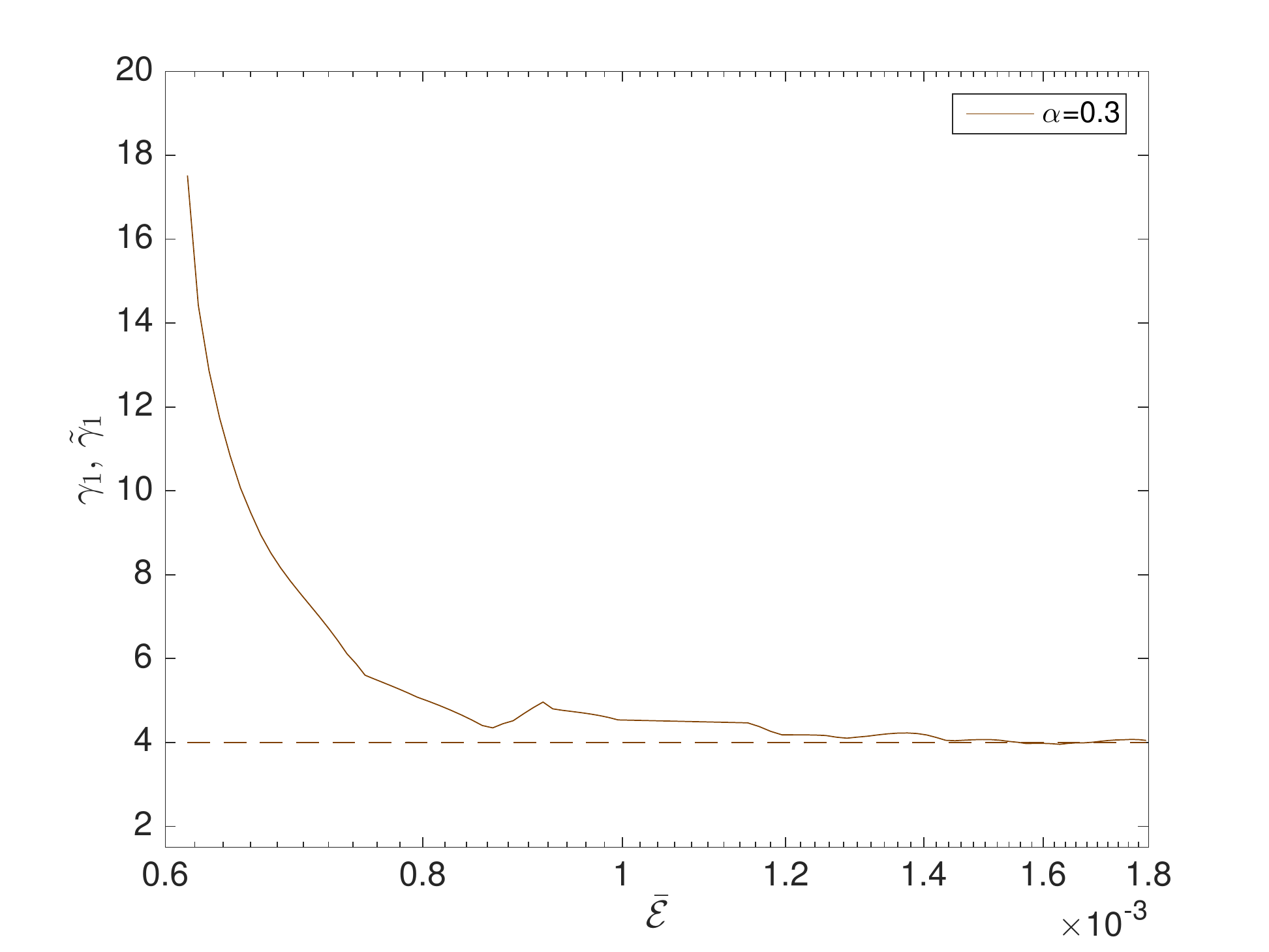}}
\caption{Prefactors $\tsigma_1$ (a,c,e) and exponents $\tgamma_1$
  (b,d,f) obtained as function of {$\bar{\E}$} via local least-squares fits to
  the relations $\R_{\E}({\tuEbar})$ versus {$\bar{\E}$} shown in Figure \ref{fig:R}
  (solid lines).  The dashed lines represent the corresponding
  prefactors $\sigma_1$ and exponents $\gamma_1$ from estimate
  \eqref{eq:dEdt}.}
\label{fig:fit1}
\end{figure}

\begin{figure}
\centering
\subfigure[][]{\includegraphics[width=3.35in]{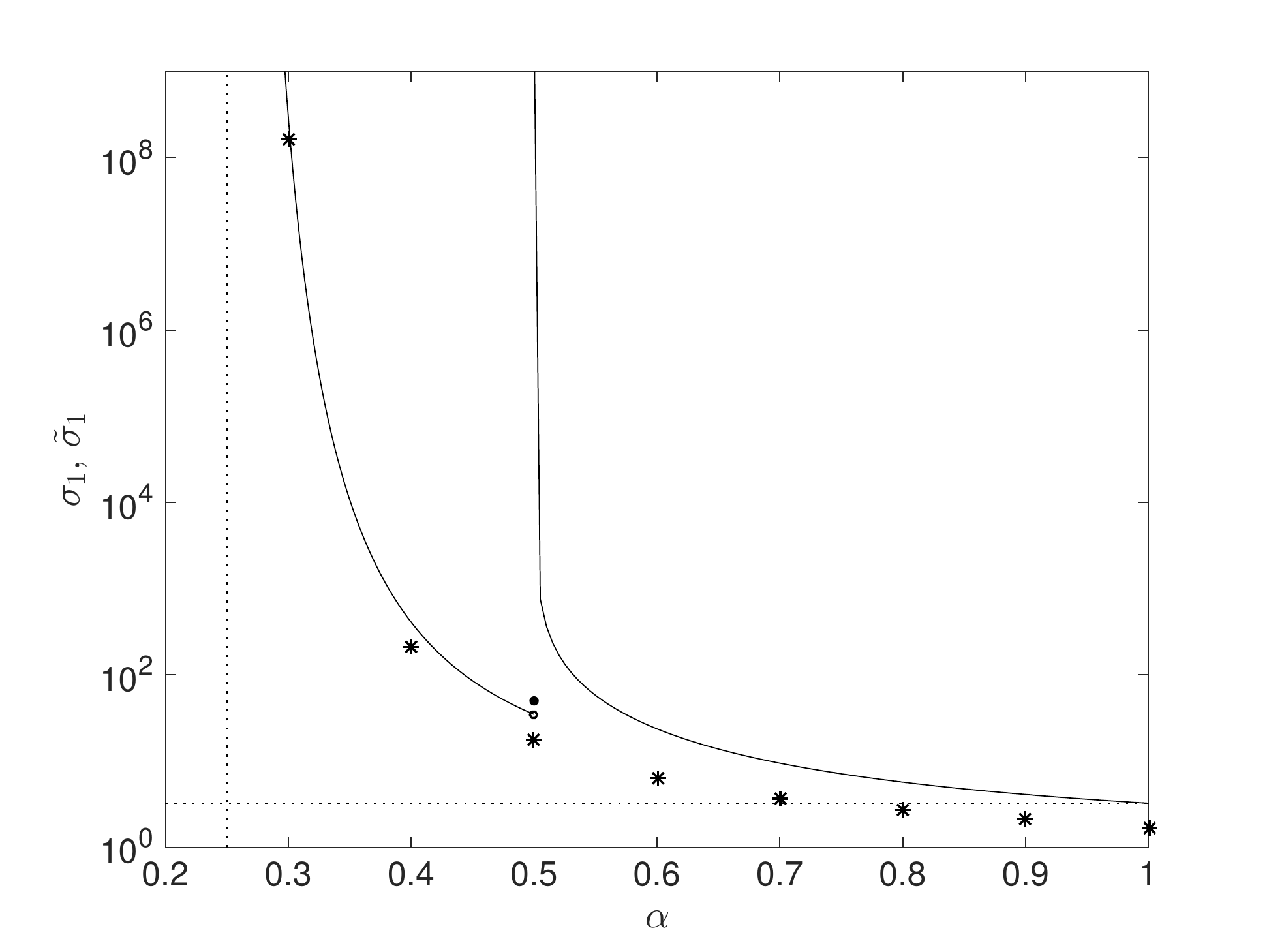}}
\hspace{-0.2in}
\subfigure[][]{\includegraphics[width=3.35in]{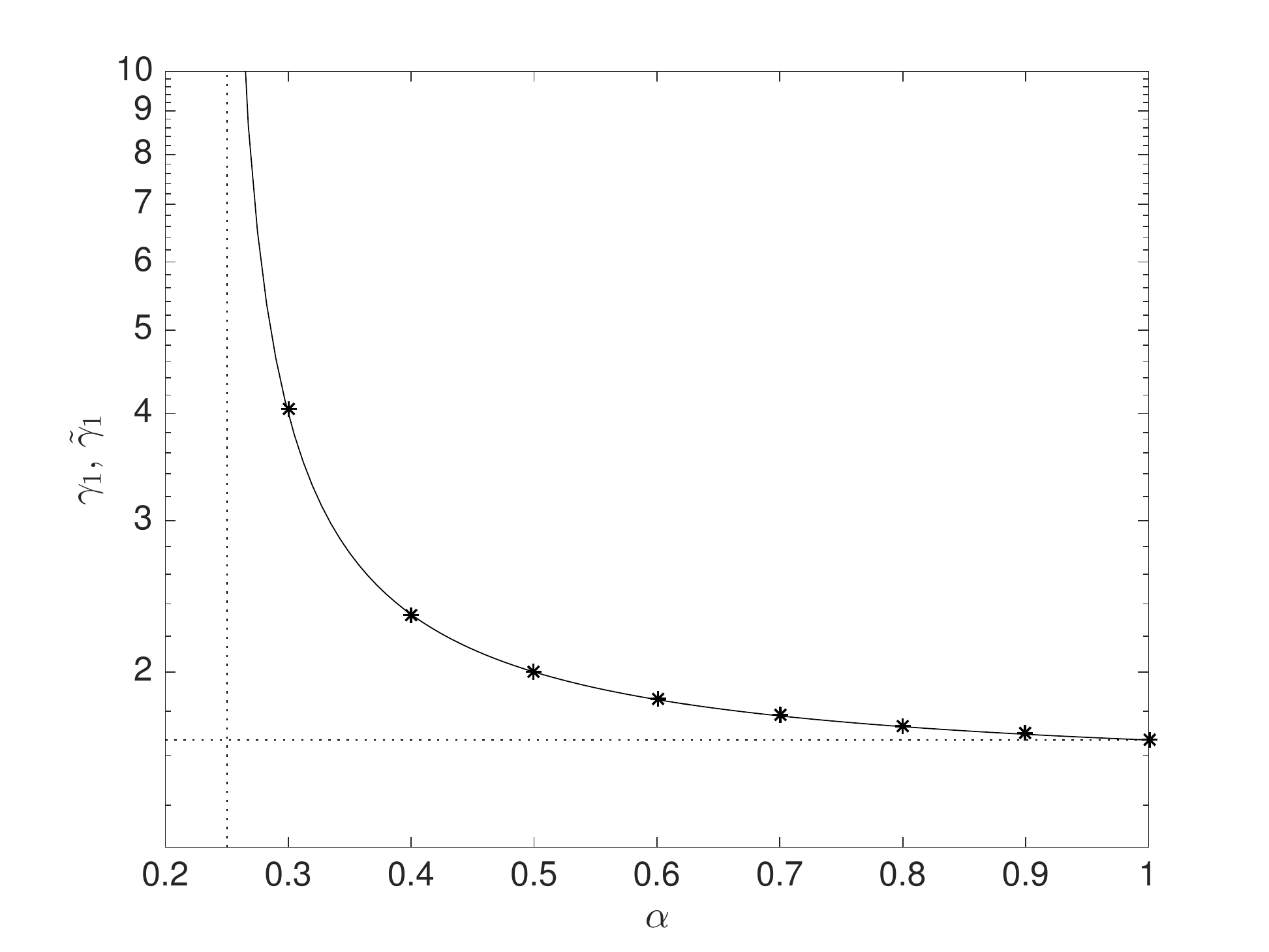}}
\caption{Prefactors (a) and exponents (b) in the power-law relation
 $\tsigma_1 {\bar{\E}}^{\tgamma_1}$ describing the dependence of
  $\R_{\E}({\tuEbar})$ on {$\bar{\E}$} shown as functions of $\alpha$:
  limiting (as ${\bar{\E}} \rightarrow \infty$, cf.~Figure
    \ref{fig:fit1}) values obtained in the least-squares fits
  (symbols) and predictions of estimate \eqref{eq:dEdt} (solid lines).}
\label{fig:sigmagamma1}
\end{figure}

\subsection{Maximum Growth Rate of Fractional Enstrophy}
\label{sec:resultsRa}

While estimate \eqref{eq:dEadt} was established for $\alpha \in
(3/4,1]$, in order to obtain insights about the maximum growth rate of
the fractional enstrophy for a broad range of fractional dissipation
exponents, in this section we solve the maximization problems
\eqref{eq:maxRa} for $\alpha \in [0,1]$. The obtained maximum growth
rate $\R_{\Ea}({\tuEabar})$ is shown as a function of
{$\bar{\E}_\alpha$} in Figures \ref{fig:Ra}(a) and
\ref{fig:Ra}(b) for $\alpha \in (3/4,1]$ and $\alpha \in (1/10,3/4]$,
respectively. In the first figure we also indicate the predictions of
estimate \eqref{eq:dEadt}. For small values of $\alpha$ and
{$\bar{\E}_\alpha$} we also observed the presence of another
branch of maximizers characterized by negative values of
$\R_{\Ea}({\tuEabar})$ --- this data is shown in Figure
\ref{fig:Ra2}(a) using a restricted range of {$\bar{\E}_\alpha$}
for clarity, whereas the positive branches obtained for the same
(small) values of $\alpha$ are presented in Figure \ref{fig:Ra2}(b).
We thus see that for small $\alpha$ and {$\bar{\E}_\alpha$} the
maximization problem \eqref{eq:maxRa} admits two distinct families of
local maximizers. The maximizers {$\tuEabar$} corresponding to
the data in Figures \ref{fig:Ra}(a,b) are shown both in the physical
and spectral space in Figure \ref{fig:tuEa} for
${\bar{\E}_\alpha} = 5,50,500$ and $\alpha = 0.1,0.2,\dots,0.9$.
We observe that, interestingly, as $\alpha$ decreases the sharp fronts
in the maximizers disappear and are replaced with oscillations
(Figures \ref{fig:tuEa}(a,c,e)). This behavior is also reflected in
the spectra of the maximizers which become less developed as $\alpha
\rightarrow 0$ (Figures \ref{fig:tuEa}(b,d,f)). The maximizers
{$\tuEabar$} obtained at the same values of $\alpha$ and
{$\bar{\E}_\alpha$,} and corresponding to the positive and
negative branches of $\R_{\Ea}({\tuEabar})$, cf.~Figures
\ref{fig:Ra2}(a) and \ref{fig:Ra2}(b), are shown in Figure
\ref{fig:tuEa2}. We see that the maximizers for which
$\R_{\Ea}({\tuEabar}) < 0$ have a simpler structure and for all
considered values of $\alpha $ and {$\bar{\E}_\alpha$} are
essentially indistinguishable from $A \sin(2\pi x)$ for some $A>0$.
The relation $\R_{\Ea}({\tuEabar})$ versus
{$\bar{\E}_\alpha$} (the upper branch shown in Figure
\ref{fig:Ra}(a,b)) reveals similar properties as observed in the case
of the classical enstrophy discussed in Section \ref{sec:resultsR},
namely, an initially steep growth followed by saturation with a
power-law behavior when ${\bar{\E}_\alpha} \rightarrow \infty$.
Performing local least-squares fits to these relations with the
formula $\tsigma_{\alpha} {\bar{\E}_\alpha}^{\tgamma_{\alpha}}$,
as described in Section \ref{sec:resultsR}, we can calculate how the
actual prefactors $\tsigma_{\alpha}$ and exponents $\tgamma_{\alpha}$
depend on {$\bar{\E}_\alpha$} and these results are shown in
Figure \ref{fig:fita}(a,b), where we have also indicated, for $\alpha
\in (3/4,1]$, the predictions of estimate \eqref{eq:dEadt}. We see
that, as ${\bar{\E}_\alpha} \rightarrow \infty$, both
$\tsigma_{\alpha}$ and $\tgamma_{\alpha}$ approach well-defined
values, which are in turn shown in Figure \ref{fig:sigmagammaa}(a,b)
as functions of $\alpha$ together with the prefactors and the
exponents obtained in estimate \eqref{eq:dEadt}. We see in Figure
\ref{fig:sigmagammaa}(b) that for $0.9 \lessapprox \alpha \lessapprox
1$ the numerically obtained exponents $\tgamma_{\alpha}$ match the
exponents $\gamma_{\alpha}$ from estimate \eqref{eq:dEadt} and a
difference appears for $0.8 \lessapprox \alpha \lessapprox 0.9$ which
grows as $\alpha$ decreases. For $0.7 \lessapprox \alpha \lessapprox
0.8$ the numerically determined exponents $\tgamma_{\alpha}$ are a
decreasing function of $\alpha$ which saturates at a constant value of
approximately $3/2$ when $\alpha \lessapprox 0.7$. At the same time,
the prefactors $\tsigma_{\alpha}$ obtained numerically are by a few
orders of magnitude smaller than the prefactors predicted by estimate
\eqref{eq:dEadt} over the entire range of $\alpha$, although they do
exhibit qualitatively similar trends with $\alpha$. For $\alpha \in
[0,3/4]$, which is outside the range of validity of estimate
\eqref{eq:dEadt}, the numerically obtained exponents
$\tgamma_{\alpha}$ are constant, indicating that, somewhat
surprisingly, in this range $\R_{\Ea}({\tuEabar})$ does not
depend on the fractional dissipation exponent $\alpha$. The
corresponding numerically obtained prefactors $\tsigma_{\alpha}$
reveal a decreasing trend with $\alpha$. We add that these trends are
accompanied by the maximizers {$\tuEabar$} becoming more regular
as $\alpha$ decreases (cf.~Figure \ref{fig:tuEa}).  We thus conclude
that the exponent $\gamma_{\alpha}$ in estimate \eqref{eq:dEadt} is
sharp over a part of the range of validity of this estimate and
appears to overestimate the actual rate of growth of fractional
enstrophy for smaller values of $\alpha$. Over the range of $\alpha$
where the exponent $\gamma_{\alpha}$ is sharp, the prefactor
$\sigma_{\alpha}$ may be improved.

\begin{figure}
\centering
\subfigure[][]{\includegraphics[width=3.35in]{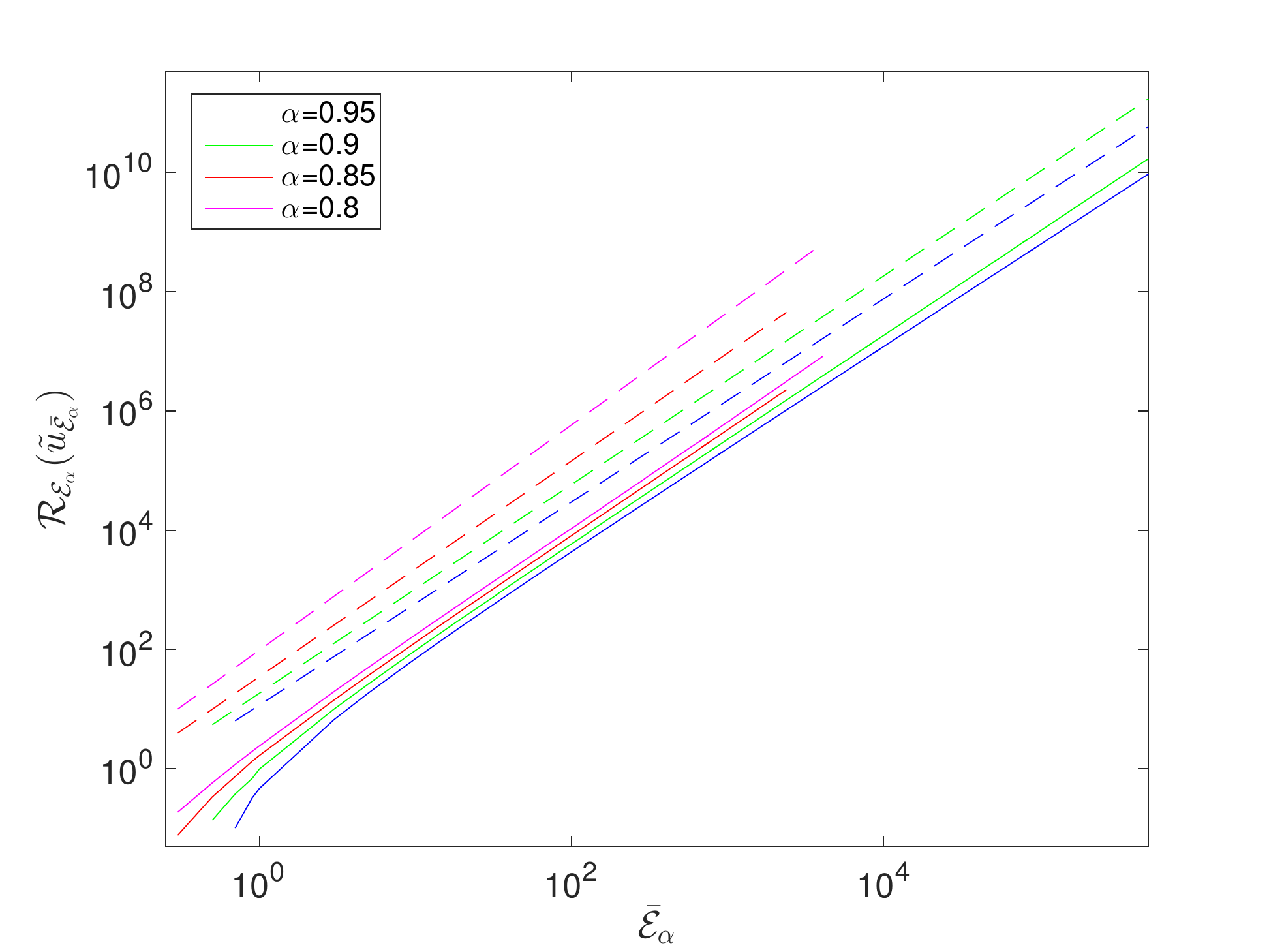}}
\hspace{-0.2in}
\subfigure[][]{\includegraphics[width=3.35in]{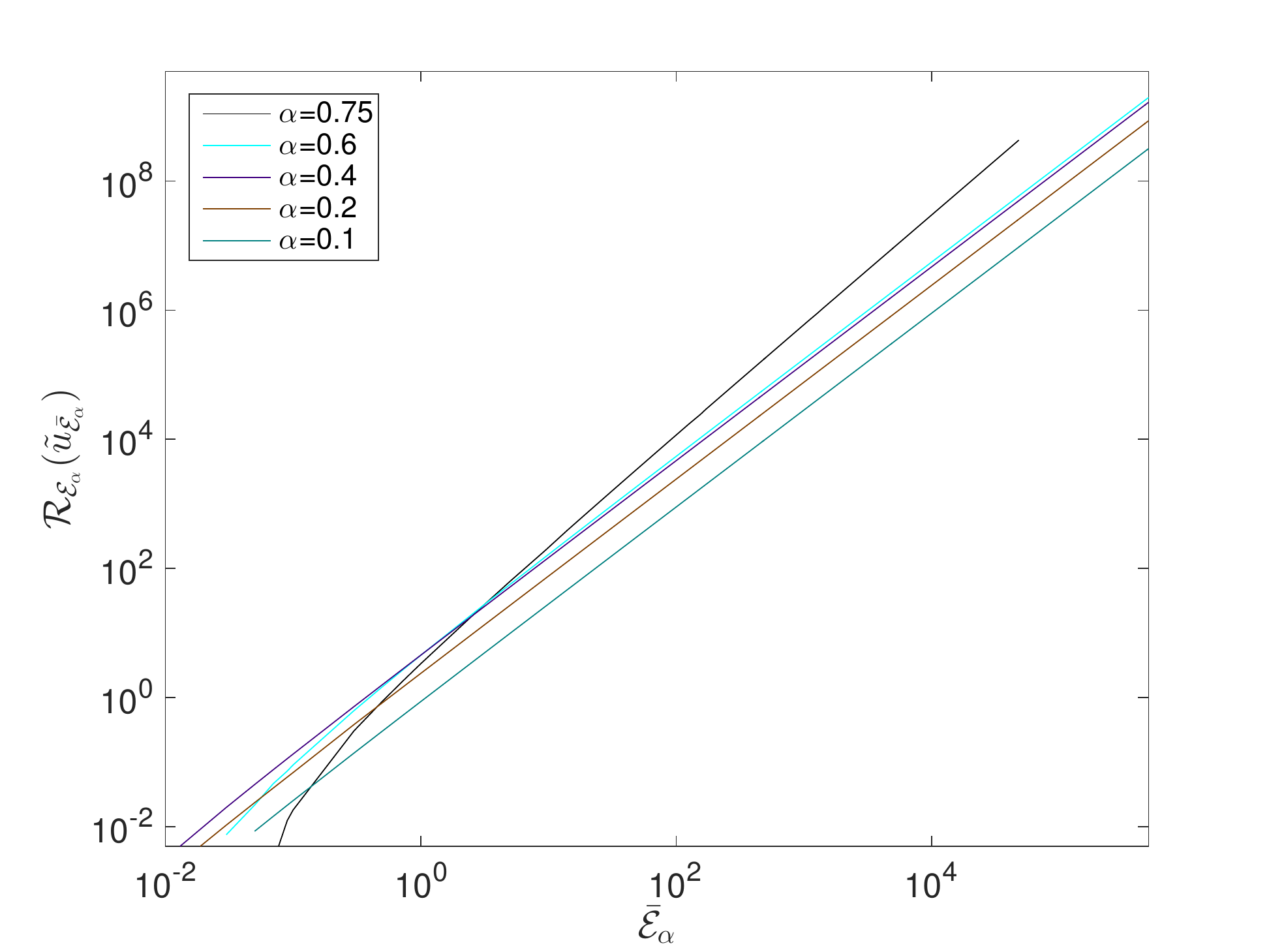}}
\caption{Dependence of the maximum fractional enstrophy rate of growth
  $\R_{\Ea}({\tuEabar})$, obtained by solving optimization problems
  \eqref{eq:maxRa}, on {$\bar{\E}_\alpha$} (solid lines) for $\alpha \in (3/4,1]$ (a)
  and $\alpha \in (1/10,3/4]$ (b). The dashed lines in panel (a)
  represent the corresponding upper bounds from estimate
  \eqref{eq:dEadt}.}
\label{fig:Ra}
\setcounter{subfigure}{0}
\centering
\subfigure[][]{\includegraphics[width=3.35in]{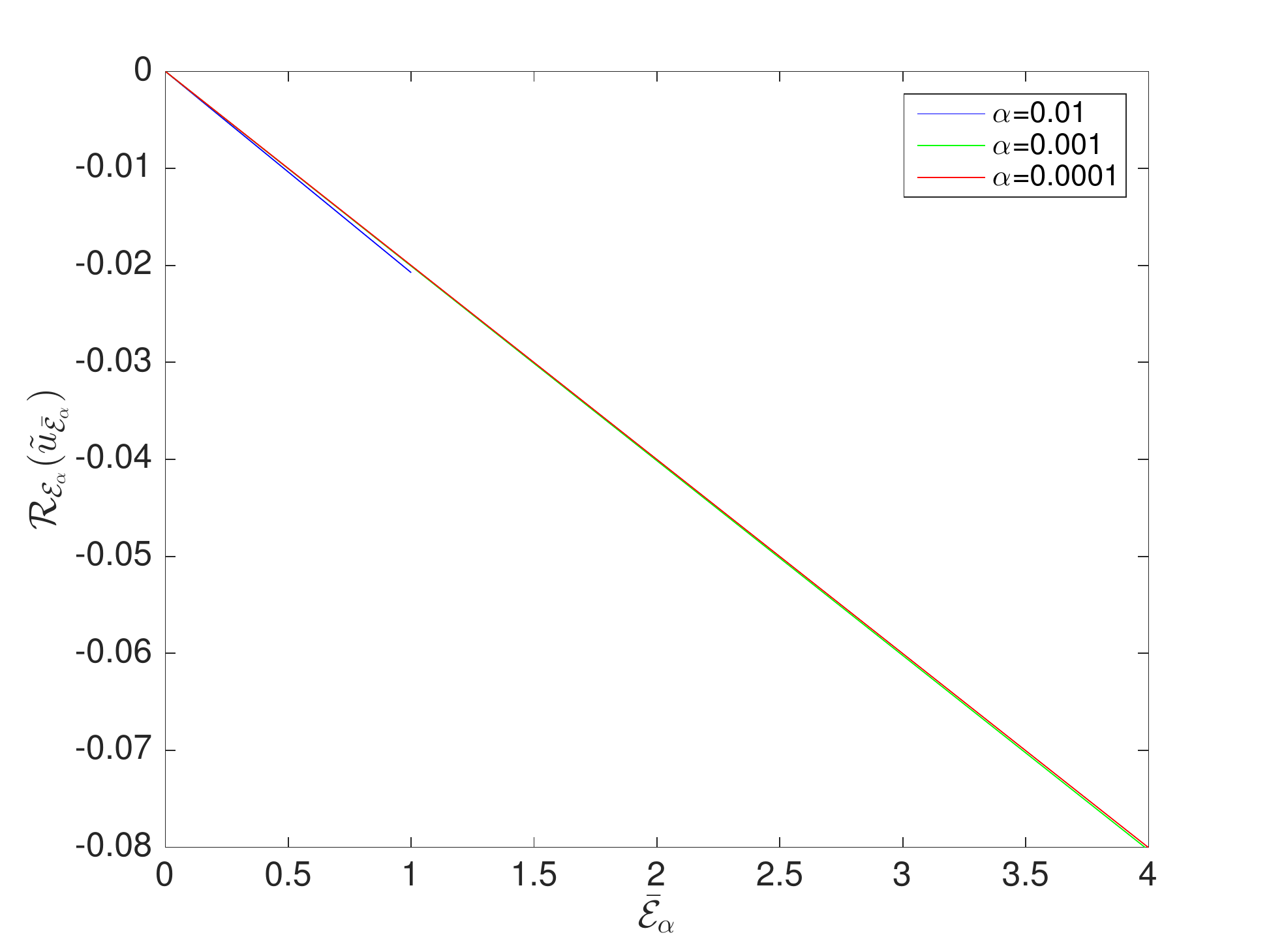}}
\hspace{-0.2in}
\subfigure[][]{\includegraphics[width=3.35in]{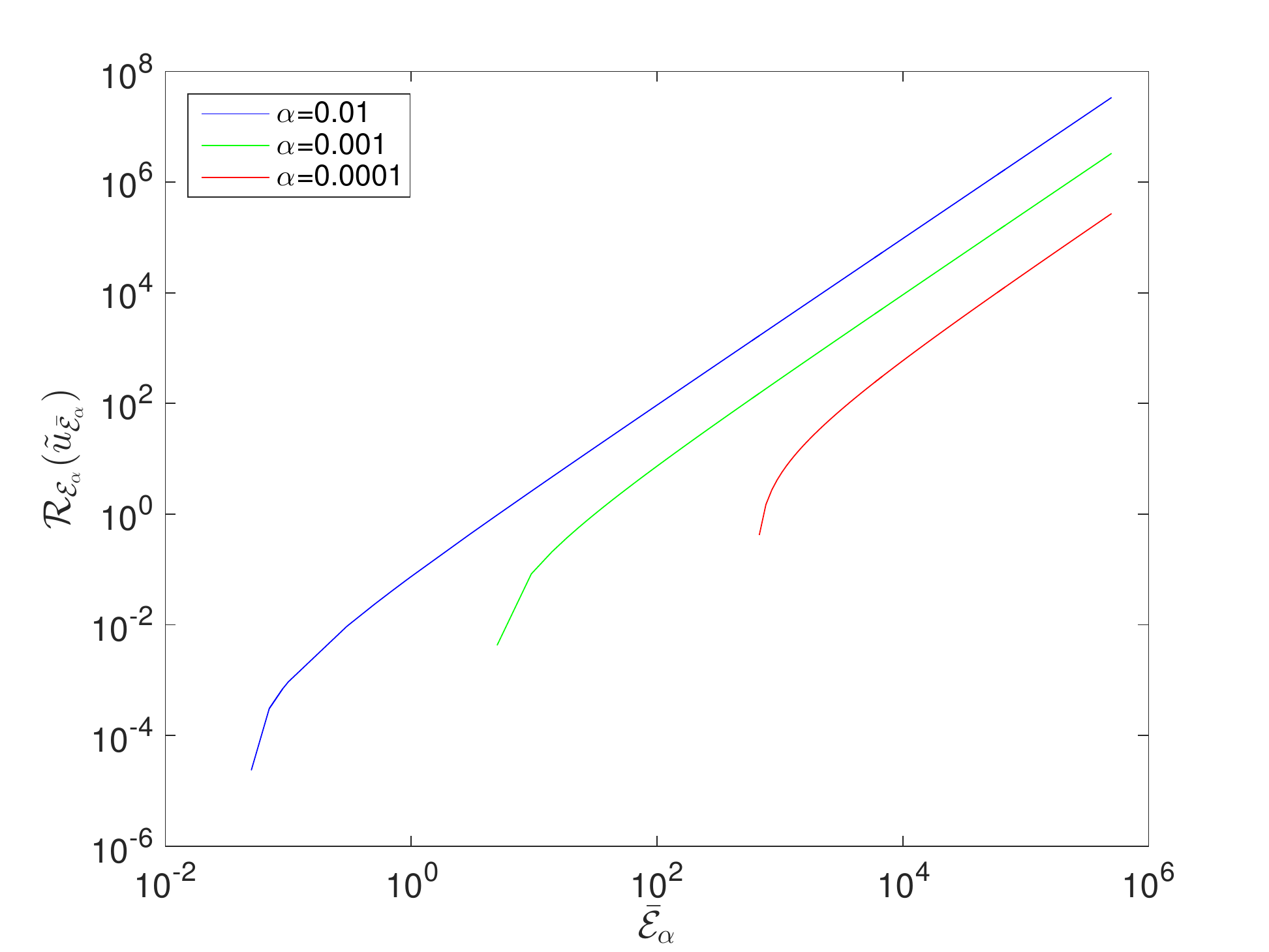}}
\caption{Dependence of the maximum fractional enstrophy rate of growth
  $\R_{\Ea}({\tuEabar})$, obtained by solving optimization problems
  \eqref{eq:maxRa}, on {$\bar{\E}_\alpha$} for small values of $\alpha$: negative
  branch (a) and positive branch (b).}
\label{fig:Ra2}
\end{figure}

\begin{figure}
\centering
\subfigure[][]{\includegraphics[width=3.35in]{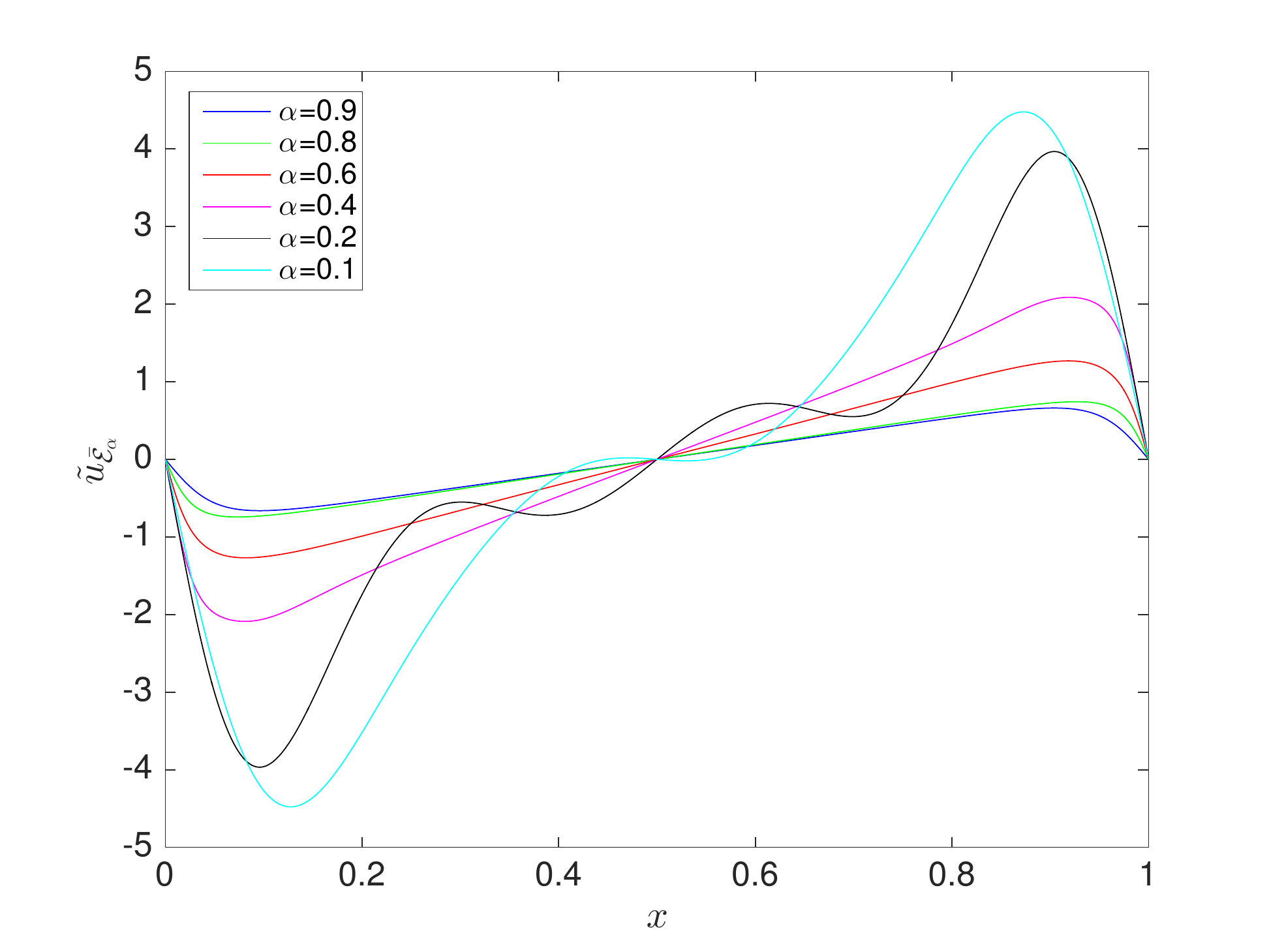}}
\hspace{-0.2in}
\subfigure[][]{\includegraphics[width=3.35in]{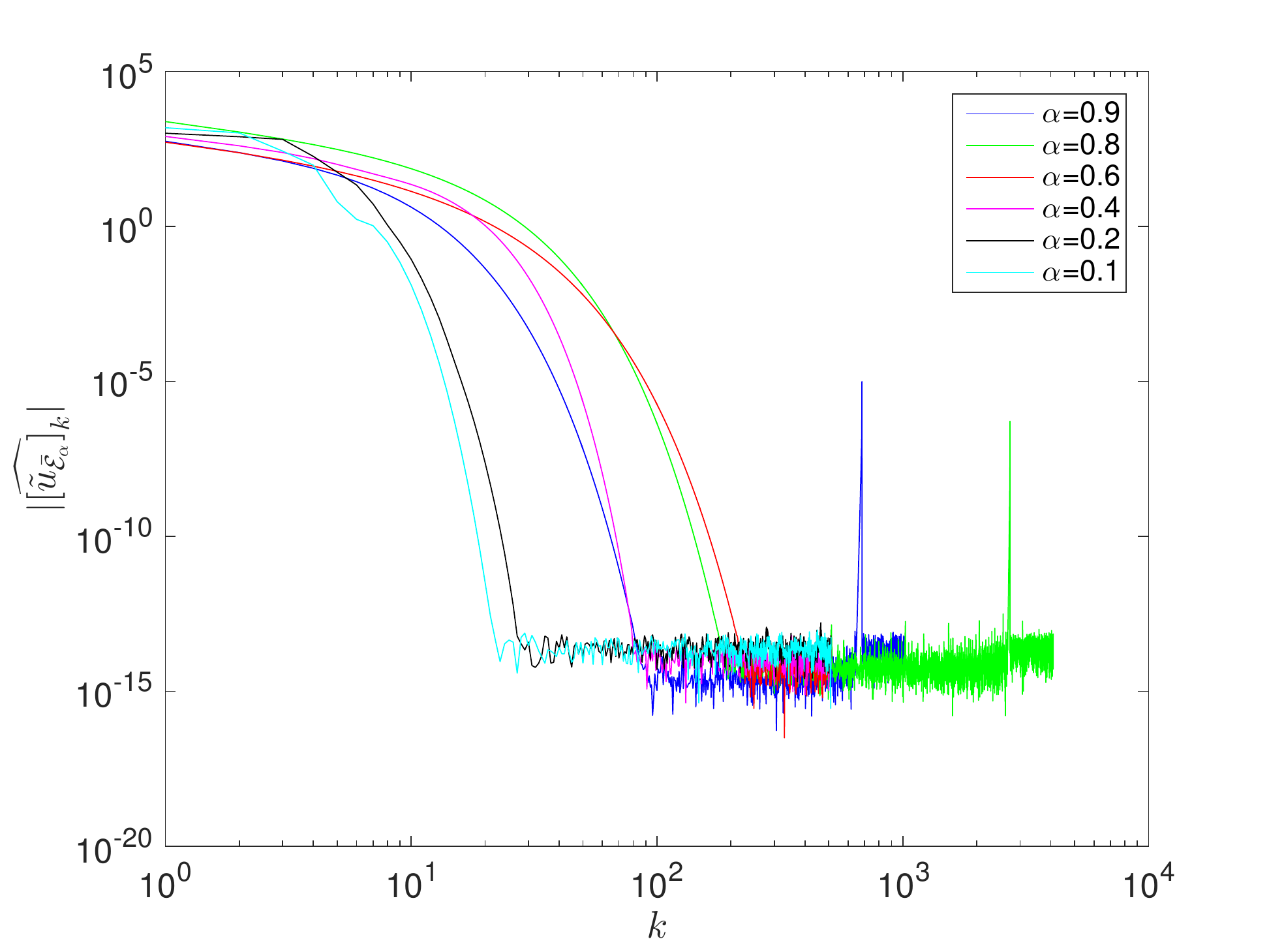}}
\subfigure[][]{\includegraphics[width=3.35in]{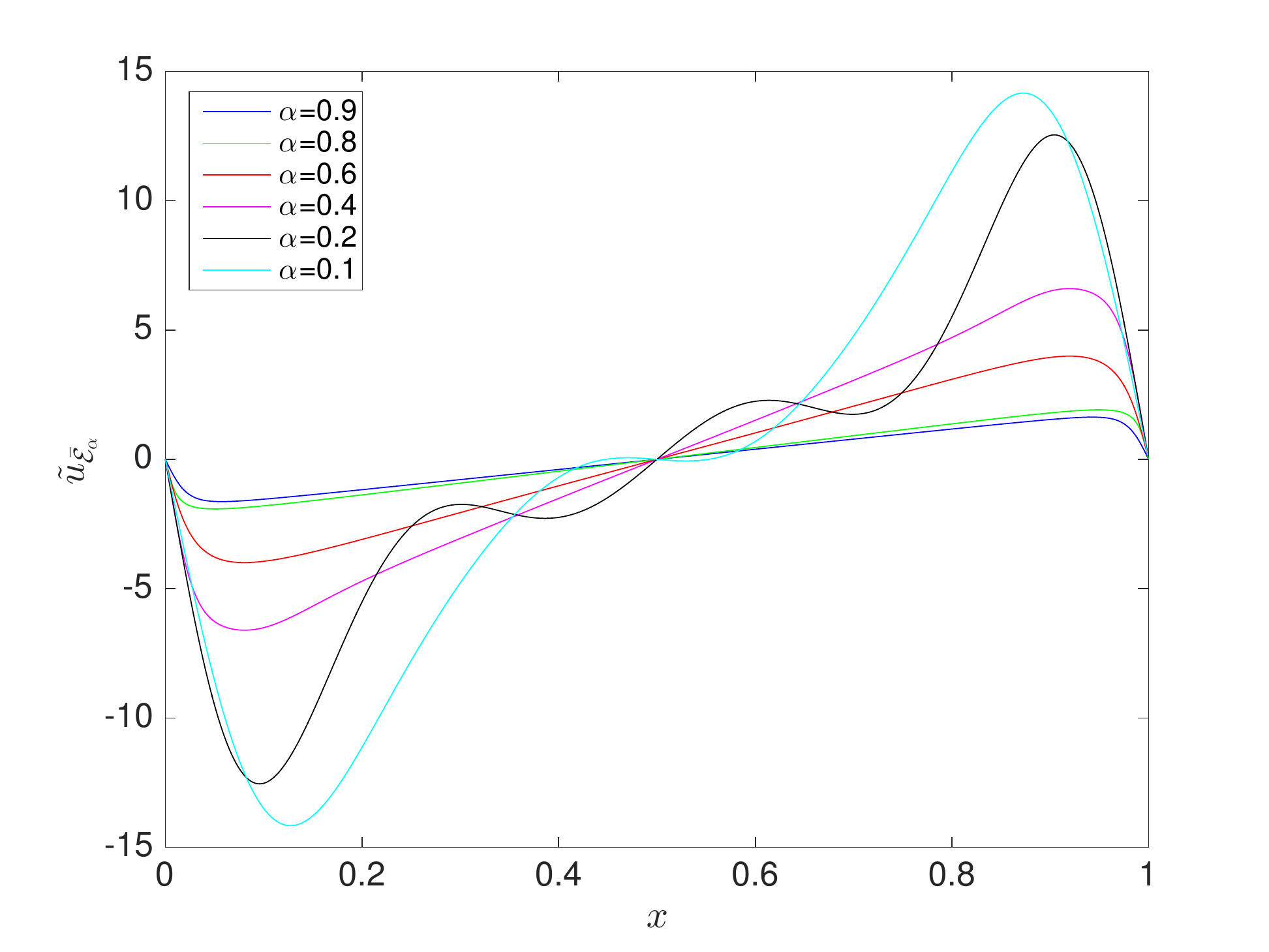}}
\hspace{-0.2in}
\subfigure[][]{\includegraphics[width=3.35in]{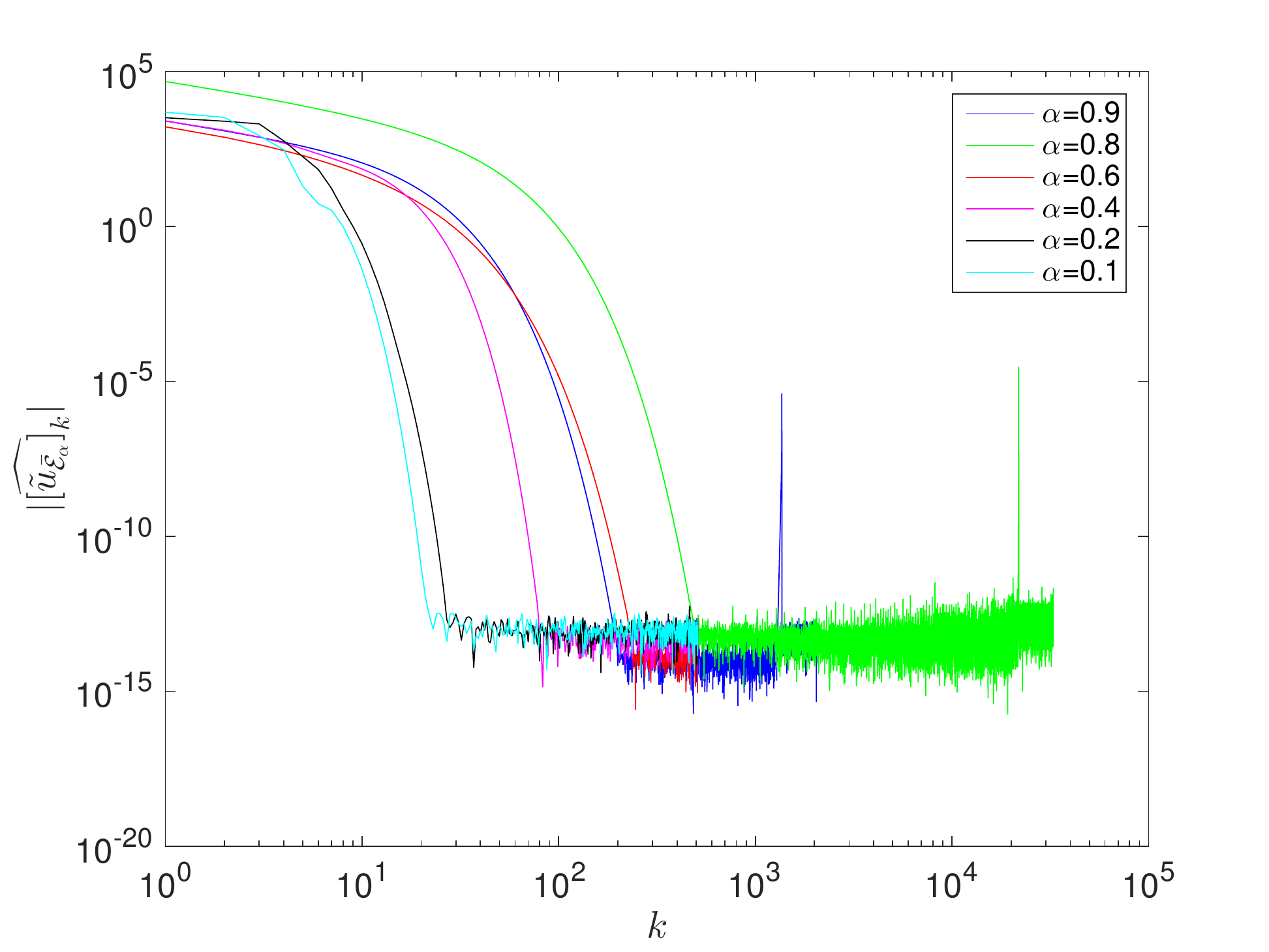}}
\subfigure[][]{\includegraphics[width=3.35in]{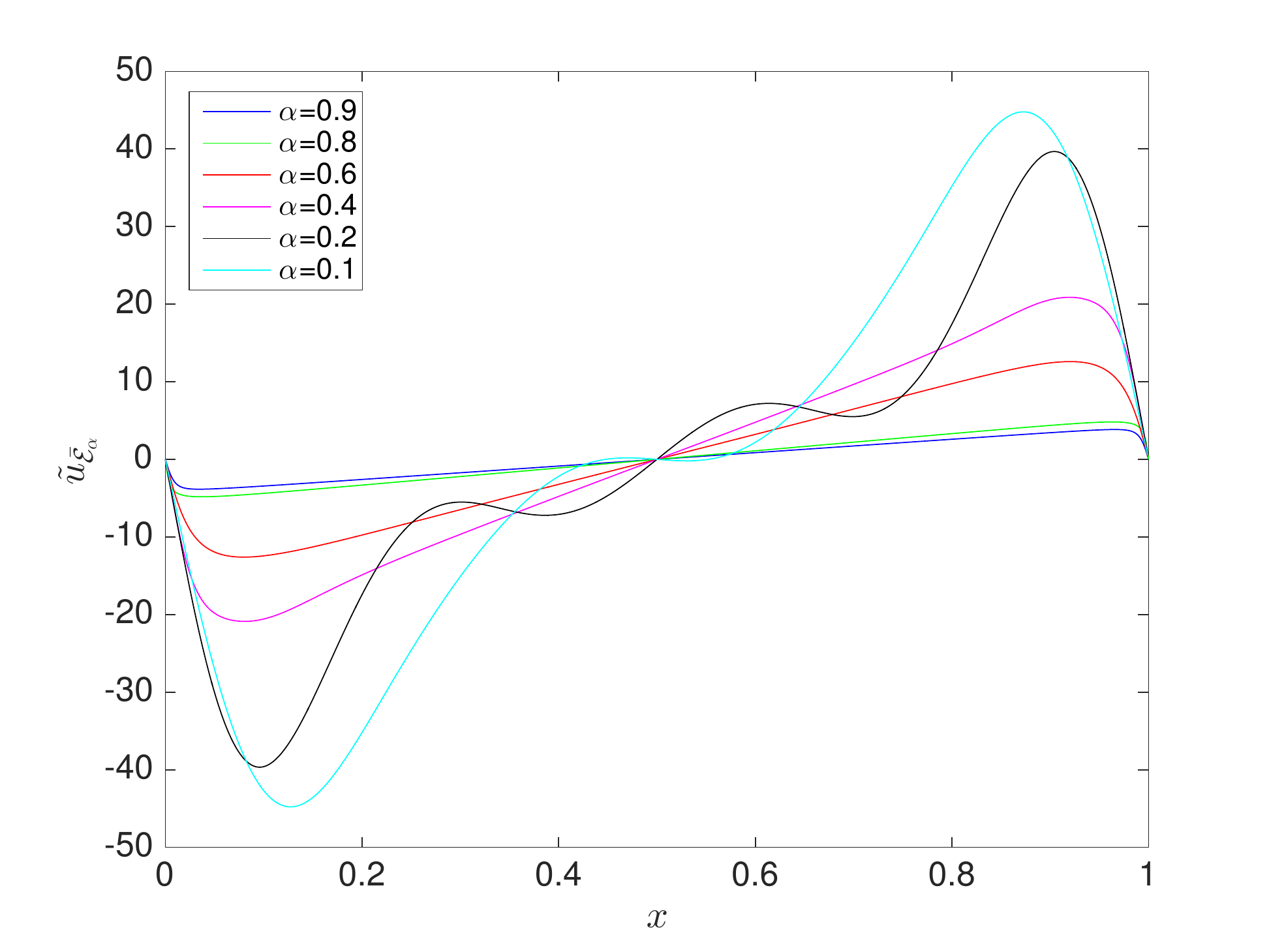}}
\hspace{-0.2in}
\subfigure[][]{\includegraphics[width=3.35in]{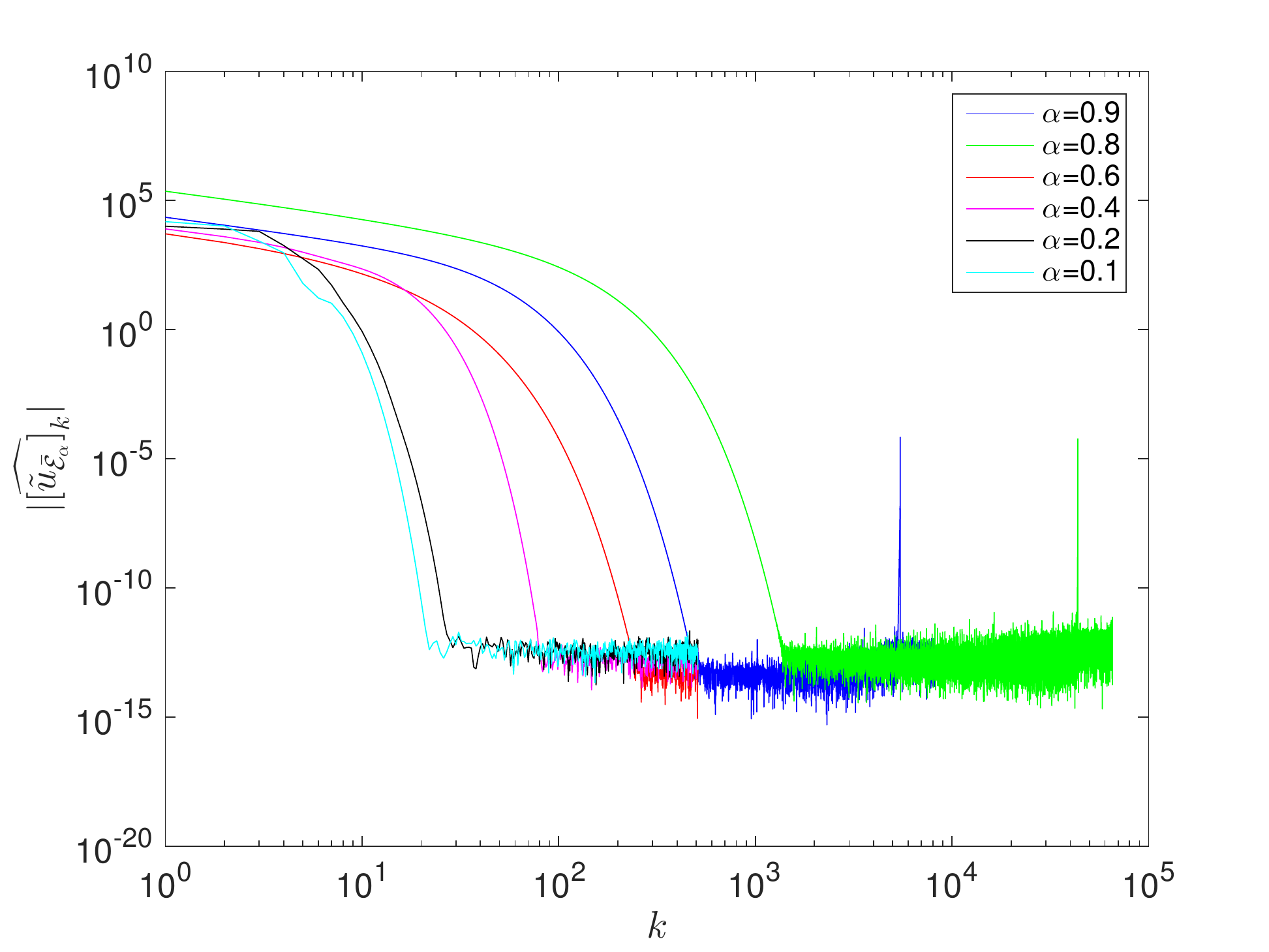}}
\caption{Maximizers {$\tuEabar$} obtained for ${\bar{\E}_\alpha}=5$ (a,b), ${\bar{\E}_\alpha}=50$ (c,d)
  and ${\bar{\E}_\alpha}=500$ (e,f) and different values of $\alpha$. The fields are
  shown in the physical (a,c,e) and spectral (b,d,f) space.}
\label{fig:tuEa}
\end{figure}

\begin{figure}
\centering
\subfigure[][]{\includegraphics[width=2.3in]{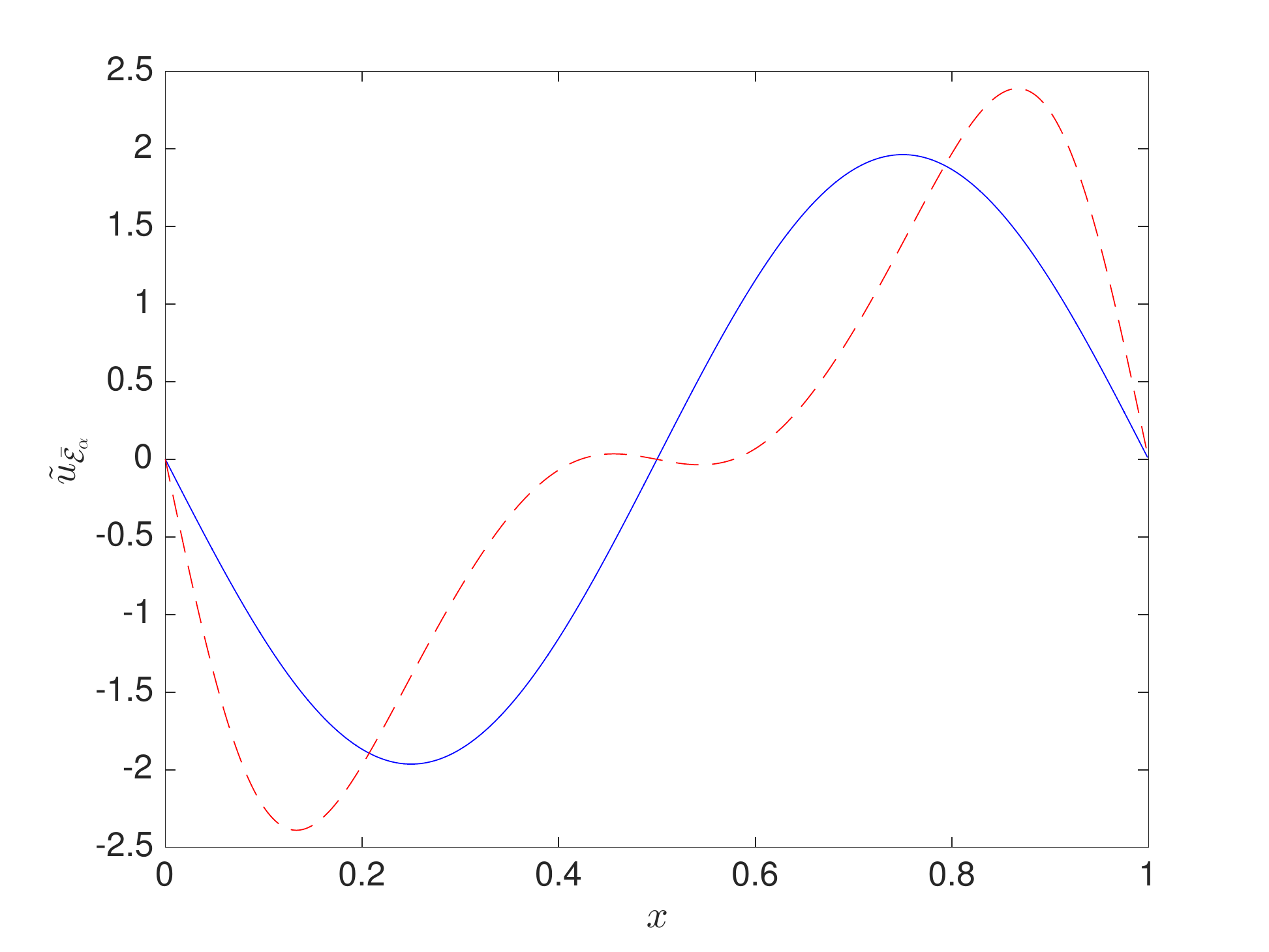}}
\hspace{-0.2in}
\subfigure[][]{\includegraphics[width=2.3in]{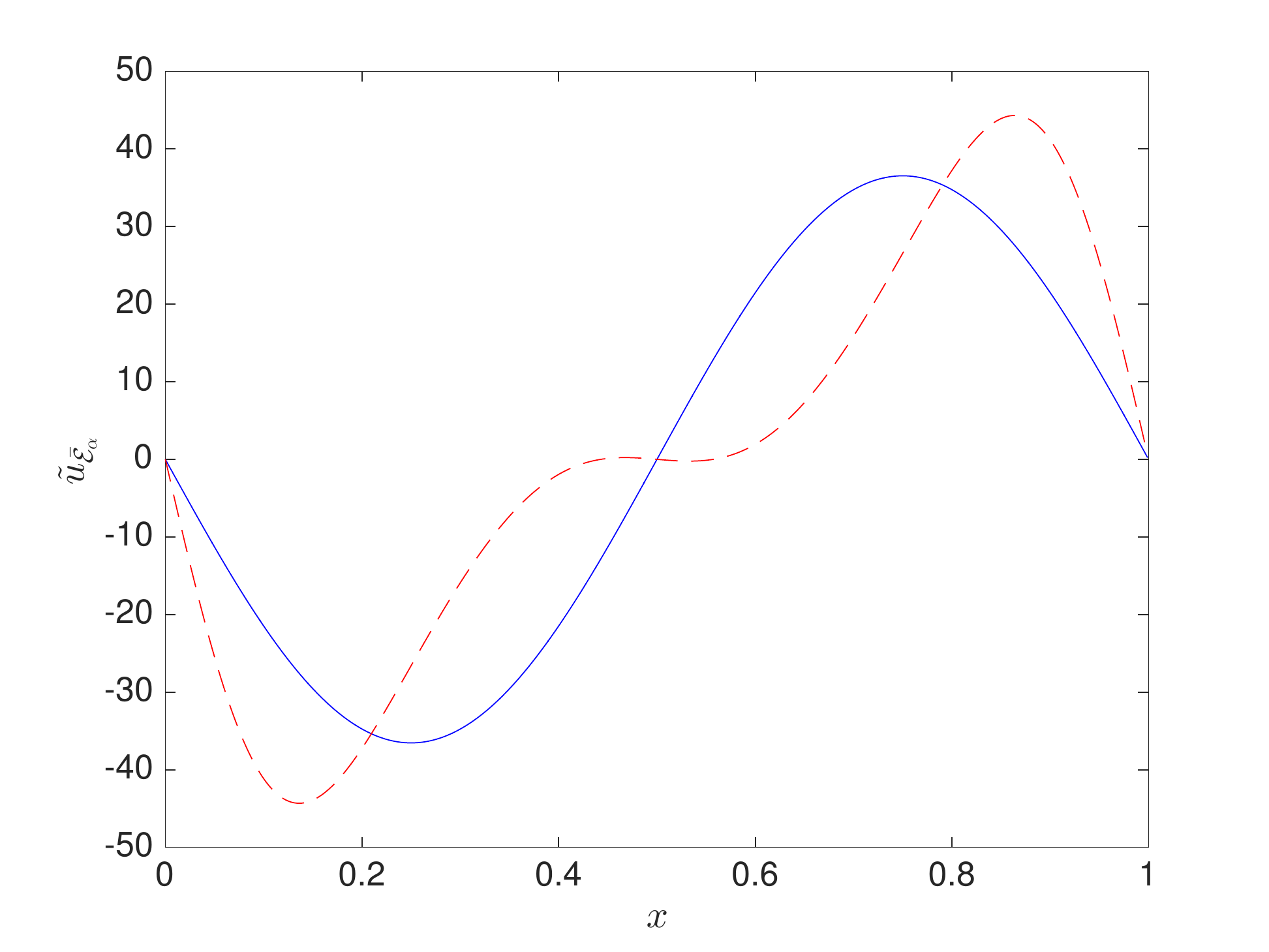}}
\hspace{-0.2in}
\subfigure[][]{\includegraphics[width=2.3in]{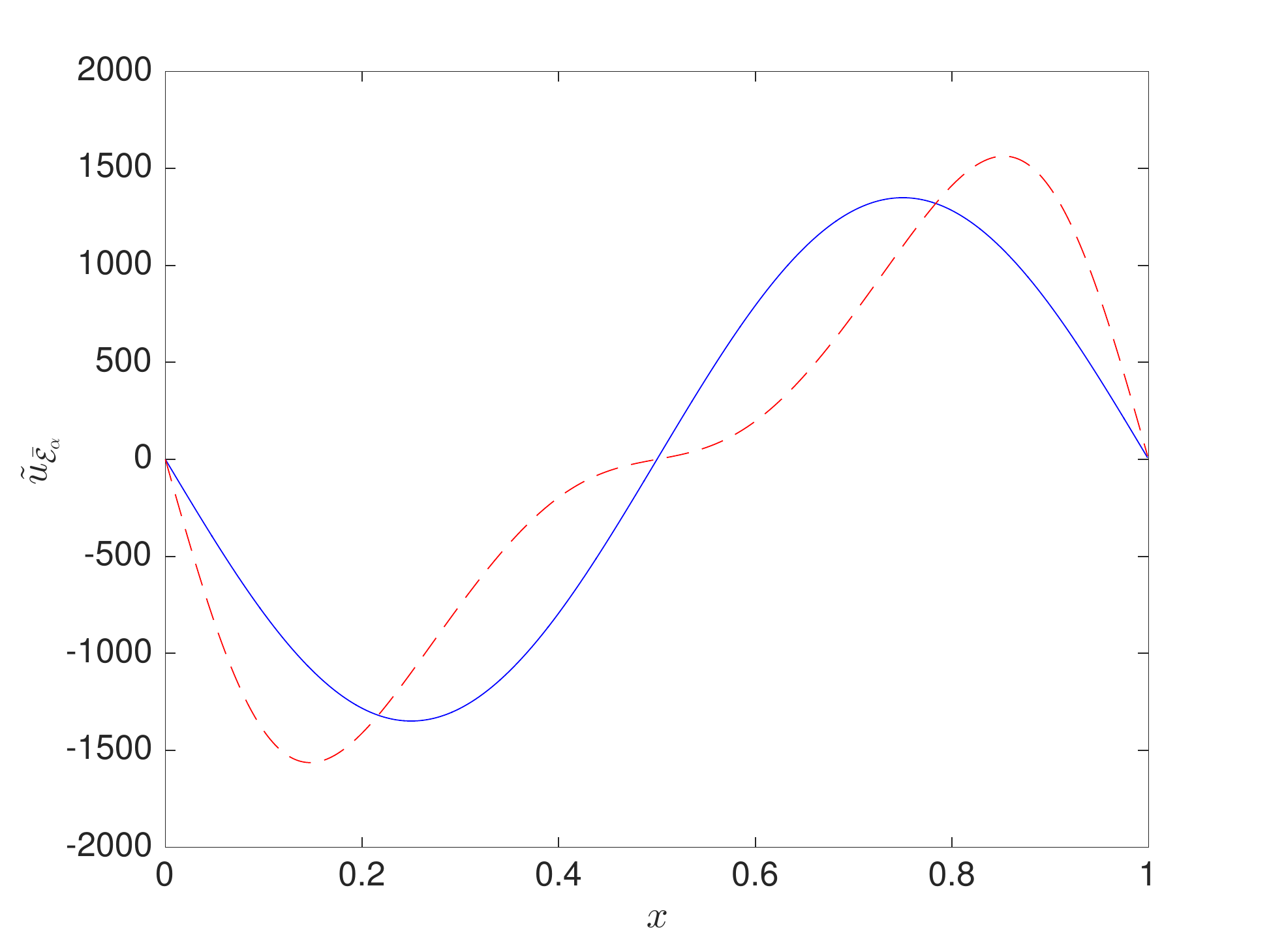}}
\caption{The maximizers {$\tuEabar$} characterized by positive
  $\R_{\Ea}({\tuEabar})$ (red dashed lines) and negative $\R_{\Ea}({\tuEabar})$
  (blue solid lines) obtained for (a) ${\bar{\E}_\alpha}=1$, $\alpha=0.01$, (b)
  ${\bar{\E}_\alpha}=335$, $\alpha=0.001$, and (c) ${\bar{\E}_\alpha}=455,000$, $\alpha=0.0001$.}
\label{fig:tuEa2}
\end{figure}

\begin{figure}
\centering
\subfigure[][]{\includegraphics[width=3.35in]{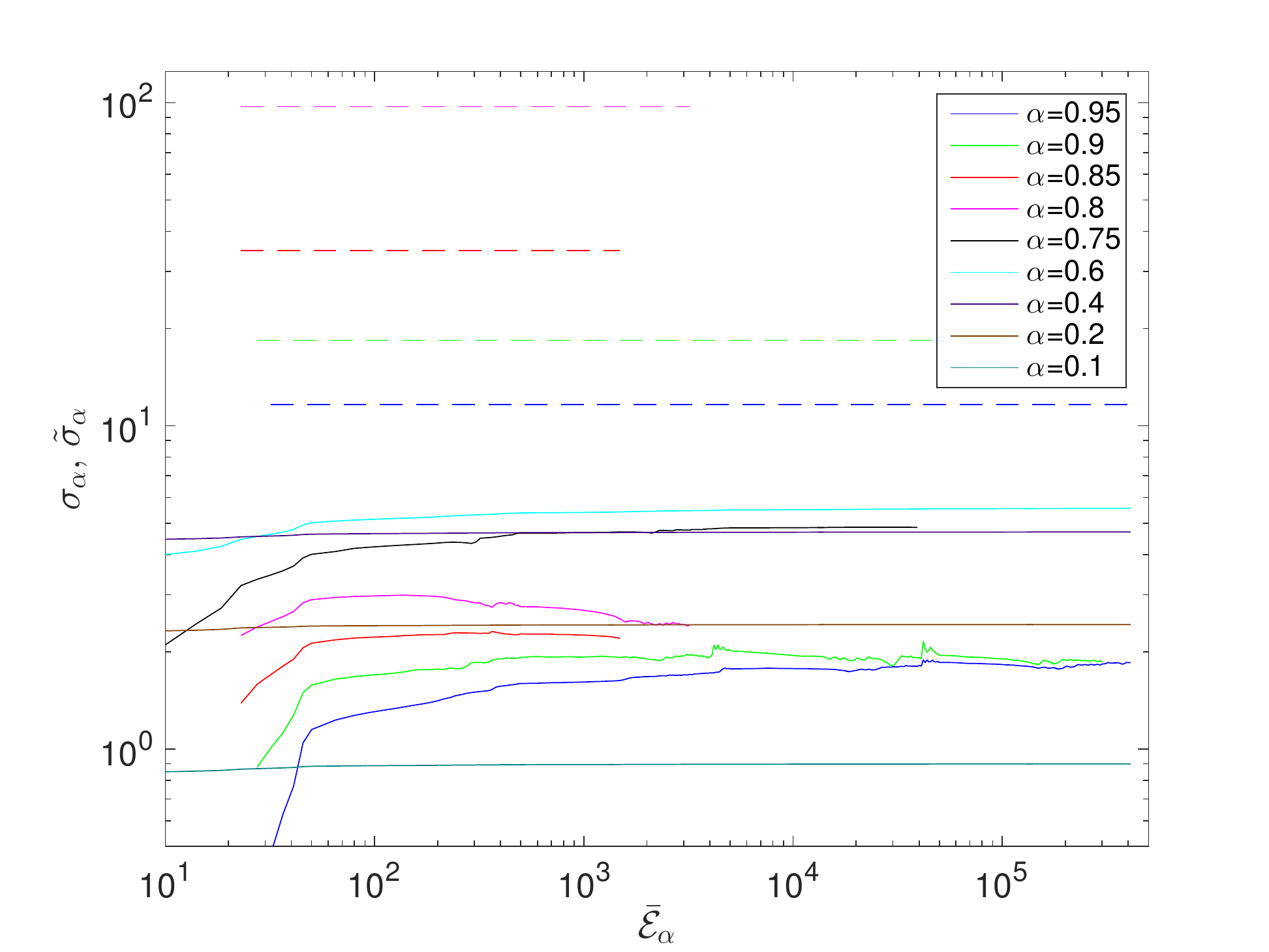}}
\hspace{-0.2in}
\subfigure[][]{\includegraphics[width=3.35in]{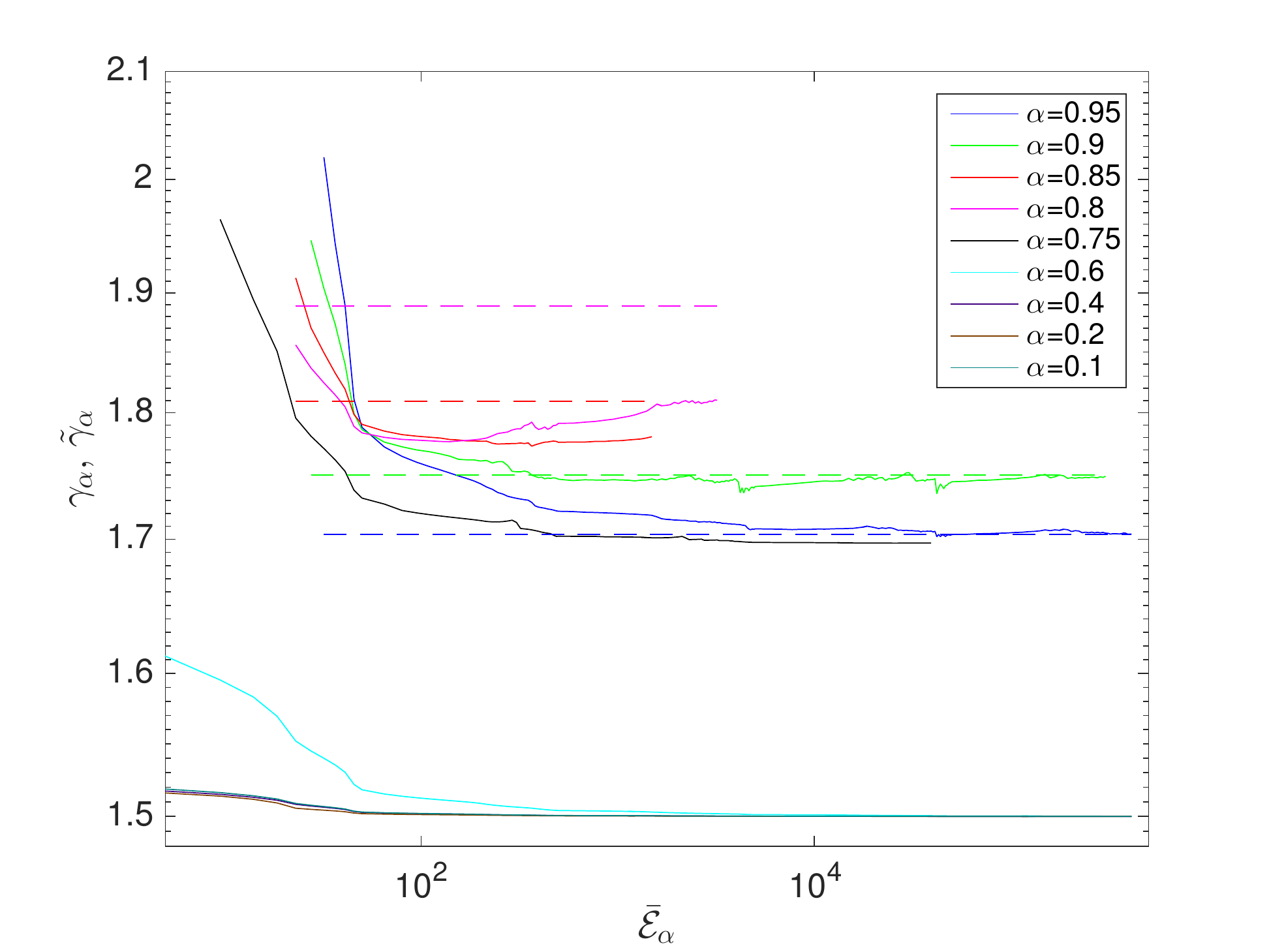}}
\caption{Prefactors $\tsigma_{\alpha}$ (a) and exponents
  $\tgamma_{\alpha}$ (b) obtained as function of {$\bar{\E}_\alpha$} via local
  least-squares fits to the relation $\R_{\Ea}({\tuEabar})$ versus {$\bar{\E}_\alpha$}
  shown Figure \ref{fig:Ra} (solid lines).  The dashed lines represent
  the corresponding prefactors $\sigma_{\alpha}$ and exponents
  $\gamma_{\alpha}$ from estimate \eqref{eq:dEadt}.}
\label{fig:fita}
\end{figure}

\begin{figure}
\centering
\subfigure[][]{\includegraphics[width=3.35in]{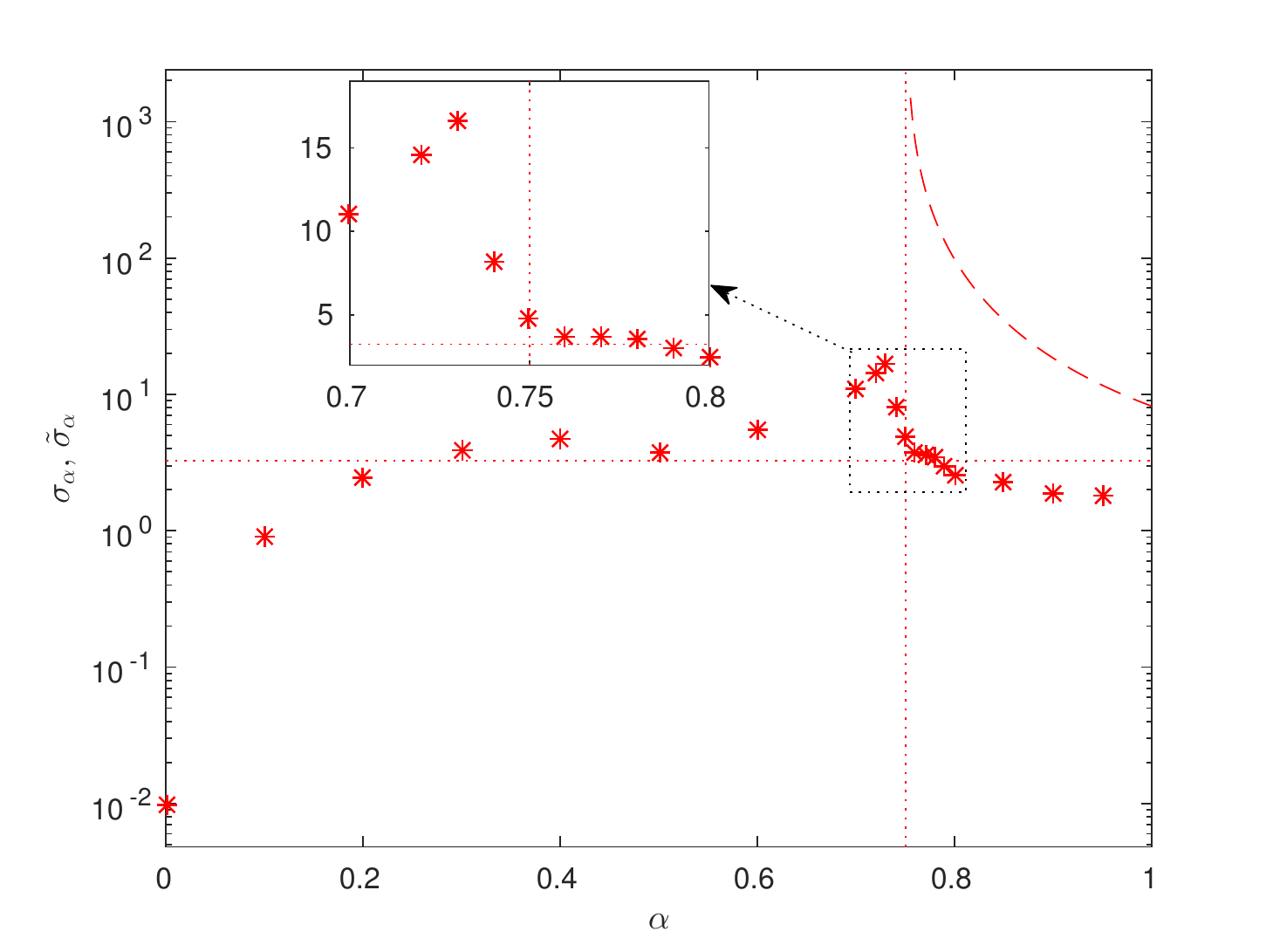}}
\hspace{-0.2in}
\subfigure[][]{\includegraphics[width=3.35in]{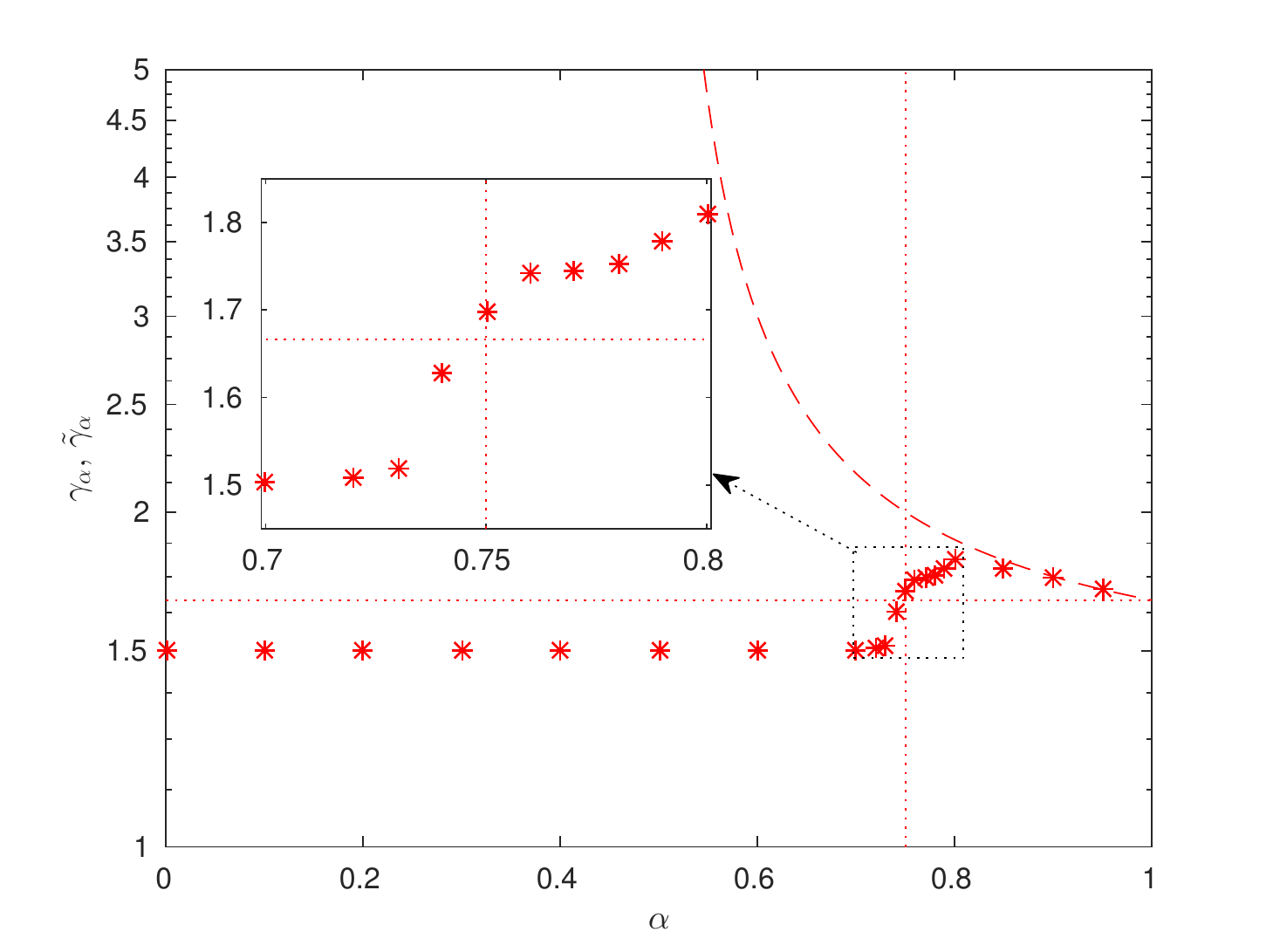}}
\caption{Prefactors (a) and exponents (b) in the power-law relation
 $\tsigma_{\alpha} {\bar{\E}_\alpha}^{\tgamma_{\alpha}}$ describing the
  dependence of $\R_{\Ea}({\tuEabar})$ on {$\bar{\E}_\alpha$} shown as functions
    of $\alpha$: limiting (as ${\bar{\E}_\alpha} \rightarrow \infty$,
  cf.~Figure \ref{fig:fita}) values obtained in the
  least-squares fits (symbols) and predictions of estimate
  \eqref{eq:dEadt} (dashed lines). The insets represent
    magnifications of the neighborhood of $\alpha = 3/4$ where estimate
    \eqref{eq:dEadt} loses its validity.  }
\label{fig:sigmagammaa}
\end{figure}

\section{Summary and Discussion}
\label{sec:final}

While the estimates on the rate of growth of the classical and
fractional enstrophy obtained in Theorems \ref{dEdt} and \ref{dEadt}
are not much different from similar results already available in the
literature \cite{kns08, ddl09}, the key finding of the present study
is that these estimates are in fact sharp, in the sense that for
different $\alpha$ the exponents $\gamma_1$ and $\gamma_{\alpha}$ in
\eqref{eq:dEdt} and \eqref{eq:dEadt} capture the correct power-law
dependence of the maximum growth rates $\R_{\E}({\tuEbar})$ and
$\R_{\Ea}({\tuEabar})$ on {$\bar{\E}$} and
{$\bar{\E}_\alpha$}, respectively (the second estimate was found
to be sharp only over a part of the range of $\alpha$ for which it is
defined). This was demonstrated by computationally solving suitably
defined constrained optimization problems and then showing that the
maximizers obtained under constraints on {$\E$} and
{$\E_\alpha$} saturate the upper bounds in the estimates for
different values of $\alpha$. Therefore, the conclusion is that the
mathematical analysis on which Theorems \ref{dEdt} and \ref{dEadt} are
based may not be fundamentally improved, other than a refinement of
exponent $\gamma_{\alpha}$ in \eqref{eq:dEadt} for $3/4 < \alpha
\lessapprox 0.9$ and an improvement of the prefactors in
\eqref{eq:dEdt} and \eqref{eq:dEadt}.  \medskip

In regard to the maximum rate of growth of the classical enstrophy, it
was found that for $\alpha \rightarrow (1/4)^+$ the exponent
$\gamma_1$ in estimate \eqref{eq:dEdt} becomes unbounded, cf.~Figure
\ref{fig:sigmagamma1}(b), which together with the computational evidence
obtained for $\alpha \in [0,1/4]$ suggests that $d\E/dt$ may be
unbounded for $\alpha$ in this range. This would indicate that for
$\alpha \in [0,1/4]$ system \eqref{eq:fburgers} is not even locally
well posed in $H^1(\I)$.
\medskip

Concerning the maximum rate of growth of the fractional enstrophy, a
surprising result was obtained for $\alpha \in [0,3/4)$, where the
exponent in the upper bound on $d\Ea/dt$ was found to be independent
of $\alpha$ (cf.~Figure \ref{fig:sigmagammaa}(b)). This indicates
that, unlike in the case of the classical enstrophy, in this range of
$\alpha$ the problem does not become more singular with the decrease
of $\alpha$ and this is in fact also reflected in the maximizers
{$\tuEabar$} becoming more regular as $\alpha \rightarrow 0$.
{In addition, this also suggests that it should be possible to
  obtain rigorous bounds on $d\Ea / dt$ valid for $\alpha \le 3 / 4$,
  although they would likely need to be derived using techniques other
  than those employed in the proof of Theorem \ref{dEadt}.}  \medskip

{It should be emphasized that although most of the individual
  inequalities used in the proofs of Theorems \ref{dEdt} and
  \ref{dEadt} are known to be sharp, the fact that the upper bounds in
  \eqref{eq:dEdt} and \eqref{eq:dEadt} were found to be sharp as well
  is not trivial. This is because, in general, these individual
  inequalities may be saturated by {\em different} fields which may
  belong to different function spaces and hence it is not obvious
  whether sharpness is preserved when these inequalities are
  ``chained'' together to form estimates \eqref{eq:dEdt} and
  \eqref{eq:dEadt}.  \medskip }

On the methodological side, it ought to be emphasized that
gradient-based iterations \eqref{eq:iter} may only identify {\em
  local} maximizers and in general it is not possible to ascertain
whether these maximizers are also global. However, our careful search
based on the continuation approach (cf.~Section
\ref{sec:continuation}) and, independently, using several different
initial guesses $u^0$ did not reveal any additional maximizers (other
than the maximizers obtained via a trivial rescaling of the solutions
as discussed in detail in \cite{ap11a}). An exception to this was the
solution of the maximization problem \eqref{eq:maxRa} for small
$\alpha$ and {$\bar{\E}_\alpha$} where a branch of maximizers
such that $\R_{\Ea}({\tuEabar}) < 0$ was also found. The presence
of this additional branch appears related to the degenerate nature of
the maximization problem \eqref{eq:maxRa} which for $\alpha = 0$ has
an uncountable infinity of trivial solutions (cf.~Section
\ref{sec:a0}).  \medskip

As regards the research program discussed in Introduction, the key
finding of the present study is that exponents $\gamma_1$ in the upper
bound on $d\E/dt$ have the same dependence on $\alpha$ and remain
sharp in the subcritical, critical and parts of the supercritical
regime. Thus, the loss of global well-posedness as $\alpha$ is reduced
to values below $1/2$ cannot be detected based on the instantaneous
rate of growth of enstrophy $d\E/dt$. {The most important open
  problem related to the present study concerns obtaining the
  corresponding estimates for the finite-time growth of $\E(u(t))$ and
  $\Ea(u(t))$ and verifying their sharpness. This question can be
  addressed using the approach developed in \cite{ap11a} and will be
  investigated in future research.}

\section*{Acknowledgments}

The authors wish to express sincere thanks to Dr.~Diego Ayala for his
help with software implementation of the approach described in Section
\ref{sec:method} and to Professor Koji Ohkitani for helpful
discussions. DY was partially supported through a Fields-Ontario
Post-Doctoral Fellowship and BP acknowledges the support through an
NSERC (Canada) Discovery Grant.



\bibliographystyle{unsrt}

\end{document}